\documentclass[a4paper,11pt,twoside]{amsart}
\usepackage[top=5.cm, bottom=5.0cm, left=3cm, right=3cm]{geometry}
\usepackage[UKenglish]{babel}
\usepackage{times}
\usepackage[pdftex]{hyperref}
\usepackage[foot]{amsaddr}
\usepackage{amsmath}
\usepackage{amsthm}	
\usepackage{amsfonts}
\usepackage{amssymb}
\usepackage{array}
\usepackage{paralist}
\usepackage{algorithmic,algorithm}
\usepackage{float}
\usepackage{graphicx}
\usepackage{adjustbox}
\usepackage{enumerate}
\usepackage{breakcites}
\usepackage[utf8]{inputenc}
\usepackage{dsfont}

\theoremstyle{plain}
\newtheorem{theorem}{Theorem}[section]

\newtheorem{assumption}[theorem]{Assumption}

\theoremstyle{definition}

\numberwithin{equation}{section}
\numberwithin{figure}{section}

\DeclareMathOperator*{\argmin}{arg\,min}

\title{Robust pricing and hedging via neural SDEs}
\author[P Gierjatowicz]{Patryk Gierjatowicz$^1$}
\email{s0837263@sms.ed.ac.uk}
\author[M. Sabate-Vidales]{Marc Sabate-Vidales$^1$}
\email{M.Sabate-Vidales@sms.ed.ac.uk}
\author[D. \v{S}i\v{s}ka]{David \v{S}i\v{s}ka$^{1,2}$}
\address{$^1$\href{https://www.maths.ed.ac.uk}{School of Mathematics, University of Edinburgh}}
\address{$^2$\href{https://vega.xyz}{Vega Protocol}}
\email{D.Siska@ed.ac.uk}
\author[L. Szpruch]{Lukasz Szpruch$^{1,3}$}
\address{$^3$\href{https://www.turing.ac.uk}{Alan Turing Institute}}
\email{L.Szpruch@ed.ac.uk}
\author[Z. \v{Z}uri\v{c}]{\v{Z}an \v{Z}uri\v{c}$^4$}
\address{$^4$\href{http://www.imperial.ac.uk}{Department of Mathematics, Imperial College London}}
\email{z.zuric19@imperial.ac.uk}

\date{\today}
\keywords{Stochastic differential equations, Deep neural network, Derivative pricing, Stochastic Gradient Descent, }

\begin{document}	%BEGINNING

\setcounter{tocdepth}{1}
%%%%%%%%%%%%%%%%%%%%%%%%%%%%%%%%%%%%%%%%%%%%%%%%%%%%%%%%%%%%%%%%%%%%%%%%%%%%%%%%%%%%%%%%%%%%%%%%%%%%%%%%%%%%%%%%%%%

%%%%%%%%%%%%%%%%%%%%%%%%%%%%%%%%%
\begin{abstract}
Mathematical modelling is ubiquitous in the financial industry and drives key decision processes. 
Any given model provides only a crude approximation to reality
and the risk of using an inadequate model is hard to detect and quantify.  
By contrast, modern data science techniques are opening the door to more robust and data-driven model selection mechanisms.  
However, most machine learning models are ``black-boxes'' as individual parameters do not have meaningful interpretation.  
The aim of this paper is to combine the above approaches achieving the best of both worlds.
Combining neural networks with risk models based on classical stochastic differential equations (SDEs), we find robust bounds for prices of derivatives and the corresponding hedging strategies
while incorporating relevant market data. 
The resulting model called neural SDE is an instantiation of generative models and is closely linked with the theory of causal optimal transport. 
Neural SDEs allow consistent calibration under both the risk-neutral and the real-world measures.
Thus the model can be used to simulate market scenarios needed for assessing risk profiles and hedging strategies. 
We develop and analyse novel algorithms needed for efficient use of neural SDEs.
We validate our approach with numerical experiments using both local and stochastic volatility models. 
\end{abstract}

\maketitle

%%%%%%%%%%%%%%%%%%%%%%%%%%%%%%%%%
{\bf 2010 AMS subject classifications:} 
Primary: 
65C30% Stochastic differential and integral equations
, 60H35%Computational methods for stochastic equations
; secondary: 
60H30%Applications of stochastic analysis (to PDE, etc.)
.\\
%%%%%%%%%%%%%%%%%%%%%%%%%%%%%%%%%%%%%%%%%%%%%%%%%%%%%%%%%%

%\tableofcontents
%%%%%%%%%

%\setcounter{tocdepth}{4}
%\tableofcontents

\section{Introduction}
%\section{Robust Pricing}

\subsection{Problem overview}

Model uncertainty is an essential part of mathematical modelling but is particularly acute in mathematical finance and economics where one cannot base models on well established physical laws. 
Until recently, these models were mostly conceived in a three step fashion: 
1) gathering statistical properties of the underlying time-series or the so called stylized facts;
2) handcrafting a parsimonious model, which would best capture the desired market characteristics without adding any needless complexity and 
3) calibration and validation of the handcrafted model.
Indeed, model complexity was undesirable, amongst other reasons, for increasing the computational effort required to perform in particular calibration but also pricing and risk calculations.
With greater uptake of machine learning methods and greater computational power more complex models can now be used. 
This is due to the fact that arguably the most complicated and computationally expensive step of calibration has been addressed. 
Indeed, in the seminal paper~\cite{hernandez2016} used neural networks to learn the calibration map from market data directly to model parameters. Subsequently, many papers followed \cite{liu2019neural,ruf2019neural,ruf2019neural,benth2020accuracy,gambara2020consistent,Sardroudi2019,Horvath2019,Bayer2019,Bayer2018,SabateSiskaSzpruch2018}. 
However, these approaches focused on the calibration of fixed parametric model, but did not address perhaps even the more important issue which is model selection and model uncertainty. 

The approach taken in this paper is fundamentally different.  
We let the data dictate the model, while still keeping a strong prior on the model form. 
This is achieved by using SDEs for the model dynamics but instead of choosing a fixed parametrization for the model SDEs we allow the drift and diffusion to be given by an overparametrized neural networks.
We will refer to these as Neural SDEs.   
These are shown to not only provide a systematic framework for model selection, but also, quite remarkably, to produce robust estimates on the derivative prices. 
Here, the calibration and model selection are done simultaneously. 
In this sense, model selection is data-driven. 
Since the neural SDE model is overparametrised, there is a large pool of possible models and the training algorithm selects a model. 
Unlike in handcrafted models, individual parameters do not carry any meaning. 
This makes it hard to argue why one model is better than another. 
Hence the ability to efficiently compute interval estimators, which algorithms in this paper provide, is critical.

In parallel to this work, a similar approach to modelling was taken in~\cite{cuchiero2020generative}, where the authors considered local stochastic volatility models with the leverage function approximated with a neural network. Their model can be seen as an example of a Neural SDEs.

Let us now consider a probability space $(\Omega, \mathcal F, (\mathcal F_t)_{t\in[0,T]},\mathbb P)$ and a random variable $\Psi \in L^2(\mathcal F_T)$ that represents the discounted payoff of a illiquid (path-dependent) derivative. 
The problem of calculating a market consistent price of a financial derivative can be seen as equivalent to finding a map that takes market data (e.g. prices of underlying assets, interest rates, prices of liquid options) and returns the no-arbitrage price of the derivative.
Typically an It\^{o} process $(X_t^\theta)_{{t}\in[0,T]}$, with parameters $\theta \in \mathbb R^p$ has been the main component used in constructing such pricing function. 
Such parametric model induces a martingale probability measure, denoted by $\mathbb Q(\theta)$, which is then used to compute no-arbitrage price of derivatives.
The market data  (input data) here is represented by payoffs $\{\Phi_i\}_{i=1}^M$ of liquid derivatives, and their corresponding market prices $\{\mathfrak p(\Phi_i)\}_{i=1}^{M}$.
We will assume throughout that this price set is free of arbitrage. 
To make the model $\mathbb Q(\theta)$ consistent with market prices, one seeks parameters $\theta^{*}$ such that 
the difference between $\mathfrak  p(\Phi_i)$ and $\mathbb E^{\mathbb Q(\theta^*)}[\Phi_i]$ is minimized for all $i = 1,\ldots,M$ (w.r.t. some metric).
If for all $i=1,\ldots,M$  we have $\mathfrak p(\Phi_i) = \mathbb E^{\mathbb Q(\theta^*)}[\Phi_i]$ then we will say the model is consistent with market data (perfectly calibrated).
%If the number of parameters $p$ is greater than the number of inputs then, 
There may be infinitely many models that are consistent with the market.
This is called Knightian uncertainty~\cite{knight1971risk,cohen2018data}. 

Let $\mathcal M$  be the set of all martingale measures / models that are perfectly calibrated to market inputs. 
In the robust finance paradigm, see~\cite{hobson1998robust,cox2011robust}, one takes conservative approach and instead of computing a single price (that corresponds to a model from $\mathcal M$) one computes the price interval $(\inf_{\mathbb Q \in \mathcal M}\mathbb E^{\mathbb Q }[\Psi],\sup_{\mathbb Q \in \mathcal M}\mathbb E^{\mathbb Q }[\Psi])$. The bounds can be computed using tools from martingale optimal transport which also, through dual representation, yields corresponding super- and sub- hedging strategies, \cite{beiglbock2013model}. Without imposing further constrains, the class of all calibrated models $\mathcal M$ might be too large and consequently the corresponding bounds too wide to be of practical use \cite{eckstein2019robust}. See however an effort to incorporate further market information to tighten the pricing interval, \cite{nadtochiy2017robust,aksamit2020robust}. Another shortcoming of working with  the entire class of calibrated models $\mathcal M$ is that, in general, it is not clear how to obtain a practical/explicit model out of the measures that yields price bounds. For example, such explicit models are useful when one wants consistently calibrate under pricing measure $\mathbb Q$ and real-world measure $\mathbb P$ as needed for risk estimation and stress testing, \cite{broadie2011efficient,pelsser2016difference} or learn hedging strategies in the presence of transactional cost and an illiquidity constrains \cite{buehler2019deep}.

\subsection{Neural SDEs}
\label{sec nsdes}

Fix $T>0$ and for simplicity assume constant interest rate $r\in \mathbb R$. 
Consider parameter space $\Theta = \Theta^b\times \Theta^\sigma \subseteq \mathbb R^{p}$ and parametric functions  $b:\mathbb R^d \times \Theta^b \rightarrow \mathbb R^{d}$ and $\sigma:\mathbb R^d \times \Theta^\sigma \rightarrow \mathbb R^{d \times n}$. Let $(W_t)_{t\in[0,T]}$ be a $n$-dimensional Brownian motion supported on $(\Omega, \mathcal F, (\mathcal F_t)_{t\in[0,T]},\mathbb Q)$ so that $\mathbb Q$ is the Wiener measure and $\Omega = C([0,T];\mathbb R^n)$. 
We consider the following parametric SDE
\begin{equation}\label{eq:nsde}
dX^{\theta}_t= b(t,X_t^{\theta}, \theta) dt + \sigma(t,X_t^\theta,\theta )dW_t\,.
\end{equation}
We split $X^\theta$ which is the entire stochastic model into traded assets and non-tradable components. 
Let $X^\theta=(S^\theta,V^\theta)$, where $S$ are the traded assets and $V$ are the components that are not traded. 
We will assume that for all $t\in[0,T]$, $x = (s, v)\in \mathbb R^d$ and $\theta \in \mathbb R^p$ we will assume that 
\[
b(t, (s,v), \theta) = \big(rs, b^V(t,(s,v),\theta)\big) \in \mathbb R^d\,\,\,\text{and}\,\,\,
\sigma(t,(s,v),\theta) = \big(\sigma^S(t,(s,v),\theta), \sigma^V(t,(s,v),\theta)\big)\,.
\]
Then we can write~\eqref{eq:nsde} as
\begin{equation}\label{eq:nsde2}
\begin{split}
dS^{\theta}_t & = r S_t^{\theta} \, dt + \sigma^S(t,X_t^\theta,\theta )\,dW_t\,,\\
dV^\theta_t & = b^V(t,X_t^{\theta}, \theta) \, dt + \sigma^V(t,X_t^\theta,\theta )\,dW_t\,,\\
X^\theta_t & = (S^\theta_t, V^\theta_t)	\,.
\end{split}
\end{equation}
Observe that $\sigma^S$ and $\sigma^V$ encode arbitrary correlation structures between the traded assets and the non-tradable components.
Moreover, we immediately see that $(e^{-rt}S_t)_{t\in[0,T]}$ is a (local) martingale and thus the model is free of arbitrage.

In a situation when $(b,\sigma)$ are defined to be neural networks (see Appendix~\ref{FFNNs}), we call the SDE~\eqref{eq:nsde} a neural SDE and we denote by $\mathcal M^{\text{nsde}}(\theta)$ the class of all solutions to \eqref{eq:nsde}. 
Note that due to universal approximation property of neural networks, see~\cite{Hor91,sontag1997complete,cuchiero2019deep}, $\mathcal M^{\text{nsde}}(\theta)$ contains large class of SDEs solutions.
Furthermore, neural networks can be efficiently trained with the stochastic gradient decent methods and hence one can easily seek calibrated models in $\mathcal M^{\text{nsde}}(\theta)$. 
Finally, neural SDE integrate black-box neural network type models with the known and well studied SDE models. 
One consequence of that is that one can: a) consistently calibrate these under the risk neutral measure as well as the real-world measure; b) easily integrate additional market information e.g constrains on realised variance; c) verify martingale property.  We want to remark that for simplicity we work in Markovian setting, but one could consider neural-SDEs with path-dependent coefficients and/or consider more general noise processes. 
We postpone analysis of theses cases to follow up paper. 
By imposing suitable conditions on the coefficients $(b,\sigma)$ we know that unique solution to~\eqref{eq:nsde} exists, \cite[Chapter 2]{MR601776}. 
These conditions can be satisfied by neural networks e.g. by applying weight clipping.
We denote the law of $X^{\theta}$ on $C([0,T];\mathbb R^d)$ by $\mathbb Q(\theta) := \mathcal L((X_t)_{t\in[0,T]})$.

Given a loss function $\ell :\mathbb R \times \mathbb R \rightarrow \mathbb R^{+}$, the search for calibrated model can be written as 
\[
\theta^{\ast} \in \arg\min_{\theta \in\Theta}\sum_{i=1}^M \ell(\mathbb E^{\mathbb Q(\theta)}[\Phi_i],\mathfrak  p(\Phi_i) )\,, \quad \text{where} \quad \mathbb E^{\mathbb Q(\theta)}[\Phi]=\int_{ C([0,T],\mathbb R^d)}\Phi(\omega) \mathcal L(X^{\theta})(d\omega)\,.
\] 

To extend the calibration consistently to the real world measure, assume that we are given some statistical facts (e.g. moments or other distributional properties) that the price process (or the non tradable components) should satisfy).
Let  $\zeta:[0,T]\times \mathbb R^d \times \mathbb R^p \to \mathbb R^n$ be another parametric function (e.g. neural network) and we extend the parameter space to 
$\Theta = \Theta^b\times \Theta^\sigma \times \Theta^\zeta \subseteq \mathbb R^{p}$.
Let 
\[
\begin{split}
b^{S,\mathbb P}(t,X^\theta_t,\theta) & := r S^\theta_t + \sigma^S(t,X^\theta_t, \theta) \zeta(t,X^\theta_t,\theta)\,,\\
b^{V,\mathbb P}(t,X^\theta_t,\theta) & := b^V(t,X^\theta_t,\theta) + \sigma^V(t,X^\theta_t, \theta) \zeta(t,X^\theta_t,\theta)\,.
\end{split}
\]
We now define a real-world measure $\mathbb P(\theta)$ via the Radon--Nikodym derivative 
\[
\frac{d\mathbb P(\theta)}{d\mathbb Q(\theta)} := \exp\left(\int_0^T \zeta(t,X^\theta_t,\theta) \,dW_t + \frac12 \int_0^T |\zeta(t,X^\theta_t,\theta)|^2\,dt \right)\,.
\]
Under appropriate assumption on $\zeta$ (e.g. bounded) the measure $\mathbb P(\theta)$ is a probability measure and by using Girsanov theorem we can find Brownian motion $(W^{\mathbb P(\theta)}_t)_{t\in[0,T]}$ such that 
\begin{equation}\label{eq:rvsde}
\begin{split}
dS^{\theta}_t & = b^{S,\mathbb P}(t, X^\theta_t, \theta) \, dt + \sigma^S(t,X_t^\theta,\theta )\,dW^{\mathbb P(\theta)}_t\,,\\
dV^\theta_t & = b^{V, \mathbb P}(t,X_t^{\theta}, \theta) \, dt + \sigma^V(t,X_t^\theta,\theta )\,dW^{\mathbb P(\theta)}_t\,.
\end{split}
\end{equation}	
This is now the Neural SDE model in real-world measure $\mathbb P(\theta)$ and one would like use market data to seek $\zeta$. Let $\mathbb P^{\text{market}}$ denote empirical distribution of market data and $(\mathbb E^{\mathbb  P^{\text{market}}}[\mathcal S_i])_{i=1}^{\tilde M}$ be a corresponding set statistics one aims to match. These might be autocorrelation function, realised variance or moments generating functions.  
%Just like with derivative payoffs we can write the real-world statistics as $\tilde \Phi_i$, with the desired values as $\mathfrak p(\tilde \Phi_i)$, $i=1,\ldots,\tilde M$. 
%For example we may have $\tilde \Phi_1 = \mathbb E^{\mathbb P(\theta)}[|X_T|^2]$.
The calibration to real-world measure, with $(b^V,\sigma^V,\sigma^{S})$ being fixed, consists of finding $\theta^{\ast}$ such that
%\[
%\theta^{\ast} \in \arg\min_{\theta \in\Theta}\bigg(\sum_{i=1}^M \ell(\mathbb E^{\mathbb Q(\theta)}[\Phi_i],\mathfrak  p(\Phi_i) ) + \sum_{i=1}^{\tilde M} \ell(\mathbb E^{\mathbb P(\theta)}[\tilde \Phi_i],\mathfrak  p(\tilde \Phi_i) ) \bigg)\,.
%\] 
\[
\theta^{\ast} \in \arg\min_{\theta \in\Theta}\sum_{i=1}^{\tilde M} \ell(\mathbb E^{\mathbb P(\theta)}[\mathcal S_i],\mathbb E^{\mathbb P^{\text{market}}}[\mathcal S_i(\omega)]) \,.
\] 
But in fact we can write
\[
\mathbb E^{\mathbb P(\theta)}[\mathcal S_i] = \mathbb E^{\mathbb Q(\theta)}\bigg[\mathcal S_i \frac{d\mathbb P(\theta)}{d\mathbb Q(\theta)}\bigg]\,.
\]
Thus we see that in this framework there needs to be no distinction between a derivative price $\Phi_i$ and a real-world statistic $\mathbb E^{\mathbb P^{\text{market}}}[\mathcal S_i]$. 
% since we can simply think of this as the derivative with risk-neutral payoff $\tilde \Phi_i \frac{d\mathbb P(\theta)}{d\mathbb Q(\theta)}$.
Hence from now on we will write only about risk-neutral calibrations bearing in mind that methodologically this leads to no loss of generality. 

Let us connect neural SDEs to the concept of generative modelling, see~ \cite{goodfellow2014generative,kingma2013auto}. 
Let $\mathbb Q^{\text{market}} \in \mathcal M$ be the true martingale measure (so by definition all liquid derivatives are perfectly calibrated  under this measure i.e. $\mathbb E^{\mathbb Q^{\text{market}}}[\Phi_i] = \mathfrak p(\Phi_i)$ for all $i=1,\ldots,M$).
%Indeed, in the special case of calibrating to vanilla derivatives we are given prices of call and put options,~$(p(k,t))_{0\leq k \leq K, 0\leq t \leq T}$, for all strikes and maturities. 
%Then Breeden-Litzenberger formula would give us marginal laws of the pricing measure for all maturities, \cite{breeden1978prices}. 
%We denote this measure by $\mathbb Q^{\text{market}}$. We do not know this measure but for each $t^j$ we have access to its empirical approximation $\mathbb Q^{\text{market}}_{t^j}=\frac{1}{M}\sum_{i=1}^M\delta_{\{k^i,t^j\}}$ through payoff of vanilla options. We stress that with methods developed herein we can calibrate neural SDEs to path-dependent derivatives as well. 
We know that when~\eqref{eq:nsde} admits a strong solution then for any $\theta \in \mathbb R^p$ there exists a measurable map $G^\theta:\mathbb R^d \times C([0,T]; \mathbb R^{n}) \to  C([0,T];\mathbb R^d)$ such that
$X^{\theta} = G^\theta(\zeta, W)$, see~\cite[Corolarry 3.23]{karatzas2012brownian}. 
Hence, one can view \eqref{eq:nsde} as a generative model that maps $\mu$, the joint distribution of $X_0$ on $\mathbb R^d$ and the Wiener measure on $C([0,T];\mathbb R^n)$ into $\mathbb Q^{\theta} = (G^\theta_t)_{\#}\mu$. 
%Let $G^\theta_t$ be a time marginal projection of $G^\theta$ such that $G^\theta_t(\zeta, (W_{s \wedge t})_{s\in[0,T]}) = X^{\theta}_t$. 
We see that  by construction $G$ is a causal transport map i.e transport map that is adapted to filtration $(\mathcal F_t)_{t\in [0,T]}$, see also~\cite{acciaio2019causal,lassalle2013causal}.

One then seeks $\theta^{\ast}$ such that $G^{\theta^{\ast}}_{\#}\mu$ is a good approximation of $\mathbb Q^{\text{market}}$ with respect to user specified metric. In this paper we work with 
\[
D(G^\theta_{\#}\mu, \mathbb Q^{\text{market}}):= \sum_{i=1}^M\ell\left(\int_{C([0,T],\mathbb R^d)} \Phi_{i}(\omega)(G^\theta_{\#}\mu)(d\omega), \int_{C([0,T],\mathbb R^d)} \Phi_i(\omega) \mathbb Q^{\text{market}}(d\omega) \right)\,.
\] 
As we shall see in Sections~\ref{sec:LVintro} and~\ref{sec LSV} there are many Neural SDE models that can be calibrated well to market data and that produce significantly different prices for derivatives that were not part of the calibration data.
In practice these would be illiquid derivatives where we require model to obtain prices.
Therefore, we compute price intervals for illiquid derivatives within the class of calibrated neural SDEs models.  To be more precise we compute 
\[
\inf_{\theta} \left\{ \mathbb E^{ \mathbb Q(\theta)}[\Psi] \, :\, D(G(\theta)_{\#}\mu_0,\mathbb Q^{\text{market}})=0 \right\} \,,\,\,\,\,
\sup_{\theta} \left\{ \mathbb E^{\mathbb Q(\theta)}[\Psi] \, :\, D(G(\theta)_{\#}\mu_0,\mathbb Q^{\text{market}})=0 \right\}\,.
\]
We solve the above constraint optimisation problem by penalisation. See  \cite{eckstein2019computation} for related ideas.  

\subsection{Key conclusions and methodological contributions of this paper}
The results in this paper presented below lead to the following conclusions. 
\begin{enumerate}[i)]
\item Neural SDEs provide a systematic framework for model selection and produce robust estimates on the derivative prices. 
The calibration and model selection are done simultaneously and the thus the model selection is data-driven. 
\item With neural SDEs, the modelling choices one makes are: networks architectures, structure of neural SDE (e.g. traded and non-traded assets), training methods and data. 
For classical handcrafted models the choice of the algorithm for calibrating parameters has not been considered as part of modelling choice, but for machine learning this is one of the key components. 
See Section~\ref{sec numerics}, where we show how the change in initialisation of stochastic gradient method used for training leads to different prices of illiquid options, thus providing one way of obtaining price bounds.
Furthermore even for basic local volalitly model that is unique for continuum of strikes and maturities, produces ranges of prices of illiquid derivatives when calibrated to finite data sets.  

\item The above optimisation problem is not convex. Nonetheless, empirical experiments in Sections \ref{sec:LVintro}-\ref{sec LSV} demonstrate that the stochastic gradient decent methods used to minimise the loss functional $D$ converges to the set of parameters for which calibrated error is of order $10^{-5}$ to $10^{-4}$ for the square loss function. 
Theoretical framework for analysing such algorithms is being developed in~\cite{siska2020gradient}.

\item By augmenting classical risk models with modern machine learning approaches we are able to benefit from expressibility of neural networks while staying within realm of classical models, well understood by traders, risk managers and regulators. 
This mitigates, to some extent, the concerns that regulators have around use of black-box solutions to manage financial risk. 
Finally while our focus here is on SDE type models, the devised framework naturally extends to time-series type models.  
\end{enumerate}

The main methodological contributions of this work are as follows. 
\begin{enumerate}[i)]
\item By leveraging martingale representation theorem we developed an efficient Monte Carlo based methods that simultaneously learns the model and the corresponding hedging strategy.
\item The calibration problem does not fit into classical framework of stochastic gradient algorithms, as the  mini-batches of the gradient of the cost function are biased. 
We provide analysis of the bias and show how the inclusion of hedging strategies in training mitigates this bias.  
\item We devise a novel, memory efficient randomised training procedure. 
The algorithm allows us to keep memory requirements constant, independently of the number of neural networks in neural SDEs.
This is critical to efficiently calibrate to path dependent contingent claims. 
We provide theoretical analysis of our method in Section~\ref{sec random training} and numerical experiment supporting the claims in Section~\ref{sec numerics}. 
 \end{enumerate}

The paper is organized as follows. 
In Section~\ref{sec:Hedging} we outline the exact optimization problem, introduce a deep neural network control variate (or hedging strategy), address the process of calibration to single/multiple option maturities and state the exact algorithms.
In Section~\ref{sec:sgd} we analyse the bias in Algorithms~\ref{alg LSV calibration vanilla} and \ref{alg LSV calibration vanilla lb exotic}.
In Section~\ref{sec random training} we show that the novel, memory-efficient, drop-out-like training procedure for path-dependent derivatives does not introduce bias in the new estimator. 
Finally, the performance of Neural Local Volatility and Local Stochastic Volatility models is presented in Section~\ref{sec numerics}. 
Some of the proofs and more detailed results from numerical experiments are relegated to the Appendix.
The code used is available at \href{https://github.com/msabvid/robust_nsde}{\texttt{github.com/msabvid/robust\_nsde}}.

\section{Robust pricing and hedging}\label{sec:Hedging}

%\subsection{Feed-forward neural networks}\label{sec:NN}s

% For $x\in \mathbf{R}^{N_0}$, $A_1\in \mathbf{R^{N_1\times N_0}}$, $A_i\in \mathbf{R^{N_i\times N_{i-1}}}$, $b_i,c_i \in\mathbf{R^{N_i}},$  $i\leq L$,  we define $L-$layer neural network (NN)  as:
%\begin{eqnarray*}
%z_1&=&\sigma(xA_1^T+b_1)A_2^T+c_1,\\
%z_j&=&\sigma(z_{j-1}+b_j)A_{j+1}^T+c_j,~1\le j\leq L ,
%\end{eqnarray*}
%where $\sigma(z)=max(0,z)$ represents ReLU activation function acting element-wise on vector $z$.

Let $\ell : \mathbb R \times \mathbb R \to [0,\infty)$ be a convex loss function such that $\min_{x\in\mathbb R,y \in\mathbb R} \ell(x,y) = 0$. 
For example we can take $\ell(x,y) = |x-y|^2$. 
Given $\ell$, our aim is to solve the following optimisation problems:
\begin{enumerate}[i)]
\item Find model parameters $\theta^*$ such that model prices match market prices:
%\begin{equation}\label{eq LSV loss}
%\theta^{\ast}=\arg\min_{\Theta}\sum_{i=0}^M \ell(\mathbb E^{\mathbb Q(\theta)}[V_i^N], C_i ),~\sum_{i=1}^M \ell(\mathbb E^{\mathbb Q(\theta^\ast)}[V_i^N], C_i )=0~,
%\end{equation}
\begin{equation}\label{eq LSV loss}
\theta^{\ast} \in \arg\min_{\theta \in \Theta}\sum_{i=1}^M \ell(\mathbb E^{\mathbb Q(\theta)}[\Phi_i],\mathfrak  p(\Phi_i) )\,.
\end{equation}
%Since the Neural SDE models we are using are overparametrized and since $\ell \geq 0$ with minimum of $0$ in practice this is equivalent to finding some $\theta^*$ such that $\sum_{i=1}^M \ell(\mathbb E^{\mathbb Q(\theta)}[\Phi_i],\mathfrak  p(\Phi_i) ) = 0$.
In practice this is equivalent to finding some $\theta^*$ such that $\sum_{i=1}^M \ell(\mathbb E^{\mathbb Q(\theta^*)}[\Phi_i],\mathfrak  p(\Phi_i) ) = 0$. This is due to inherent overparametrization of Neural SDEs and the fact that $\ell\geq 0$ reaches minimum at zero.
\item Find model parameters $\theta^{l,*}$ and $\theta^{u,*}$ which provide robust arbitrage-free price bounds for an illiquid derivative, subject to available market data: 
\begin{equation}\label{eq LSV loss lb}
\begin{split}
	 \theta^{l,\ast} & \in  \arg \min_{\theta \in \Theta} \mathbb E^{\mathbb Q(\theta)}[\Psi]\,,\,\,\,~\text{ subject to}\,\,\, \sum_{i=1}^M \ell(\mathbb E^{\mathbb Q(\theta)}[\Phi_i], \mathfrak   p(\Phi_i)) = 0\,\,, \\
	  \theta^{u,\ast} & \in  \arg \max_{\theta \in \Theta} \mathbb E^{\mathbb Q(\theta)}[\Psi],\,\,\,~\text{ subject to}\,\,\, \sum_{i=1}^M \ell(\mathbb E^{\mathbb Q(\theta)}[\Phi_i], \mathfrak  p(\Phi_i)) = 0\,. 
\end{split}	
\end{equation}
\end{enumerate}
The no-arbitrage price of $\Psi$ over the class of neural SDEs used is then in $\Big[\mathbb E^{\mathbb Q(\theta^{l,*})}, \mathbb E^{\mathbb Q(\theta^{u,*})}\Big]$.

\subsection{Learning hedging strategy as a control variate}

A starting point in the derivation of the practical algorithm is to estimate $\mathbb E^{\mathbb Q(\theta)}[\Phi]$ using a Monte Carlo estimator. 
Consider $(X^{i,\theta})_{i=1}^N$, a $N$ i.i.d copies of \eqref{eq:nsde} and    let $\mathbb Q^N(\theta):=\frac{1}{N}\sum_{i=1}^N \delta_{X^{i,\theta}}$ be empirical approximation of $\mathbb Q(\theta)$.   
Due to the Law of Large Numbers, $\mathbb E^{\mathbb Q^N(\theta)}[\Phi]$ converges to $\mathbb E^{\mathbb Q(\theta)}[\Phi]$ in probability. Moreover, the Central Limit Theorem
tells us that
\[
\mathbb P \left(\mathbb E^{\mathbb Q(\theta)}[\Phi] \in \left[\mathbb E^{\mathbb Q^N(\theta)}[\Phi] - z_{\alpha/2}\frac{\sigma}{\sqrt{N}},
\mathbb E^{\mathbb Q^N(\theta)}[\Phi] + z_{\alpha/2}\frac{\sigma}{\sqrt{N}}\right] \right) \rightarrow 1 \,\,\, \text{as}\,\,\, N\rightarrow \infty~,
\]
where $\sigma = \sqrt{\mathbb V\text{ar}[\Phi]}$ and $z_{\alpha/2}$ is such that
 $1-\text{CDF}_Z(z_{\alpha/2})=\alpha/2$ with $Z$ the standard normal distribution. 
We see that by increasing $N$, we reduce the width of the above confidence intervals, but this increases the overall computational 
cost. 
A better strategy is to find a good control variate i.e. we seek a random variable ${\Phi^{cv}}$
such that:
\begin{equation}\label{eq LSV cv properties}
\mathbb E^{\mathbb Q^N(\theta)}[\Phi^{cv}] = \mathbb E[\Phi] \quad \text{and} \quad \mathbb V\text{ar}[\Phi^{cv}] < \mathbb V\text{ar}[\Phi]~.
\end{equation}
In the following we construct $\Phi^{cv}$ using hedging strategy. Similar approach has recently been developed in~\cite{SabateSiskaSzpruch2018} in the context of pricing and hedging with deep networks.  

Martingale representation theorem (see for example Th. 14.5.1 in ~\cite{cohen2015stochastic}) provides
a general methodology for finding Monte Carlo estimators with the above stated properties~\eqref{eq LSV cv properties}. 
%Let $\bar S^{\theta}_t:= e^{-rt}S^{\theta}_t$ and note that due to~\eqref{eq:nsde2} we have 
%\[
%d(\bar S^{\theta}_t) = e^{-rt}\sigma^S(t,S_t^{\theta},\theta)\,dW_t^{\mathbb Q(\theta)}\,.
%\]
%In particular the process $(\bar S^{\theta}_t)_{t}$ is a (local) $\mathbb Q(\theta)$-martingale. 

If $\Phi$ is such that $\mathbb E^{\mathbb Q}[|\Phi|^2]<\infty$, then there exists a unique process $Z=(Z_t)_t$
adapted to the filtration $(\mathcal F_t)_{t\in[0,T]}$ 
with $\mathbb E^{\mathbb Q}\left[ \int_t^T |Z_s|^2\,ds \right]<\infty$ such that 
%\[
%\Phi = \mathbb E[\Phi | \mathcal F_0^W] + \int_0^T Z_s dW^{\mathbb Q(\Theta)}_s,~\text{hence:}
%\] 
\[
\mathbb E[\Phi | \mathcal F_0] = \Phi - \int_0^T Z_s \, dW_s\,.
\]
Define 
\[
\Phi^{cv}:= \Phi - \int_0^T Z_s\, dW_s\,,
\]
and note that
\[
\mathbb E^{\mathbb Q(\theta)}[ \Phi^{cv} | \mathcal F_0] = \mathbb E^{\mathbb Q(\theta)}[\Phi | \mathcal F_0] \quad \text{and} \quad \mathbb V\text{ar}^{\mathbb Q(\theta)}[\Phi^{cv} | \mathcal F_0 ]  = 0~.
\]
The process $Z$ has more explicit representations using corresponding (possibly path dependent) backward Kolomogorov equation or Bismut--Elworthy--Li formula.
Both approaches require further approximation, see~\cite{SabateSiskaSzpruch2018} and~\cite{SabateSiskaSzpruch2020}.
Here, this approximation will be provided by an additional neural network. 
Without loss of generality assume that $\Phi=\phi((X^{\theta}_{t})_{t\in[0,T]})$
for some  $\phi:C([0,T],\mathbb R^d) \rightarrow \mathbb R$. In the remaining of the paper, we will slightly abuse the notation to write $\Phi$ indistinctively
as both the option and the mapping $C([0,T],\mathbb R^d) \rightarrow \mathbb R$.

Consider now a neural network  $\mathfrak h : [0,T] \times C([0,T],\mathbb R^d) \times \mathbb R^p  \rightarrow \mathbb R^d $ with parameters $\xi \in \mathbb R^{p'}$ with $p'\in \mathbb N$  and define the following learning task, in which $\theta$ (the parameters on the Neural SDE model) is fixed: 

Find  
\begin{equation}\label{eq loss cv}
\xi^{\ast} \in \argmin_{\xi} \mathbb V\text{ar}\left[\Phi((X^{\theta}_t)_{t\in[0,T]}) -  \int_0^T \mathfrak h(s,(X^{\theta}_{s\wedge t})_{t\in[0,T]},\xi) dW_s \bigg | \mathcal F_0\right]\,.
\end{equation}
%Note that uniqueness of the process $Z$ in martingale representation theorem implies that if we can find $\xi^*$ such that 
%\[
%\mathbb V\text{ar}\left[\phi((X^{\theta}_t)_{t\in[0,T]}) -  \int_0^T \mathfrak h(s,(X^{\theta}_{s\wedge t})_{t\in[0,T]},\xi^*) d\bar X^{\theta}_s \bigg | \mathcal F_0 \right] = 0\,,
%\]  
%then $\mathfrak  h(s,(X_{s\wedge t})_{t\in[0,T]},\xi^*)  = Z_s $  a.s $ s\in[0,T]$. 
In a similar manner one can derive $\Psi^{cv}$ for the payoff of the  illiquid derivative for which we seek 
the robust price bounds.  Then \eqref{eq LSV loss lb} can be restated as 
\begin{equation}\label{eq LSV loss lb cv}
\begin{split}
	 \theta^{l,\ast} & \in \arg \min_{\theta \in \Theta} \mathbb E^{\mathbb Q(\theta)}[\Psi^{cv}]\,,\,\,\,~\text{ subject to}\,\,\, \sum_{i=1}^M \ell(\mathbb E^{\mathbb Q(\theta)}[\Phi_i^{cv}], \mathfrak   p(\Phi_i)) = 0\,, \\
	  \theta^{u,\ast} & \in \arg \max_{\theta \in \Theta} \mathbb E^{\mathbb Q(\theta)}[\Psi^{cv}]\,,\,\,\,~\text{ subject to}\,\,\, \sum_{i=1}^M \ell(\mathbb E^{\mathbb Q(\theta)}[\Phi_i^{cv}], \mathfrak  p(\Phi_i)) = 0\,. 
\end{split}	
\end{equation}
The learning problem~\eqref{eq LSV loss lb cv} is better  than~\eqref{eq LSV loss lb} from the point of view of algorithmic implementation, as it will enjoy lower Monte Carlo variance and hence
will require simulation of fewer paths of the Neural SDE in each step of the stochastic gradient algorithm. 
Furthermore, when using~\eqref{eq LSV loss lb cv} we learn a (possibly abstract)
hedging strategy for trading in the underlying asset 
to replicate the derivative payoff. 
Since the market may be incomplete this abstract hedging strategy may not be usable in practice. 
More precisely, since the process $X^\theta$ will contain tradable as well as non-tradable assets, the control variate for the latter has to be adapted by either performing a projection or deriving a strategy for the corresponding tradable instrument.

To deduce a real hedging strategy recall that $X^{\theta} = (S^\theta, V^\theta)$ with $S^\theta$ being the tradable assets and $V^\theta$ the non-tradable components.
Decompose the abstract hedging strategy as $\mathfrak h = (\mathfrak h^S, \mathfrak h^V)$.
Let $\bar S^{\theta}_t:= e^{-rt}S^{\theta}_t$ and note that due to~\eqref{eq:nsde2} we have 
\[
d(\bar S^{\theta}_t) = e^{-rt}\sigma^S(t,X_t^{\theta},\theta)\,dW_t\,.
\]
If we can solve $\mathfrak h^S = e^{-rt} \bar{\mathfrak h}^S_t \sigma^S(t, X^\theta_t, \theta)$ for $\bar{\mathfrak h}^S_t$ then this is a real hedging strategy.  

Therefore, an alternative approach to~\eqref{eq loss cv}, possibly yielding a better hedge, but worse variance reduction would be to consider finding
\begin{equation}\label{eq loss cv real}
\bar \xi^{\ast} \in  \argmin_{\bar \xi} \mathbb V\text{ar}\left[\Phi((X_t)_{t\in[0,T]}) -  \int_0^T \bar {\mathfrak h}(r,(X_{r\wedge t})_{t\in[0,T]},\bar \xi) \,d \bar S^\theta_r \bigg | \mathcal F_0\right]
\end{equation}
for some other neural network $\bar{\mathfrak h}$. 
This is the version we present in Algorithms~\ref{alg LSV calibration vanilla} and~\ref{alg LSV calibration vanilla lb exotic}.

\subsection{Time discretization}  
In order to implement the~\eqref{eq LSV loss lb cv} we define partition $\pi$ of $[0,T]$ as $\pi := \{t_0, t_1, \ldots, t_{N_{\text{steps}}}=T\}$. We first approximate the stochastic integral in~\eqref{eq loss cv} with the appropriate Riemann sum. 
Depending on the choice of the neural network architecture approximating $\sigma$ in the Neural SDE~\eqref{eq:nsde} we may have $\sigma$ which grows super-linearly as a function of $x$. 
In such a case the moments of the classical Euler scheme are known blow up in the finite time, see~\cite{hutzenthaler2011strong}, even if  moments of the solution to the SDE are finite.
In order to avoid blow ups of moments of the simulated paths during training we apply tamed Euler method, see \cite{hutzenthaler2012strong,szpruch2018integrability}. The tamed Euler scheme is given by
\begin{equation}
\label{eq tamed scheme}
X_{t_{k+1}}^{\pi, \theta} = X_{t_{k}}^{\pi, \theta} + \frac{b(t_k, X_{t_{k}}^{\pi, \theta},\theta)}{1+|b(t_k,X_{t_{k}}^{\pi, \theta},\theta) |\sqrt{\Delta t_k}} \Delta t_k +\frac{\sigma(t_k,X_{t_{k}}^{\pi, \theta}, \theta)}{1+|  \sigma(t_k,X_{t_{k}}^{\pi, \theta},\theta) |\sqrt{\Delta t_k}}\Delta W_{t_{k+1}}\,,
\end{equation}
with $\Delta t_k = t_{k+1}-t_k$ and $\Delta W_{t_{k+1}} = W_{t_{k+1}}- W_{t_{k}}$.

\subsection{Algorithms} 
We now present the algorithm to calibrate the Neural SDE~\eqref{eq:nsde} to market prices of derivatives (Algorithm~\ref{alg LSV calibration vanilla}) and the algorithm to find robust price bounds
for an illiquid derivative (Algorithm~\ref{alg LSV calibration vanilla lb exotic}).
Note that during training we aim to calibrate the SDE~\eqref{eq:nsde}, and
at the same time, adapt the abstract hedging strategy to minimise the variance~\eqref{eq loss cv}.  
Therefore, we alternate two optimisations: 
\begin{enumerate}[i)]
\item During each epoch, we optimise the parameters $\theta$ of the Neural SDE, 
while the parameters of the hedging strategy $\xi$ are fixed. 
In order to calculate the 
Monte Carlo estimator $\mathbb E^{\mathbb Q^N(\theta)}[\Phi^{cv}]$ we generate $N_{\text{trn}}$ paths
$(x_{t_n}^{\pi,\theta,i})_{n=0}^{N_\text{steps}} := (s_{t_n}^{\pi,\theta,i}, v_{t_n}^{\pi,\theta,i})_{n=0}^{N_\text{steps}}$, $i=1,\ldots,N_{\text{trn}}$ using tamed Euler scheme on~\eqref{eq:nsde}. 
Furthermore, we create a copy of the generated paths, denoting them $(\tilde x_{t_n}^{\pi,,i})_{n=0}^{N_\text{steps}} := (\tilde s_{t_n}^{\pi,,i}, \tilde v_{t_n}^{\pi,,i})_{n=0}^{N_\text{steps}}$ 
such that each $\tilde{x}_{t_k}^{\pi,,i}$ does not depend on $\theta$ anymore  
for the purposes of backward propagation when 
calculating the gradient.
The paths $(\tilde{x}_{t_n}^{\pi,,i})_{n=0}^{N_\text{steps}}$ will be used as input to the parametrisation of the abstract hedging strategy $\mathfrak h$; as a result, in Algorithm~\ref{alg LSV calibration vanilla} during this phase of the optimisation, the purpose of the hedging strategy is reducing the variance of $\mathbb E^{\mathbb Q^N(\theta)}[\Phi^{cv}]$ in order to speed up the convergence of the gradient descent algorithm. 

\item During each epoch, we optimise the parameters $\xi$ of the parametrisation of the hedging strategy, while the parameters $\theta$ of the Neural SDE are fixed. 
\end{enumerate}

In both Algorithms~\ref{alg LSV calibration vanilla} and~\ref{alg LSV calibration vanilla lb exotic} as well as in
the numerical experiments, we use the squared error for the nested loss function:
$\ell(x,y)=\vert x - y \vert^2$.
Furthermore, the calibration of the Neural SDE to market derivative prices with robust price bounds for illiquid derivative in Algorithm~\ref{alg LSV calibration vanilla lb exotic}
is done using a constrained optimisation using the method of Augmented Lagrangian~\cite{hestenes1969multiplier}, with 
the update rule of the Lagrange multipliers specified in Algorithm~\ref{alg LSV Augmented Lagrangian}.

\begin{algorithm}[ht]\caption{Calibration to market European option prices for one maturity}\label{alg LSV calibration vanilla}
\begin{algorithmic}
\STATE{Input: $\pi=\{t_0,t_1,\ldots,t_{N_{\text{steps}}}\}$ time grid for numerical scheme.}
\STATE{Input: $(\Phi_i)_{i=1}^{N_{\text{prices}}}$ option payoffs.}
\STATE{
Input: Market option prices
$\mathfrak p(\Phi_j)$, $j=1,\ldots,N_{\text{prices}}$.
}
\STATE{Initialisation:  $\theta$ for neural SDE parameters, $N_{\text{trn}}\in \mathbb N$ large.}
\STATE{Initialisation: $\xi$ for control variate approximation.}
%\FOR{$i:1:N_{\text{trn}}$}
%
%\ENDFOR
\FOR{$\text{epoch}:1:N_{\text{epochs}}$}
\STATE{ 
Generate $N_{\text{trn}}$ paths
$(x_{t_n}^{\pi,\theta,i})_{n=0}^{N_\text{steps}} := (s_{t_n}^{\pi,\theta,i}, v_{t_n}^{\pi,\theta,i})_{n=0}^{N_\text{steps}} $, $i=1,\ldots,N_{\text{trn}}$ using Euler scheme on~\eqref{eq:nsde}.
and create copies $(\tilde x_{t_n}^{\pi,i})_{n=0}^{N_\text{steps}} := (\tilde s_{t_n}^{\pi,i}, \tilde v_{t_n}^{\pi,i})_{n=0}^{N_\text{steps}}$ 
such that each $\tilde{x}_{t_k}^{\pi,i}$ does not depend on $\theta$ anymore, thus $\partial_{\theta}\tilde x_{t_k}^{\pi,i}=0$.
}
\STATE{\textbf{During one epoch}: Freeze $\xi$, use Adam (see ~\cite{kingma2014adam}) to update $\theta$, where   
 \[
 \theta = \widehat {\argmin_{\theta}} 
 \sum_{j=1}^{N_{\text{prices}}}   \left( \mathbb E^{N_{\text{trn}}}
 \left[\Phi_j\left( X^{\pi,\theta} \right) - \sum_{k=0}^{N_{\text{steps}}-1}\bar{\mathfrak h}(t_k, \tilde{X}_{t_k}^{\pi,}, \xi_{j})\Delta  \tilde{\bar{S}}^{\pi,}_{t_{k}} \right]-\mathfrak p(\Phi_j)\right)^2
  \]
and where $\mathbb E^{N_{\text{trn}}}$ denotes the empirical expected value calculated on the $N_{\text{trn}}$ paths. 
}
\STATE{\textbf{During one epoch}: Freeze $\theta$, use Adam to update $\xi$, by optimising the sample variance 
\[
\xi = \widehat {\argmin_{\xi}} \sum_{j=1}^{N_{\text{prices}}} \mathbb V\text{ar}^{N_{\text{trn}}} \left[\Phi_j\left( X^{\pi,\theta} \right) - \sum_{k=0}^{N_{\text{steps}}-1} \bar{\mathfrak h} (t_k, X_{t_k}^{\pi,\theta}, \xi_{j})\Delta  \tilde{\bar S}^{\pi,\theta}_{t_{k}} \right]
\]
}
\ENDFOR
\RETURN $\theta, \xi_{j}$ for all prices $(\Phi_i)_{i=1}^{N_{\text{prices}}}$.
\end{algorithmic}
\end{algorithm}

\begin{algorithm}[ht]
\caption{Calibration to vanilla prices for one maturity with lower bound for exotic price}
\label{alg LSV calibration vanilla lb exotic}
\begin{algorithmic}
\STATE{Input: $\pi=\{t_0,t_1,\ldots,t_{N_{\text{steps}}}\}$ time grid for numerical scheme.}
\STATE{Input: $(\Phi_i)_{i=1}^{N_{\text{prices}}}$ option payoffs.}
\STATE{
Input: Market option prices
$\mathfrak p(\Phi_j)$, $j=1,\ldots,N_{\text{prices}}$.
}
\STATE{Initialisation:  $\theta$ for neural SDE parameters, $N_{\text{trn}}\in \mathbb N$ large.}
\STATE{Initialisation: $\xi$ for control variate approximation.}
\STATE{Initialisation: $\lambda,c$ for Augmented Lagrangian algorithm for constrained optimisation.}
\FOR{$\text{epoch}:1:N_{\text{epochs}}$}
\STATE{ 
Generate $N_{\text{trn}}$ paths
$(x_{t_n}^{\pi,\theta,i})_{n=0}^{N_\text{steps}} := (s_{t_n}^{\pi,\theta,i}, v_{t_n}^{\pi,\theta,i})_{n=0}^{N_\text{steps}} $, $i=1,\ldots,N_{\text{trn}}$ using the Euler-type scheme on~\eqref{eq:nsde}
and create copies $(\tilde x_{t_n}^{\pi,,i})_{n=0}^{N_\text{steps}} := (\tilde s_{t_n}^{\pi,i}, \tilde v_{t_n}^{\pi,i})_{n=0}^{N_\text{steps}}$ 
such that each $\tilde{x}_{t_k}^{\pi,i}$ does not depend on $\theta$ anymore, thus $\partial_{\theta}\tilde x_{t_k}^{\pi,i}=0$.
}
\STATE{
\textbf{During one epoch}: Freeze $\xi$, use Adam to find $\theta^{N_{\text{trn}}}$, where   
\[
 \begin{split}
 f(\theta) &:= \mathbb E^{N_{\text{trn}}} \left[\Psi(X^{\pi,\theta})-\sum_{k=0}^{N_{\text{steps}}-1}\bar{\mathfrak h}(t_k, (\tilde{X}^{\pi,}_{t_k\wedge t_j})_{j=0}^{N_{\text{steps}}}, \xi_{\Psi})\Delta \tilde{\bar S}^{\pi,}_{t_{k}} \right] \\
 h(\theta) &:=  \sum_{j=1}^{N_{\text{prices}}}   \left( \mathbb E^{N_{\text{trn}}}
 \left[\Phi_j\left( X^{\pi,\theta} \right) - \sum_{k=0}^{N_{\text{steps}}-1}\bar{\mathfrak h}(t_k, \tilde{X}_{t_k}^{\pi,}, \xi_{j})\Delta  \tilde{\bar S}^{\pi,}_{t_{k}} \right]-\mathfrak p(\Phi_j)\right)^2
 \\
 \theta &= \widehat {\argmin_{\theta}} \, f(\theta) + \lambda h(\theta) + c \cdot(h(\theta))^2
\end{split}
\]  
%\blue{What is $c(\cdot)$ in this algorithm?}
and where $\mathbb E^{N_{\text{trn}}}$ denotes the empirical expected value calculated on the $N_{\text{trn}}$ paths.  
}

\STATE{\textbf{During one epoch}: Freeze $\theta$, use Adam to update $\xi$,    
\[
\begin{split}
\xi &= \widehat {\argmin_{\xi}} \sum_{j=0}^{N_{\text{prices}}} \mathbb V\text{ar}^{N_{\text{trn}}} \left[\Phi_j\left(X^{\pi,\theta} \right) - \sum_{k=0}^{N_{\text{steps}}-1}\bar{\mathfrak h}(t_k, X_{t_k}^{\pi,\theta}, \xi_{j})\Delta \tilde{\bar S}^{\pi,\theta}_{t_{k}} \right] + \\
  &+ \mathbb V\text{ar}^{N_{\text{trn}}} \left[ \Psi\left(X^{\pi,\theta}\right)-\sum_{k=0}^{N_{\text{steps}}-1} \bar{\mathfrak h}(t_k, (X_{t_k\wedge t_j}^{\pi,\theta})_{j=0}^{N_{\text{steps}}}, \xi_{\Psi})\Delta \tilde{\bar S}^{\pi,\theta}_{t_{k}} \right] 
\end{split}
\]
}
 
 \STATE{\textbf{Every 50 updates of $\theta$}: Update $\lambda,c$ using Algorithm~\ref{alg LSV Augmented Lagrangian}}
\ENDFOR
\RETURN $\theta, \xi$.
	
\end{algorithmic}
\end{algorithm}

\begin{algorithm}[htb]\caption{Augmented Lagrangian parameters update}\label{alg LSV Augmented Lagrangian}
\begin{algorithmic}
\STATE{Input:  $\lambda>0, c>0$}
\STATE{Input: $f(\theta)$ approximated exotic price with current values of $\theta$}
\STATE{Input: $MSE(\theta)$ MSE of calibration to Vanilla prices with current values of $\theta$}
\STATE{Update $\lambda:=\lambda + c\, MSE(\theta)$}
\STATE{Update $c:=2c$}
\RETURN $c$, $\lambda$
\end{algorithmic}
\end{algorithm}

\subsection{Algorithm for multiple maturities}\label{sec many maturities}

Algorithms~\ref{alg LSV calibration vanilla} and~\ref{alg LSV calibration vanilla lb exotic}
calibrate the SDE~\eqref{eq LSV SDE} to one set of derivatives. 
If the derivatives for which we have liquid market prices can be grouped by  
maturity, as is the case e.g. for call / put prices, we
can use a more efficient algorithm to achieve the calibration. 
 
This follows the natural approach used e.g. in~\cite{cuchiero2020generative}, and in \cite{SabateSiskaSzpruch2018} in the context of learning PDEs, where the networks for $b(t,X_t^{\theta},\theta)$ and 
$\sigma(t,X_t^{\theta},\theta)$ are split into different networks, one
per maturity.
Let $\theta=(\theta_1,\ldots,\theta_{N_{m}})$, where $N_m$ is the number of maturities. 
Let 
\begin{equation} 
\label{eq one network per maturity}
\begin{split}
b(t,X_t^{\theta},\theta) & := \mathbf 1_{t\in[T_{i-1},T_i]}(t)b^i(t,X_t^{\theta},\theta_i), \quad i\in\{1,\ldots,N_{m}\},	\\
\sigma(t,X_t^{\theta},\theta) & := \mathbf 1_{t\in[T_{i-1},T_i]}(t)\sigma^i(t,X_t^{\theta},\theta_i), \quad i\in\{1,\ldots,N_{m}\},	
\end{split}
\end{equation}
with each $b^i$ and $\sigma^i$ a feed forward neural network. 

Regarding the SDE parametrisation, we fit feed-forward neural networks (see Appendix~\ref{FFNNs}) to the diffusion of the SDE of the price process under the risk-neutral measure. In the particular case, where we calibrate the Neural SDE to market data, without imposing any bounds on the resulting exotic option prices, one can then do an incremental learning as follows,
\begin{enumerate}
	\item Consider the first maturity $T_i$ with $i=1$.
	\item Calibrate the SDE using Algorithm~\ref{alg LSV calibration vanilla} to the vanilla prices in maturity $T_i$.
	\item Freeze the parameters of $\sigma_{i}$, set $i:=i+1$,
	and go back to previous step.  
\end{enumerate}
The above algorithm is memory efficient, as it only needs to backpropagate through that last maturity in each gradient descent step.

\section{Analysis of the stochastic approximation algorithm for the calibration problem}\label{sec:sgd}

\subsection{Classical stochastic gradient}
First, let us review the basics about stochastic gradient algorithm.  
Let $H:\Omega \times \Theta \rightarrow \mathbb R^d$. Consider the following optimisation problem 
\[
\min_{\theta \in \Theta}h(\theta)\,,\quad\text{where}\quad h(\theta):=\mathbb E[H(\theta)]\,.
\]
Notice that the minimization task~\eqref{eq LSV loss}
 does not fit this pattern as in our case the expectation is inside $\ell$.

Nevertheless we know that the classical gradient algorithm, with the learning rates $(\eta_k)_{k=1}^{\infty}$, $\eta_k >0$ for all $k$,  applied to this optimisation problem is given by
\[
\theta_{k+1} = \theta_k - \eta_k \partial_\theta (\mathbb E[H(\theta_k)])\,.
\]
Under suitable conditions on $H$ and on $\eta_k$, it is known that $\theta_{k}$ converges to a minimiser of $h$, see~\cite{benveniste2012adaptive}. 
As $\mathbb E[H(\theta_k)]$ can rarely be computed explicitly, the above algorithm is not practical and is replaced with stochastic gradient descent (SGD) given by
\[
\theta_{k+1} = \theta_k - \eta_k \frac{1}{N}\sum_{i=1}^{N} \partial_{\theta}H^i(\theta_k)\,,
\] 
where $(H^i(\theta))_{i=1}^{N_{batch}}$ are independent samples from the distribution of $H(\theta)$ 
and $N \in\mathbb N$ is the size of the mini-batch. 
In particular $N$ could be one. 
The choice of a ``good'' estimator for $\mathbb E[H(\theta)]$ in the context of stochastic gradient algorithms is an active research area research, see e.g.~\cite{majka2020multi}. 
When the estimator of $\mathbb E[H(\theta)]$ is unbiased, the SGD can be shown to converge to a minimum of $h$, \cite{benveniste2012adaptive}.

\subsection{Stochastic algorithm for the calibration problem}
\label{sec stoch alg for calibration}

Recall that our overall objective in calibration is to minimize some $J=J(\theta)$ given by
\[
J(\theta) = \sum_{i=1}^M \ell\Big(\mathbb E^{\mathbb Q(\theta)}[\Phi_i^{\text{cv}}], \mathfrak p(\Phi_i)\Big)\,.
\]	

We write $X^{\theta}:=(X_t^{\theta})_{t\in[0,T]}$ and note that  $\mathbb E^{\mathbb Q(\theta)} \big[ \Phi_i \big] = \mathbb E \big[ \Phi_i(X^\theta) \big]$.  
%In the calibration problem \eqref{eq LSV loss}--\eqref{eq LSV loss lb cv}
%\[
%h(\theta) := \ell(\mathbb E^{\mathbb Q}[\Phi^{cv}(X^{\theta})],\mathfrak  p(\Phi) )\,.
%\]
Noting that in the calibration part of \eqref{eq LSV loss}--\eqref{eq LSV loss lb cv} the $\mathfrak h(s,(\tilde{X}_{s\wedge t})_{t\in[0,T]},\xi) d \tilde{\bar{S}}_s$ is fixed and hence $\mathbb E^{\mathbb Q}[\partial_\theta \Phi^{cv}_i(X^{\theta})]=\mathbb E^{\mathbb Q}[\partial_\theta \Phi_i(X^{\theta})]$

We differentiate $J = J(\theta)$ and work with the pathwise representation of this derivative (using language from \cite{glasserman2013monte}). For that we impose the following assumption.     

\begin{assumption}\label{ass diff payoff}
We assume that payoffs $G:=(\Psi,\Phi)$, $G:C([0,T],\mathbb R^d) \rightarrow \mathbb R$ are such that 
\[
\partial_{\theta} \mathbb E^{\mathbb Q}\left[ G(X^{\theta}) \right] = 
 \mathbb E^{\mathbb Q}\left[ \partial_{\theta} G(X^{\theta}) \right]\,. 
\]
\end{assumption}
We refer reader to \cite[chapter 7]{glasserman2013monte} for exact conditions when exchanging integration and differentiation is possible. We also remark that for the payoffs for which the Assumption~\ref{ass diff payoff} does not hold, one can use likelihood method and more generally Malliavin weights approach for computing greeks, \cite{fournie1999applications}. 
We don't pursue this here for simplicity.

Writing $\ell = \ell(x,y)$, applying Assumption~\ref{ass diff payoff} and noting that $\mathbb E^{\mathbb Q(\theta)} \big[ \Phi_i \big] = \mathbb E \big[ \Phi(X^\theta) \big]$ we see that
\[
\begin{split}
\partial_\theta J(\theta) & = \sum_{i=1}^M (\partial_x \ell)\Big(\mathbb E^{\mathbb Q(\theta)}[\Phi_i^{\text{cv}}], \mathfrak p(\Phi_i) \Big) \partial_\theta \mathbb E^{\mathbb Q(\theta)} \big[ \Phi_i^{\text{cv}} \big] \\
& = \sum_{i=1}^M (\partial_x \ell)\Big(\mathbb E[\Phi_i^{\text{cv}}(X^\theta)], \mathfrak p(\Phi_i) \Big) \mathbb E \big[ \partial_\theta \Phi_i(X^\theta) \big]\,. 	
\end{split}
\]
Hence, if we wish to update $\theta$ to some $\tilde \theta$ in such a way that $J$ is decreased then we need to take (for some $\gamma > 0$) 
\[
\tilde \theta = \theta - \gamma \sum_{i=1}^M (\partial_x \ell)\Big(\mathbb E[\Phi_i^{\text{cv}}(X^\theta)], \mathfrak p(\Phi_i) \Big) \mathbb E \big[ \partial_\theta \Phi_i(X^\theta) \big] 
\]
so that
\[
\begin{split}
\frac{d}{d\varepsilon} J(\theta + \varepsilon (\tilde \theta - \theta)) \Big \rvert_{\varepsilon=0} 
 & = \sum_{i=1}^M (\partial_x \ell)\Big(\mathbb E[\Phi_i^{\text{cv}}(X^\theta)], \mathfrak p(\Phi_i) \Big) \mathbb E \big[ \partial_\theta \Phi_i(X^\theta) \big](\tilde \theta - \theta)\\
& = - \gamma \bigg| \sum_{i=1}^M (\partial_x \ell)\Big(\mathbb E[\Phi_i^{\text{cv}}(X^\theta)], \mathfrak p(\Phi_i) \Big) \mathbb E \big[ \partial_\theta \Phi_i(X^\theta) \big]\bigg|^2 \leq 0\,.
\end{split}
\]
If we had one network for each time step, leading to some resnet-like-network architecture for the time discretization then it may be more efficient to use a backward equation representation in the training. 
This representation can be derived using similar analysis as in~\cite{jabir2019mean}  (see also~\cite{siska2020gradient}). 

Since the summation plays effectively no role in further analysis we will assume, without loss of generality, that $M = 1$ and work with the objective
\[
h(\theta) = \ell\Big(\mathbb E^{\mathbb Q(\theta)}[\Phi^{\text{cv}}], \mathfrak p(\Phi)\Big)\,.	
\]
Then in the gradient step update we have 
\[
\partial_{\theta}h(\theta) =  \partial_x \ell\left(\mathbb E^{\mathbb Q}[\Phi^{cv}(X^{\theta})],\mathfrak  p(\Phi) \right) \mathbb E^{\mathbb Q}[\partial_\theta \Phi(X^{\theta})]\,,
\]
Since $\ell$ is typically not an identity function, a mini-batch estimator of 
$\partial_{\theta}h(\theta)$, obtained by replacing $\mathbb Q$ with $\mathbb Q^N$ given by
\[
\partial_{\theta}h^{N}(\theta) :=  \partial_x \ell\left(\mathbb E^{\mathbb Q^N}[\Phi^{cv}(X^{\theta})],\mathfrak  p(\Phi) \right) \mathbb E^{\mathbb Q^N}[\partial_\theta \Phi(X^{\theta})]\,,
\]
is a biased estimator of $\partial_{\theta}h$. Nonetheless the bias can be estimated in terms of number of samples $N$ and the variance. The fact the bias is controlled by the variance justifies why it is important to reduce the variance when calibrating the models with stochastic gradient algorithm. 
An alternative perspective is to view $\partial_{\theta}h^{N}(\theta) $ as non-linear function of $\mathbb Q^N$. It turns out that that there is a general theory studying smoothness and corresponding expansions of such functions of measures and we refer reader to \cite{chassagneux2019weak} for more details. 

In Appendix~\ref{sec grad des bias} we provide a result on the bias for a general loss function. 
For the square loss function the bias is given below.

\begin{theorem}
\label{eq cor sq loss bias}
Let Assumption \ref{ass diff payoff} hold. Consider the family of neural SDEs \eqref{eq:nsde}. For $\ell(x,y)= |x - y|^{2}$, we have
\[
\begin{split}
 	&\left| \mathbb E^{\mathbb Q}\left[\partial_{\theta} h^{N}(\theta)\right] -  \partial_{\theta}h(\theta) \right|
 	  \leq  \frac{2}{N} \left(  \mathbb Var^{\mathbb Q}[\Phi^{cv}(X^{\theta})] \right)^{1/2} \left( \mathbb Var^{\mathbb Q}[\partial_{\theta}\Phi(X^{\theta})]   \right)^{1/2}\,.
 \end{split}
\]
\end{theorem}
\begin{proof}
This is an immediate consequence of Theorem~\ref{th bias}.
\end{proof}

Hence, we see that by reducing the variance of the first term we are also reducing the bias of the gradient. This justifies superiority of learning task \eqref{eq LSV loss lb cv} over \eqref{eq LSV loss lb}.

\section{Analysis of the randomised training}\label{sec random training}
\subsection{Case of general cost function $\ell$}

While the idea of calibrating to one maturity at the time described in Section~\ref{sec many maturities} works well if our aim is only to calibrate to vanilla options, it cannot be directly applied to learn robust bounds for path dependent derivatives, see~\eqref{eq LSV loss lb}. This is because the payoff of path dependent derivatives, in general, is not an affine function of maturity. On the other hand training all neural networks at every maturity all at once, makes every step of the gradient algorithm used for training computationally heavy. 

In what follows we introduce a randomisation of the gradient so that at each step of the gradient algorithm the derivatives with respect to the network parameters are computed at only one maturity at the time while keeping parameters at all other maturities unchanged.
This is similar to the popular dropout method, see~\cite{srivastava2014dropout}, that is known to help with overfitting when training deep neural networks
but for us the main aim is computational efficiency.
Recall how we split the networks for drift and diffusion %, see~\eqref{eq one network per maturity}
\begin{equation} 
\tag{\ref{eq one network per maturity}'}
\begin{split}
b(t,X_t^{\theta},\theta) & := \mathbf 1_{t\in[T_{i-1},T_i]}(t)b^i(t,X_t^{\theta},\theta_i), \quad i\in\{1,\ldots,N_{m}\},	\\
\sigma(t,X_t^{\theta},\theta) & := \mathbf 1_{t\in[T_{i-1},T_i]}(t)\sigma^i(t,X_t^{\theta},\theta_i), \quad i\in\{1,\ldots,N_{m}\},	
\end{split}
\end{equation}
Let $U\sim \mathbb U[1,\ldots,N_{m}]$ be a uniform random variable over set $[1,\dots, N_{m}]$ defined on a new probability space $( \Omega^{\mathbb  U}, \mathcal F^{\mathbb  U}, ( \mathcal F^{\mathbb  U}_t)_{t\in[0,T]}, \mathbb P^{\mathbb  U})$. 
Let $Z$ be given by 
\[
\begin{split}
dZ_t^{\theta}(U) = & \bigg(\sum_{i=1}^{N_{m}} \mathbf 1_{[T_{i-1},T_i]}(t) \partial_x b^i(t,X_t^{\theta},\theta_i) Z_t^{\theta}(U)   +  N \mathbf 1_{[T_{U-1},T_U]}(t) \partial_{\theta_{U}} b^{U}(t,X_t^{\theta},\theta_{U}) \bigg)\,dt \\  
& + \bigg(\sum_{i=1}^{N_{m}} \mathbf 1_{[T_{i-1},T_i]}(t)\partial_x \sigma^i(t,X_t^{\theta},\theta_i) Z_t^{\theta}(U)   +  N \mathbf 1_{[T_{U-1},T_U]}(t) \partial_{\theta_{U}} \sigma^{U}(t,X_t^{\theta},\theta_{U}) \bigg)\, dW_t \,, 
\end{split}
\]
where $b^U,\sigma^U$ are simply neural networks sampled from the random index $U$.

\begin{theorem}
Assume $\partial_x [b,\sigma](t,\cdot,\theta)$ exists and is bounded with $(t,\theta)$ fixed and $\partial_\theta [b,\sigma](t,x,\cdot)$ exists and is bounded with $(t,x)$ fixed. 
Let
$
h(\theta)= \ell(\mathbb E^{\mathbb Q(\theta)}[\phi(X_t^{\theta})], \mathfrak  p(\Phi))
$
and let its randomised gradient be 
\[
(\partial_\theta h)(\theta,U) = \partial_x \ell  (\mathbb E^{\mathbb Q(\theta)}[\phi(X_t^{\theta})] ,\mathfrak  p(\Phi)) \,\mathbb E^{ \mathbb Q(\theta)}[(\partial_x\phi)( X_t^{\theta})Z_t^{\theta}(U) ]\,.
\]
Then $\mathbb E^{\mathbb U}[(\partial_\theta h)(\theta,U)] = (\partial_\theta h)(\theta) $.
In other words the randomised gradient is an unbiased estimator of $(\partial_\theta h)(\theta)$.
\end{theorem}
Less stringent assumption on derivatives of $b$ and $\sigma$ are possible, but we do not want to overburden the present article with technical details.
\begin{proof}
It is well known, e.g \cite{krylov1999kolmogorov,kunita1997stochastic}, that
\[
\begin{split}
d(\partial_\theta X_t^{\theta}) = \sum_{i=1}^{N_{m}} \mathbf 1_{[T_{i-1},T_i]}(t)   \bigg[ & \Big( (\partial_x b^i(t,X_t^{\theta},\theta_i) \partial_{\theta} X_t^{\theta}   +   \partial_{\theta_i}b^i(t,X_t^{\theta},\theta_i) \Big)\,dt \\ 
& + \Big( (\partial_x \sigma^i(t,X_t^{\theta},\theta_i) \partial_{\theta} X_t^{\theta}   +   \partial_{\theta_i}\sigma^i(t,X_t^{\theta},\theta_i) \Big)\,dW_t \bigg] \,. 	
\end{split}
\]
%Note, this is an afine SDE and can be solve explicitly. Indeed, let
%\[
%\begin{split}
%A_t & :=\sum_{i=1}^{N_{m}} \mathbf 1_{t\in[T_{i-1},T_i]}(t)    \partial_x b^i(t,X_t^{\theta},\theta_i)\,, \quad Q_t:=\sum_{i=1}^{N_{m}} \mathbf 1_{t\in[T_{i-1},T_i]}(t)  \partial_{\theta_i}b^i(t,X_t^{\theta},\theta_i)\,, 	\\
%B_t & :=\sum_{i=1}^{N_{m}} \mathbf 1_{t\in[T_{i-1},T_i]}(t)    \partial_x \sigma^i(t,X_t^{\theta},\theta_i)\,, \quad R_t:=\sum_{i=1}^{N_{m}} \mathbf 1_{t\in[T_{i-1},T_i]}(t)   \partial_{\theta_i}\sigma^i(t,X_t^{\theta},\theta_i)\,, 	
%\end{split}
%\] 
%and note that 
%\[
%\partial_\theta X_t^{\theta}= \Lambda_{t,0}\left( \partial_\theta X_t^{\theta} + \int_0^t  \Lambda_{s,0}^{-1} \big(Q_s - A_s R_s \big)\, ds + \int_0^t \Lambda_{s,0}^{-1} R_s\, dW_s  \right)\,,
%\]
%where
%\[
%\Lambda_{t,0} = \exp\left(  \int_{0}^t \bigg( A_s - \frac{B_s^2}{2}\bigg)\, ds + \int_0^t B_s\, dW_s \right)\,.
%\]
Let $U\sim \mathbb U[1,\ldots,N_{m}]$ be a uniform random variable over set $[1,\dots, N_{m}]$ defined on a new probability space $( \Omega^{\mathbb  U}, \mathcal F^{\mathbb  U}, ( \mathcal F^{\mathbb  U}_t)_{t\in[0,T]}, \mathbb P^{\mathbb  U})$. 
We introduce process $Z$ as follows
\[
\begin{split}
dZ_t^{\theta}(U) = & \bigg(\sum_{i=1}^{N_{m}} \mathbf 1_{[T_{i-1},T_i]}(t) \partial_x b^i(t,X_t^{\theta},\theta_i) Z_t^{\theta}(U)   +  N \mathbf 1_{[T_{U-1},T_U]}(t) \partial_{\theta_{U}} b^{U}(t,X_t^{\theta},\theta_{U}) \bigg)\,dt \\  
& + \bigg(\sum_{i=1}^{N_{m}} \mathbf 1_{[T_{i-1},T_i]}(t)\partial_x \sigma^i(t,X_t^{\theta},\theta_i) Z_t^{\theta}(U)   +  N \mathbf 1_{[T_{U-1},T_U]}(t) \partial_{\theta_{U}} \sigma^{U}(t,X_t^{\theta},\theta_{U}) \bigg)\, dW_t \,. 
\end{split}
\]
Note that 
\[
\mathbb E^{\mathbb U}[ N \mathbf 1_{[T_{U-1},T_U]}(t)\partial_{\theta_{U}} b^{U}(t,X_t^{\theta},\theta_{U}) ] = \sum_{i=1}^{N_{m}} \mathbf 1_{[T_{i-1},T_i]}(t) \partial_{\theta_i}b^i(t,X_t^{\theta},\theta_i) 
\]
and
\[
\mathbb E^{\mathbb U}[ N \mathbf 1_{[T_{U-1},T_U]}(t) \partial_{\theta_{U}} \sigma^{U}(t,X_t^{\theta},\theta_{U}) ] = \sum_{i=1}^{N_{m}} \mathbf 1_{[T_{i-1},T_i]}(t) \partial_{\theta_i}\sigma^i(t,X_t^{\theta},\theta_i) \,.
\]
Now using Fubini-type Theorem for Conditional Expectation, \cite[Lemma A5]{hammersley2019weak}, we have
\[
\begin{split}
d\mathbb E^{\mathbb U} \big[Z_t^{\theta}(U)\big] =  & \bigg(\sum_{i=1}^{N_{m}} \mathbf 1_{t\in[T_{i-1},T_i]}(t) \partial_x b^i(t,X_t^{\theta},\theta_i) \mathbb E^{\mathbb U} \big[Z_t^{\theta}(U)\big]   +  \sum_{i=1}^{N_{m}} \mathbf 1_{t\in[T_{i-1},T_i]}(t) \partial_{\theta_{i}} b^{i}(t,X_t^{\theta},\theta_{i}) \bigg)\,dt \\  
& + \bigg(\sum_{i=1}^{N_{m}} \mathbf 1_{t\in[T_{i-1},T_i]}(t) \partial_x \sigma^i(t,X_t^{\theta},\theta_i) E^{\mathbb U} \big[Z_t^{\theta}(U)\big]  +  \sum_{i=1}^{N_{m}} \mathbf 1_{t\in[T_{i-1},T_i]}(t) \partial_{\theta_{i}} \sigma^{i}(t,X_t^{\theta},\theta_{i}) \bigg)\, dW_t \,. 
\end{split}
\]
Hence the process $\mathbb E^{\mathbb U} \big[Z_t^{\theta}(U)\big]$ solves the same linear equation as $\partial_\theta X^\theta$. 
As the equation has unique solution we conclude that  
\begin{equation} \label{eq equivalence of gradients}
\mathbb E^{\mathbb U}[ Z_t^{\theta}(U) ] = \partial_\theta X_t^{\theta} \,.
\end{equation}
Recall that
$
h(\theta)= \ell(\mathbb E^{\mathbb Q(\theta)}[\phi(X_t^{\theta})], \mathfrak  p(\Phi))
$,
and so
\[
(\partial_\theta h)(\theta) = \partial_x \ell (\mathbb E^{\mathbb Q(\theta)}[\phi(X_t^{\theta})] ,\mathfrak  p(\Phi)) \,\mathbb E^{ \mathbb Q(\theta)}[(\partial_x\phi)( X_t^{\theta})\partial_\theta X_t^{\theta}]\,.
\]
Recall the randomised gradient 
\[
(\partial_\theta h)(\theta,U) =  \partial_x \ell  (\mathbb E^{\mathbb Q(\theta)}[\phi(X_t^{\theta})] ,\mathfrak  p(\Phi)) \,\mathbb E^{ \mathbb Q(\theta)}[(\partial_x\phi)( X_t^{\theta})Z_t^{\theta}(U) ]\,.
\]
Note that due to \eqref{eq equivalence of gradients} 
\[
\mathbb E^{\mathbb U}[\mathbb E^{ \mathbb Q(\theta)}[(\partial_x\phi)( X_t^{\theta})Z_t^{\theta}(U) ]] 
 =\mathbb E^{ \mathbb Q(\theta)}[(\partial_x\phi)( X_t^{\theta})\partial_\theta X_t^{\theta}]\,.
\]
this implies that $\mathbb E^{\mathbb U}[(\partial_\theta h)(\theta,U)] = (\partial_\theta h)(\theta) $. 
\end{proof}

\subsection{Case of square loss function $\ell$}
\label{sec unbiased error}

Here we show that, in the special case when $\ell(x,y)= |x - y|^{p}$, randomised gradient as described in Section~\ref{sec random training} is an unbiased estimator of the full gradient even in the case when $\mathbb Q$ is replaced by its empirical measure $\mathbb Q^N$. Consequently standard theory on stochastic approximation applies. We base the presentation on the case $p=2$, as general case works in exact the same way. Let $(\bar \Omega, \bar {\mathcal F}, ( \bar{\mathcal F_t})_{t\in[0,T]} ,\bar {\mathbb P})$ 
be a copy of $( \Omega, \mathcal F, ( \mathcal F_t)_{t\in[0,T]}, \mathbb P)$. Then we write
\[
h^N(\theta):= (\mathbb E^{\mathbb Q^N(\theta)}[\phi(X_t^{\theta})] - \mathfrak  p(\Phi))^2 
= (\mathbb E^{ {\bar{\mathbb Q}^N(\theta)}}[\phi(\bar X_t^{\theta})] - \mathfrak  p(\Phi))(\mathbb E^{\mathbb Q^N(\theta)}[\phi(X_t^{\theta})] - \mathfrak  p(\Phi_i))\,.
\]
See also \cite{cuchiero2020generative} for the same observation.  The gradient  of $h$ is given by 
\[
\begin{split}
(\partial_\theta h^N)(\theta)
= & (\mathbb E^{ {\bar{\mathbb Q}^N(\theta)}}[\phi(\bar X_t^{\theta})] - \mathfrak  p(\Phi))(\mathbb E^{\mathbb Q^N(\theta)}[(\partial_x \phi)(X_t^{\theta}) \partial_\theta X_t^{\theta} ] - \mathfrak  p(\Phi_i)) \\
& + (\mathbb E^{ {\bar{\mathbb Q}^N(\theta)}}[(\partial_x\phi)(\bar X_t^{\theta})\partial_\theta \bar X_t^{\theta}] - \mathfrak  p(\Phi))(\mathbb E^{\mathbb Q^N(\theta)}[\phi(X_t^{\theta})] - \mathfrak  p(\Phi_i))\,.
\end{split}
\]
Equivalently 
\[
\begin{split}
(\partial_\theta h^N)(\theta)
(\partial_\theta h^N)(\theta) = & 2 (\mathbb E^{ {\bar{\mathbb Q}^N(\theta)}}[\phi(\bar X_t^{\theta})] - \mathfrak  p(\Phi))(\mathbb E^{\mathbb Q^N(\theta)}[(\partial_x \phi)(X_t^{\theta}) \partial_\theta X_t^{\theta} ] - \mathfrak  p(\Phi_i)) \\
 = & 2 (\mathbb E^{ {\bar{\mathbb Q}^N(\theta)}}[(\partial_x\phi)(\bar X_t^{\theta})\partial_\theta \bar X_t^{\theta}] - \mathfrak  p(\Phi))(\mathbb E^{\mathbb Q^N(\theta)}[\phi(X_t^{\theta})] - \mathfrak  p(\Phi_i))\,.
\end{split}
\]
To implement the above algorithm one simply needs to generate two independent sets of samples.  Furthermore, 
\[
\begin{split}
(\partial_\theta h^N)(\theta,U)
(\partial_\theta h^N)(\theta) = & 2 (\mathbb E^{ {\bar{\mathbb Q}^N(\theta)}}[\phi(\bar X_t^{\theta})] - \mathfrak  p(\Phi))(\mathbb E^{\mathbb Q^N(\theta)}[(\partial_x \phi)(X_t^{\theta}) Z_t^{\theta}(U) ] - \mathfrak  p(\Phi_i)) \\
 = & 2 (\mathbb E^{ {\bar{\mathbb Q}^N(\theta)}}[(\partial_x\phi)(\bar X_t^{\theta})Z_t^{\theta}(U)] - \mathfrak  p(\Phi))(\mathbb E^{\mathbb Q^N(\theta)}[\phi(X_t^{\theta})] - \mathfrak  p(\Phi_i))\,.
\end{split}
\]
is an unbiased estimator of $(\partial_\theta h^N)(\theta)$.

% For instance for the diffusion in the price process, we do
%\[
%\sigma^S(t,X_t,V_t, \theta) = \mathbf 1_{t\in[T_{i-1},T_i]}(t)\sigma_i^X(t,X_t,V_t, \theta_i), \quad i\in\{1,\ldots,N_{\text{maturities}}\},
%\]
%with each $\sigma_i^X$ a feed forward neural network. One can then do an incremental learning as follows,
%\begin{enumerate}
%	\item Consider the first maturity $T_i$ with $i=1$.
%	\item Calibrate the SDE using Algorithms~\ref{alg LSV calibration vanilla}, ~\ref{alg LSV calibration vanilla lb exotic} to the vanilla prices in maturity $T_i$.
%	\item Freeze the parameters of $\sigma_{i}^S,b_{i}^V,\sigma_{i}^V$, set $i:=i+1$,
%	and go back to previous step.  
%\end{enumerate}

\section{Testing neural SDE calibrations}
\label{sec numerics}

All algorithms were implemented using \textsc{PyTorch}, see~\cite{paszke2017automatic} and~\cite{paszke2019pytorch}.
The code used is available at \href{https://github.com/msabvid/robust_nsde}{\texttt{github.com/msabvid/robust\_nsde}}.
Our target data (European option prices for various strikes and maturities) is described in Appendix~\ref{CalData}. 
We assume that there is one traded asset $S = (S_t)_{t\in [0,T]}$.
We calibrate to European option prices
\[
\mathfrak p(\Phi) := \mathbb E^{\mathbb Q(\theta)}[\Phi] = e^{-rT} \mathbb E^{\mathbb Q(\theta)} \left[\left(S_T - K \right)_+ | \, S_0 = 1\right]
\]
for maturities of $2,4,\ldots,12$ months and typically $21$ uniformly spaced strikes between in $[0.8, 1.2]$. 
As an example of an illiquid derivative for which we wish to find robust bounds we take the lookback option
\[
\mathfrak p(\Psi) := \mathbb E^{\mathbb Q(\theta)}[\Psi] = e^{-rT} \mathbb E^{\mathbb Q(\theta)} \left[\max_{t\in[0,T]}S_t - S_T | \, X_0 = 1 \right]\, .
\]

\subsection{Local volatility neural SDE model}
\label{sec:LVintro}

In this section we consider a Local Volatility (LV) Neural SDE model.
It has been shown by~\cite{dupire1994pricing} (see also~\cite{gyongy1986mimicking}) that if the market data would consist of a continuum of call / put prices for all strikes and maturities then there is a unique function $\sigma$ such that with the price process
\begin{align}\label{LV}
  dS_t = rS_t\, dt + S_t \, \sigma(t,S_t)dW_t\,,\,\,\,S_0=1\,
\end{align}
the model prices and market prices match exactly. 
In practice only some calls / put prices are liquid in the market and so to apply~\cite{dupire1994pricing} one has to interpolate, in an arbitrage free way, the missing data. 
The choice of interpolation method is a further modelling choice on top of the one already made by postulating that the risky asset evolution is governed by~\eqref{LV}.   

We will use a Neural SDE instead of directly interpolating the missing data. 
Let our LV Neural SDE model be given by
\begin{equation}\label{NeuralLV}
dS_t^\theta=rS_t^\theta dt+\sigma(t,S_t^\theta, \theta)S_t^\theta dW_t^{\mathbb Q},\\
\end{equation}
where $S_t^\theta \geq 0$, $S_0^\theta=1$ and $\sigma:[0,T]\times \mathbb R \times \mathbb R^p \to \mathbb R^+$ allows us to calibrate the model to observed market prices.

\subsection{Local stochastic volatility neural SDE model} 
\label{sec LSV}

In this section we consider a Local Stochastic Volatility (LSV) Neural SDE model.
See, for example,~\cite{tian2015calibrating}.
As in the Local Volatility Neural SDE model~\eqref{NeuralLV}, we have the risky asset price price process $(S_t)_{t\in[0,T]}$, where the drift is equal to the risk-free bond rate $r$. 
However, the volatility function in the LSV Neural SDE model
now depends on $t$, $S_t$ and a stochastic process $(V_t)_{t\in[0,T]}$.
Here $(V_t)_{t\in[0,T]}$ is not a traded asset.
The model is then given by
\begin{equation}\label{eq LSV SDE}
\begin{split}
dS_t & = r S_t dt + \sigma^S(t,S_t, V_t, \nu)S_t \,dB^S_t, \quad S_0 = 1,\\
dV_t & = b^V(V_t, \phi)\,dt + \sigma^V(V_t, \varphi) \,dB^V_t, \quad V_0 = v_0, \\
d\langle B^S, B^V\rangle_t & = \rho dt 
\end{split}
\end{equation}
where 
$\theta := \{\nu, \phi, \varphi, v_0, \rho \},\,\,\rho,v_0\in\mathbb R, 
$
as the set of (multi-dimensional) parameters that we aim to optimise so that 
the model is calibrated to the observed market data.

\subsection{Deep learning setting for the LV and LSV neural SDE models}
In the SDE~\eqref{NeuralLV} the function $\sigma$ and the SDE~\eqref{eq LSV SDE} the functions $\sigma^S$, $b^V$ and $\sigma^V$ are parametrised by one feed-forward neural 
network per maturity (see Section~\ref{sec many maturities} and Appendix~\ref{FFNNs}) with $4$ hidden layers with $50$ neurons in each layer. 
The non-linear activation
function used in each of the hidden layers is the linear rectifier \texttt{relu}. In addition, in $\sigma^S$ and $\sigma^V$
after the output layer we apply the non-linear rectifier \texttt{softplus}$(x)=\log(1+\exp(x))$ to ensure a positive output. 

The parameterisation of the hedging strategy for the vanilla option prices is also a feed-forward
linear network with 3 hidden layers, 20 neurons per hidden layer and \texttt{relu} activation functions. 
However, in order to get one hedging strategy per vanilla option 
considered in the market data, the output of $\mathfrak{h}(t_k, s_{t_k, \theta}^{i}, \xi_{K_j})$ has as many neurons as strikes and maturities.

Finally, the parameterisation of the hedging strategy for the exotic options price is also a feed-forward
network with 3 hidden layers, 20 neurons per hidden layer and \texttt{relu} activation functions.

The Neural SDEs~\eqref{NeuralLV} and~\eqref{eq LSV SDE} were discretized using the tamed Euler scheme~\eqref{eq tamed scheme} with $N_{\text{steps}} = 8 \times 12$ uniform time steps for $T=1$ year (i.e. $16$ for every $2$ months). 
The number of Monte Carlo trajectories in each stochastic gradient descent iteration was $N=4\times 10^4$ and the abstract hedging strategy was used as a control variate.

Finally in the evaluation of the calibrated Neural SDE, the option prices are calculated using $N = 4\times 10^5$ trajectories of the calibrated Neural SDEs, which we generated with $2\times 10^5$ Brownian paths and their antithetic paths; in addition we also used the learned hedging strategies to calculate the Monte Carlo estimators with lower variance $\mathbb E^{\mathbb Q(\theta)}[\Phi_i^{cv}]$ and $\mathbb E^{\mathbb Q(\theta)}[\Psi^{cv}]$.

\begin{figure}[h]
\centering 
\includegraphics[width=0.45\textwidth]{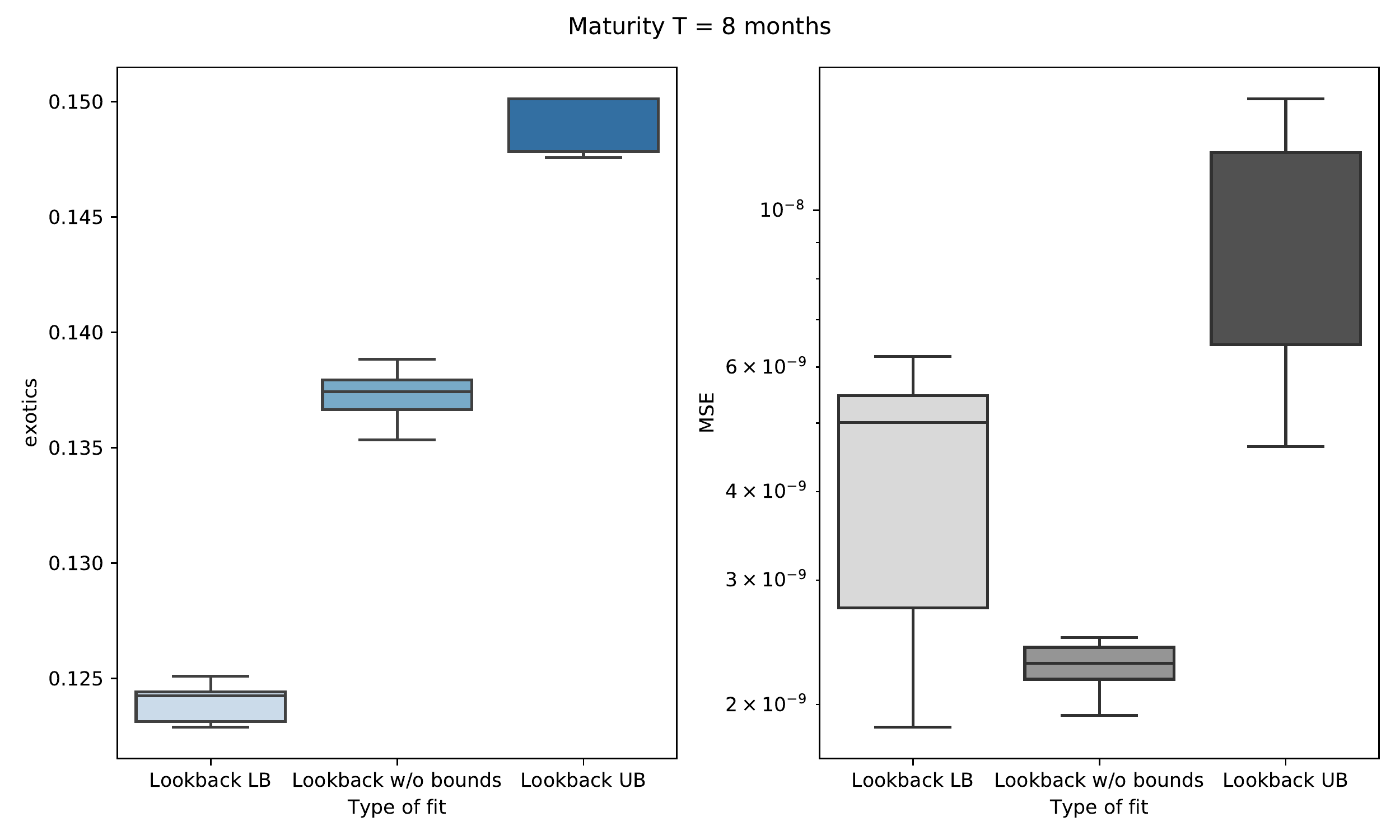}	
\includegraphics[width=0.45\textwidth]{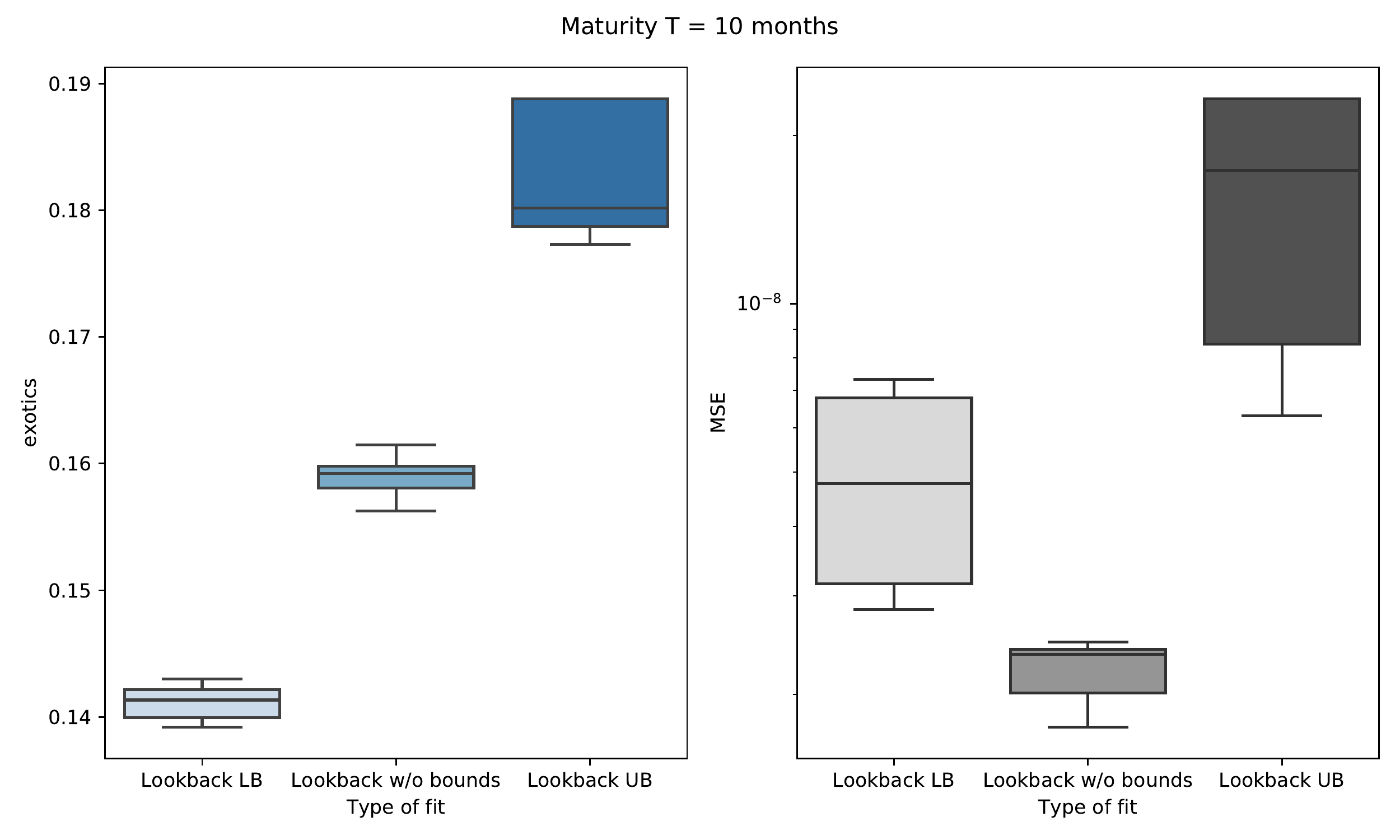}
\includegraphics[width=0.45\textwidth]{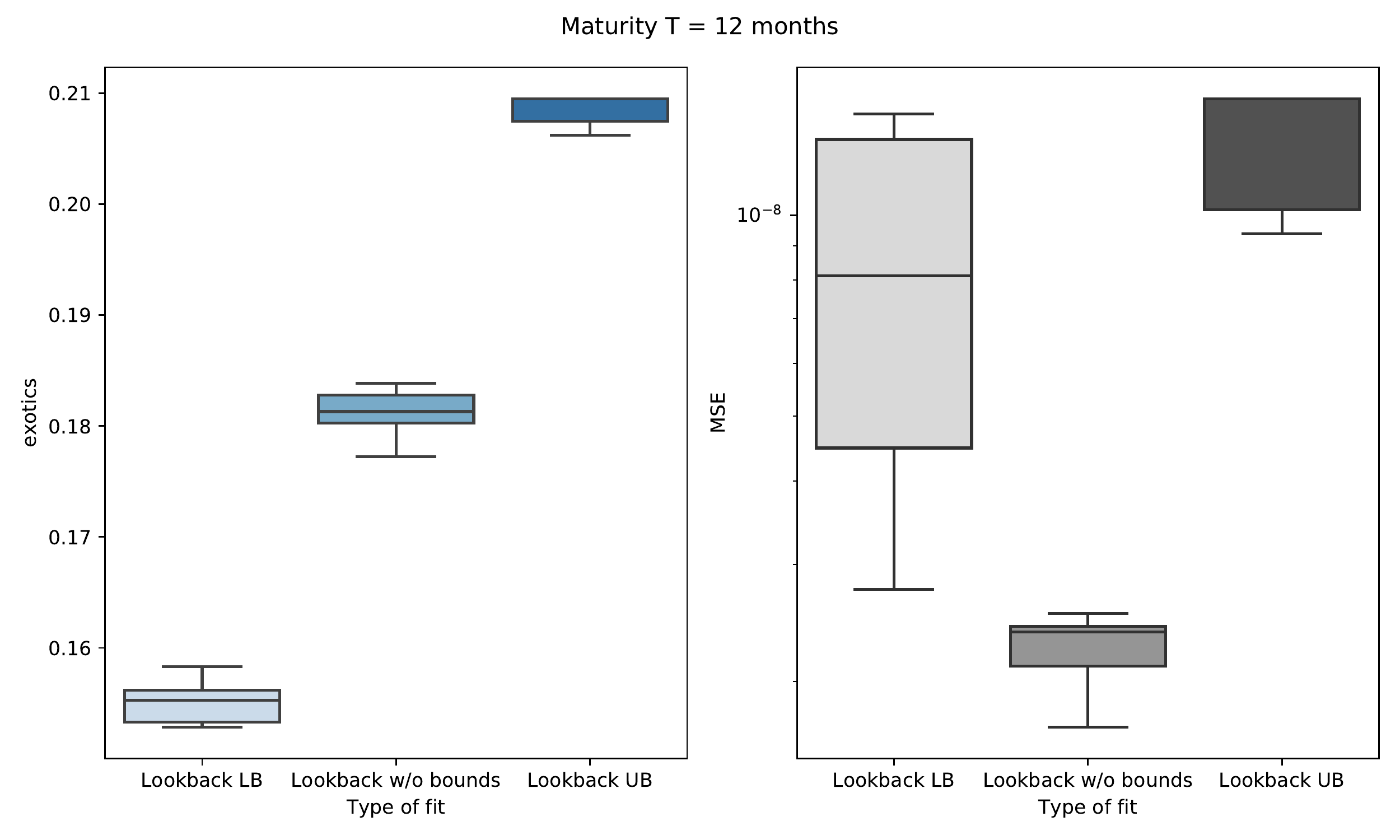}
\caption{Box plots for the Local Volatility model~\eqref{NeuralLV}. 
Exotic option price quantiles are in blue in the left-hand box-plot groups.
The MSE quantiles of market data calibration is in grey, in the right-hand box-plot groups. 
Each box plot comes from 10 different runs of Neural SDE calibration. 
The three box-plots in each group arise respectively from aiming for a lower bound of the illiquid derivative(left), only calibrating to market and then pricing the illiquid derivative (middle) and aiming for a lower bound of the illiquid derivative (right).}
\label{fig LV boxplots}
\end{figure}

\subsection{Conclusions from calibrating for LV neural SDE}
Each calibration is run $10$ times with different initialisations of the network parameters, with the goal to check the robustness of the exotic option price 
 $\mathbb E^{\mathbb Q(\theta)}[\Psi]$ for each calibrated Neural SDE. 
The blue boxplots in Figure~\ref{fig LV boxplots} provide different quantiles for the exotic option 
 price $\mathbb E^{\mathbb Q(\theta)}[\Psi]$ and the obtained bounds after running 
We make the following observations from calibrating LV Neural SDE:
\begin{enumerate}[i)]
\item It is possible to obtain high accuracy of calibration with MSE of about $10^{-9}$ for $6$ month maturity, about $10^{-8}$ for 12 month maturity when the only target is to fit market data. If we are minimizing / maximizing the illiquid derivative price at the same time then the MSE increases somewhat so that it is about $10^{-8}$ for both $6$ and $12$ month maturities. 
See Figure~\ref{fig LV boxplots}. 
The calibration has been performed using $K=21$ strikes.
\item The calibration is accurate not only in MSE on prices but also on individual implied volatility curves, see Figure~\ref{FigUnc}
 and others in Appendix~\ref{LVfitquality}.
\item As we increase the number of strikes per maturity the range of possible values for the illiquid derivative narrows. 
See Figure~\ref{Strikechart} and Tables~\ref{InitialisationImpact11strikes}, \ref{InitialisationImpact21strikes}, \ref{InitialisationImpact31strikes} and~\ref{InitialisationImpact41strikes}.
The conjecture is that as the number of strikes (and maturities) would increase to infinity we would recover the unique $\sigma$ given by the Dupire formula that fits the continuum of European option prices.
\item With limited amount of market data (which is closer to practical applications) even the LV Neural SDE produces noticeable ranges for prices of illiquid derivatives, see again Figure~\ref{fig LV boxplots}.  
\end{enumerate}

\begin{figure}[h]
\centering 
\includegraphics[clip,width=0.3\textwidth]{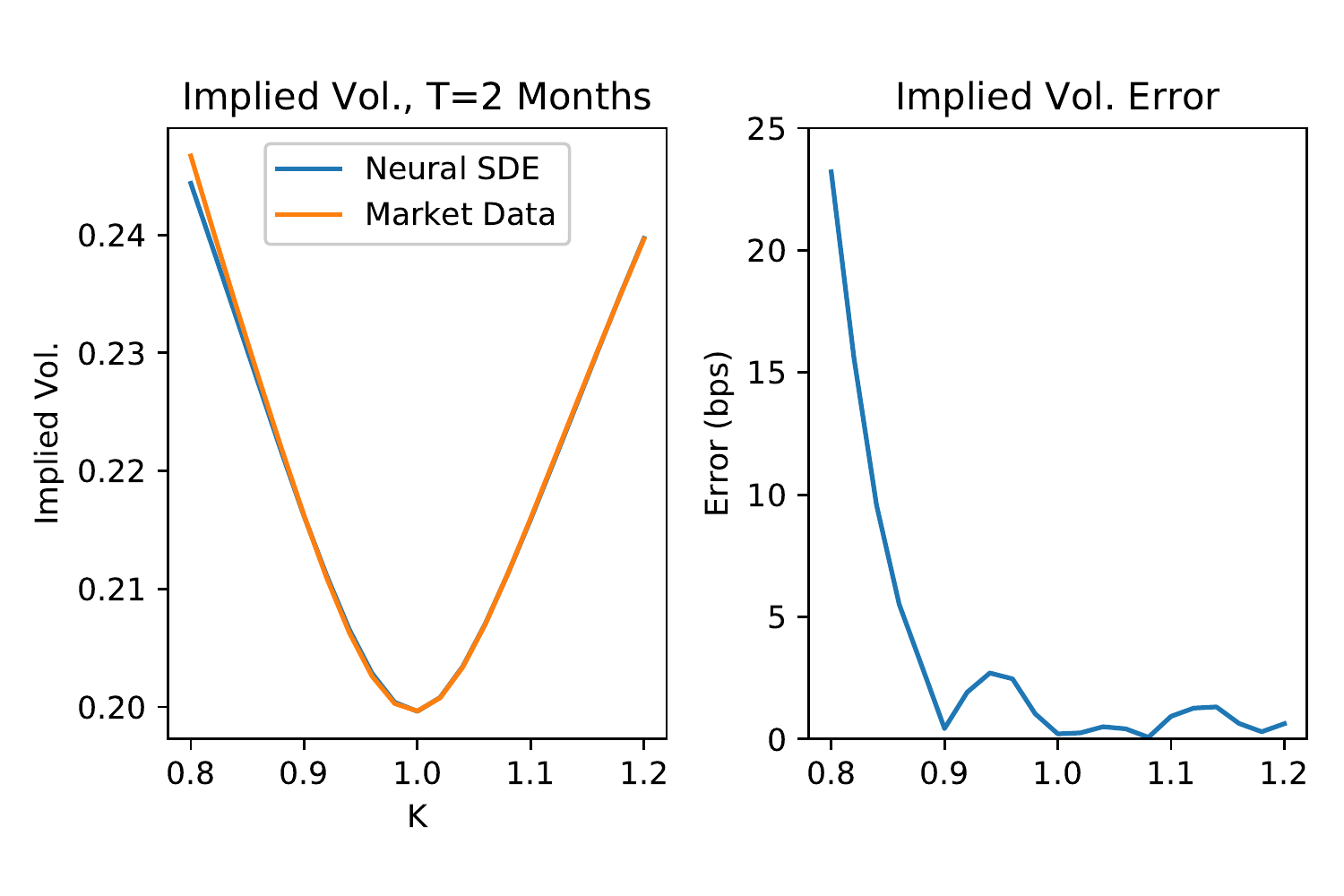}
\includegraphics[clip,width=0.3\textwidth]{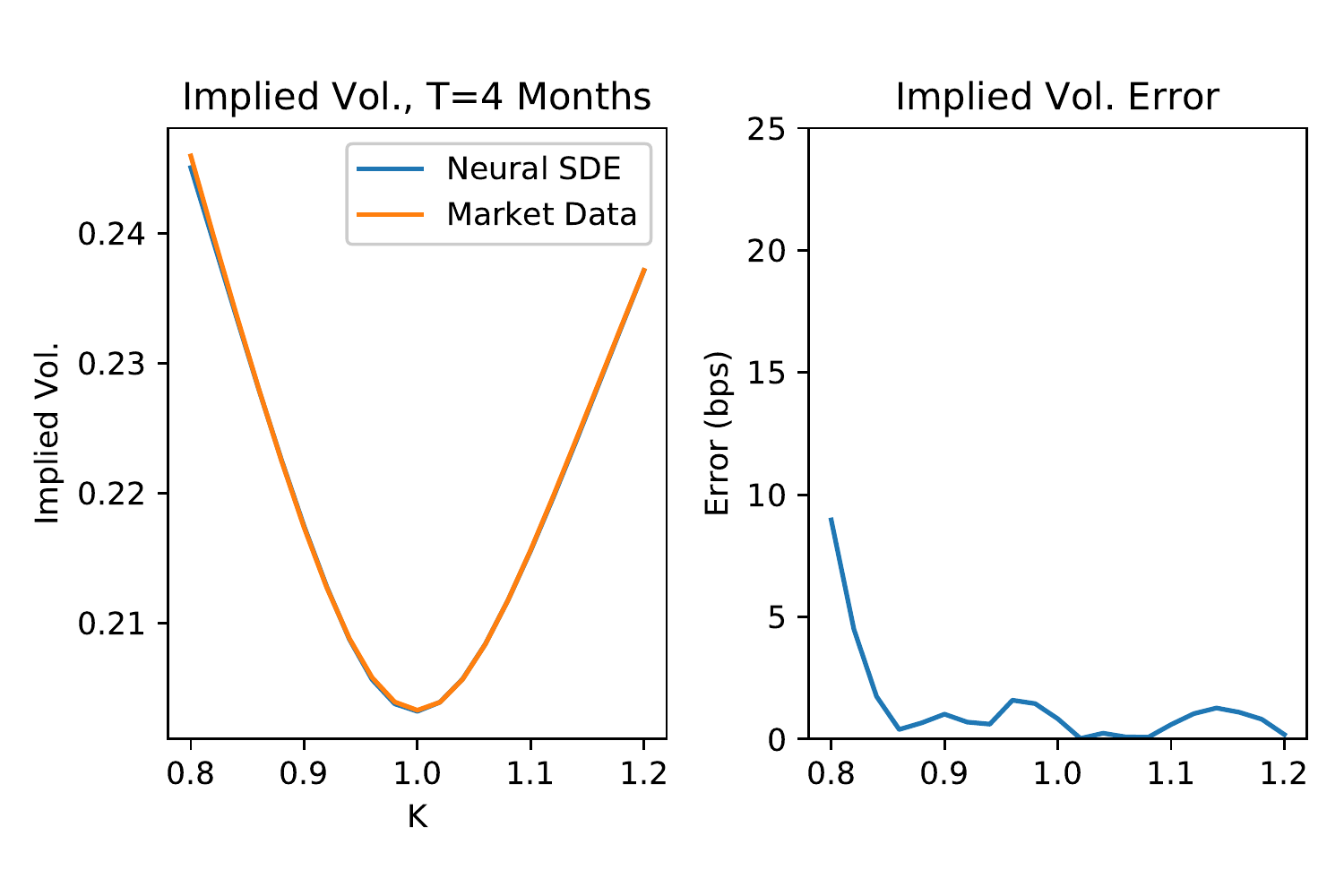}
\includegraphics[clip,width=0.3\textwidth]{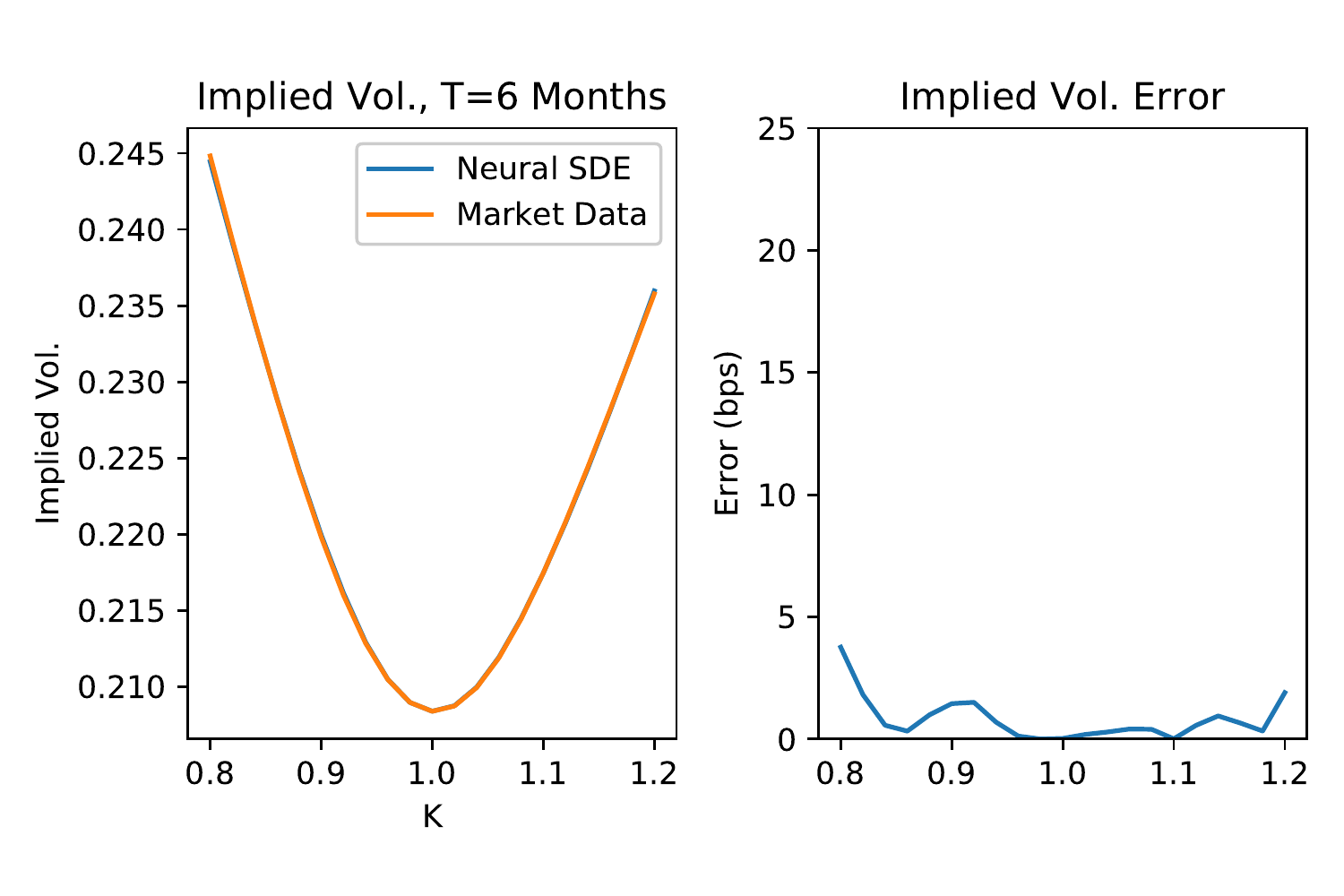} 
\includegraphics[clip,width=0.3\textwidth]{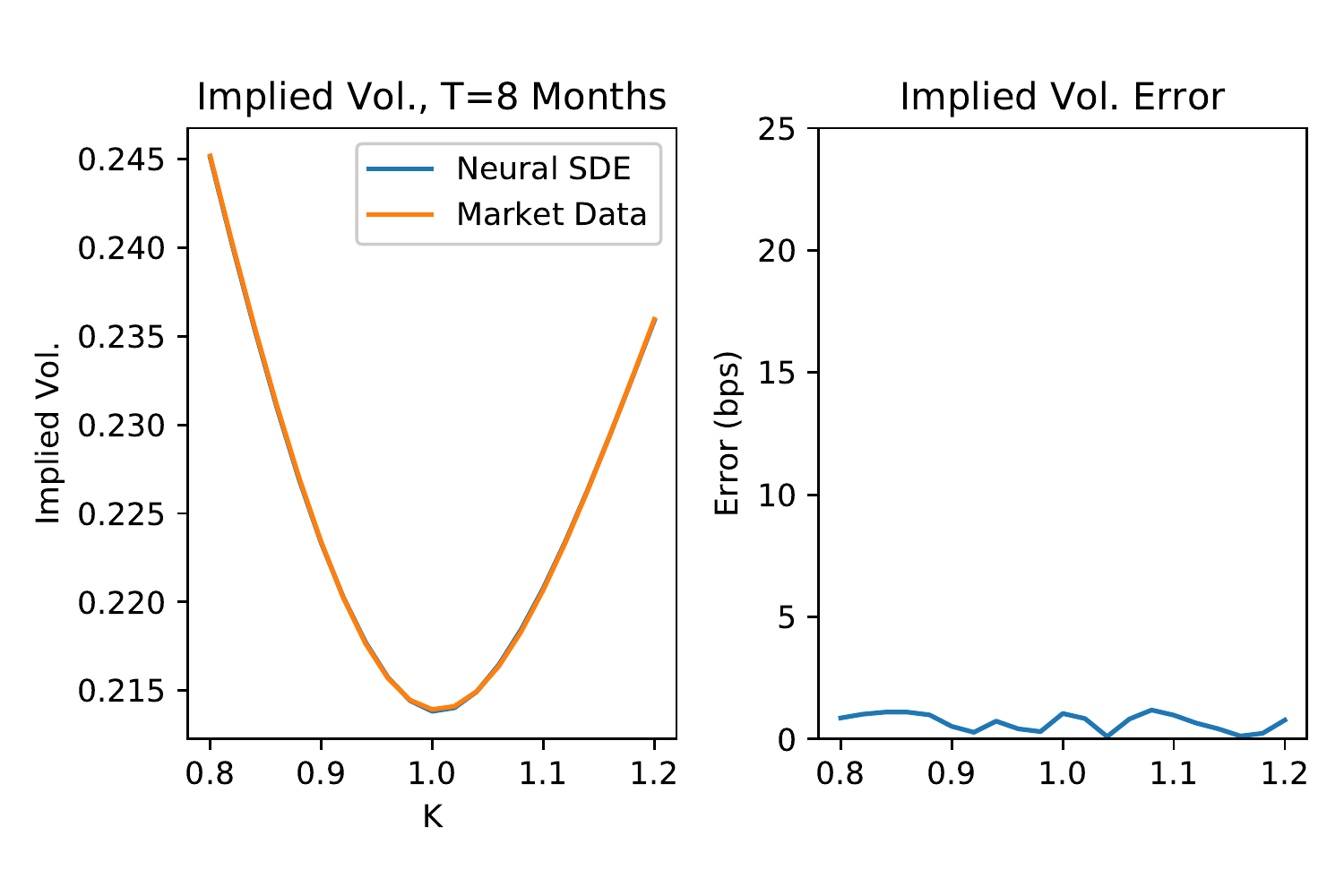}
\includegraphics[clip,width=0.3\textwidth]{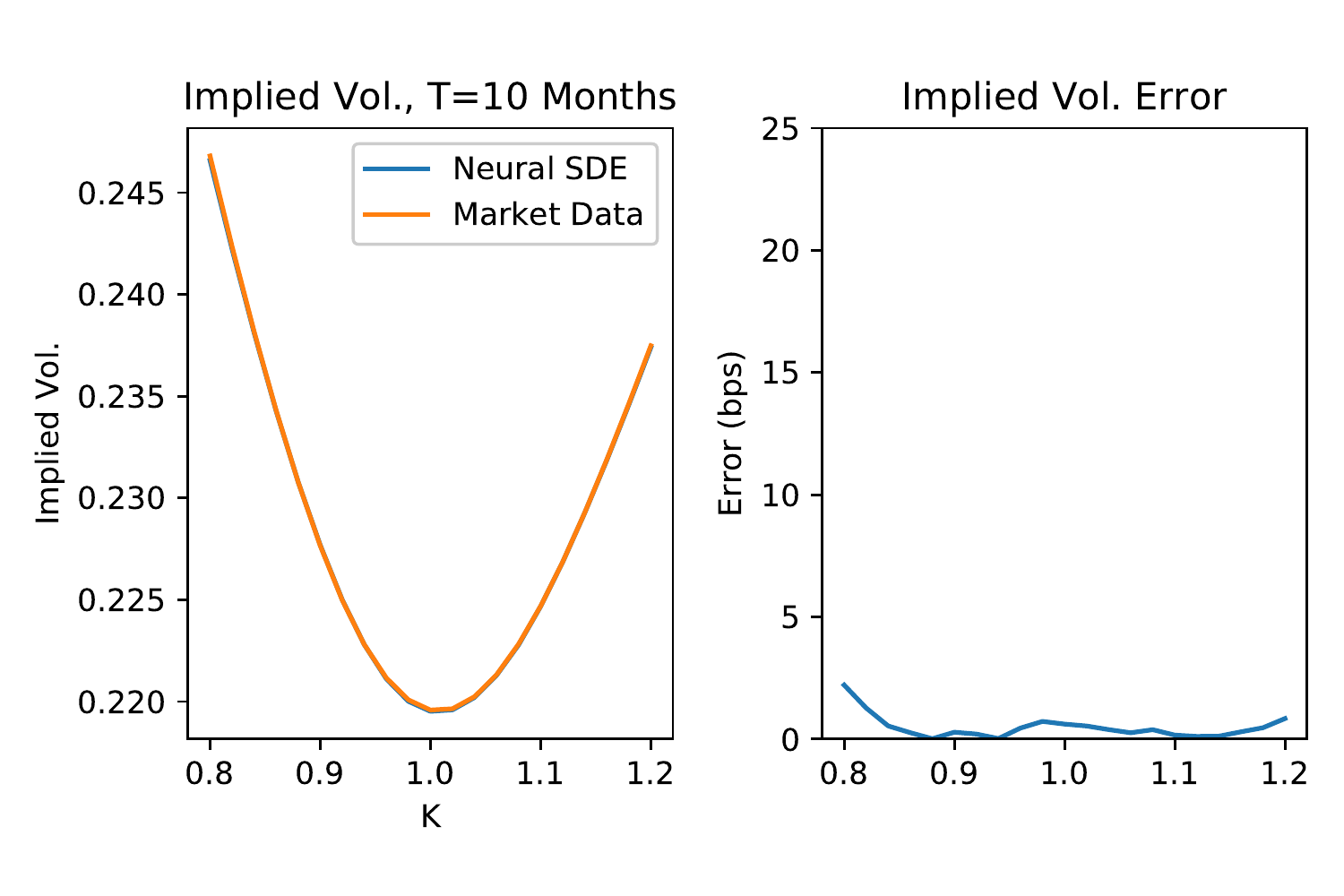} 
\includegraphics[clip,width=0.3\textwidth]{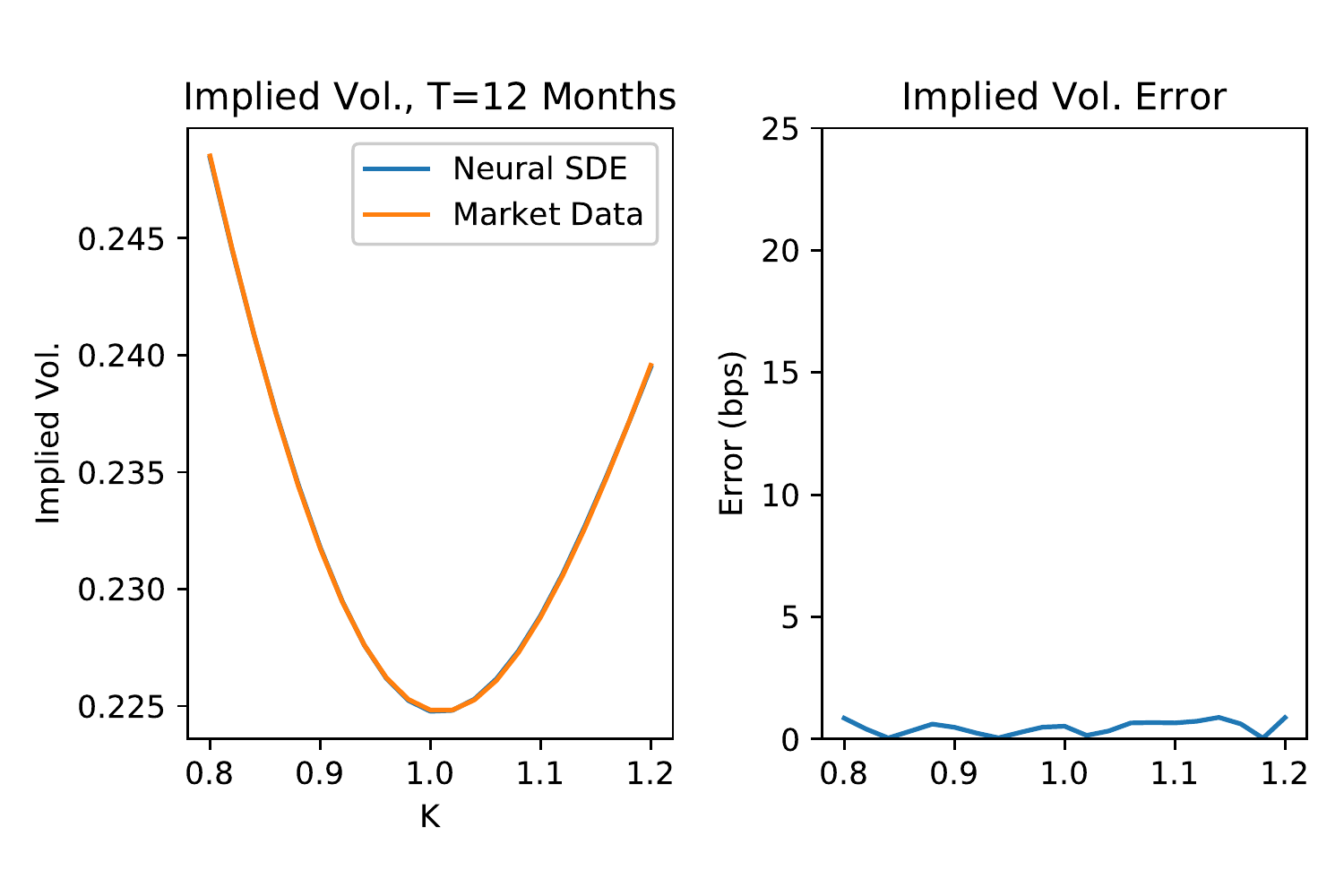}
\caption{Calibrated neural SDE LV model and target market implied volatility comparison.}
\label{FigUnc}
\end{figure}

In Appendix~\ref{LVtables} we provide more details on how different random seeds, different constrained optimization algorithms and different number of strikes used in the market data input affect the illiquid derivative price.  

In Appendix~\ref{LVfitquality} we present market price and implied  volatility fit for constrained and unconstrained calibrations. 
High level of accuracy in all calibrations is achieved due to the hedging neural network incorporated into model training.

\begin{figure}[h]
  \centering
  \includegraphics[clip,width=0.45\textwidth]{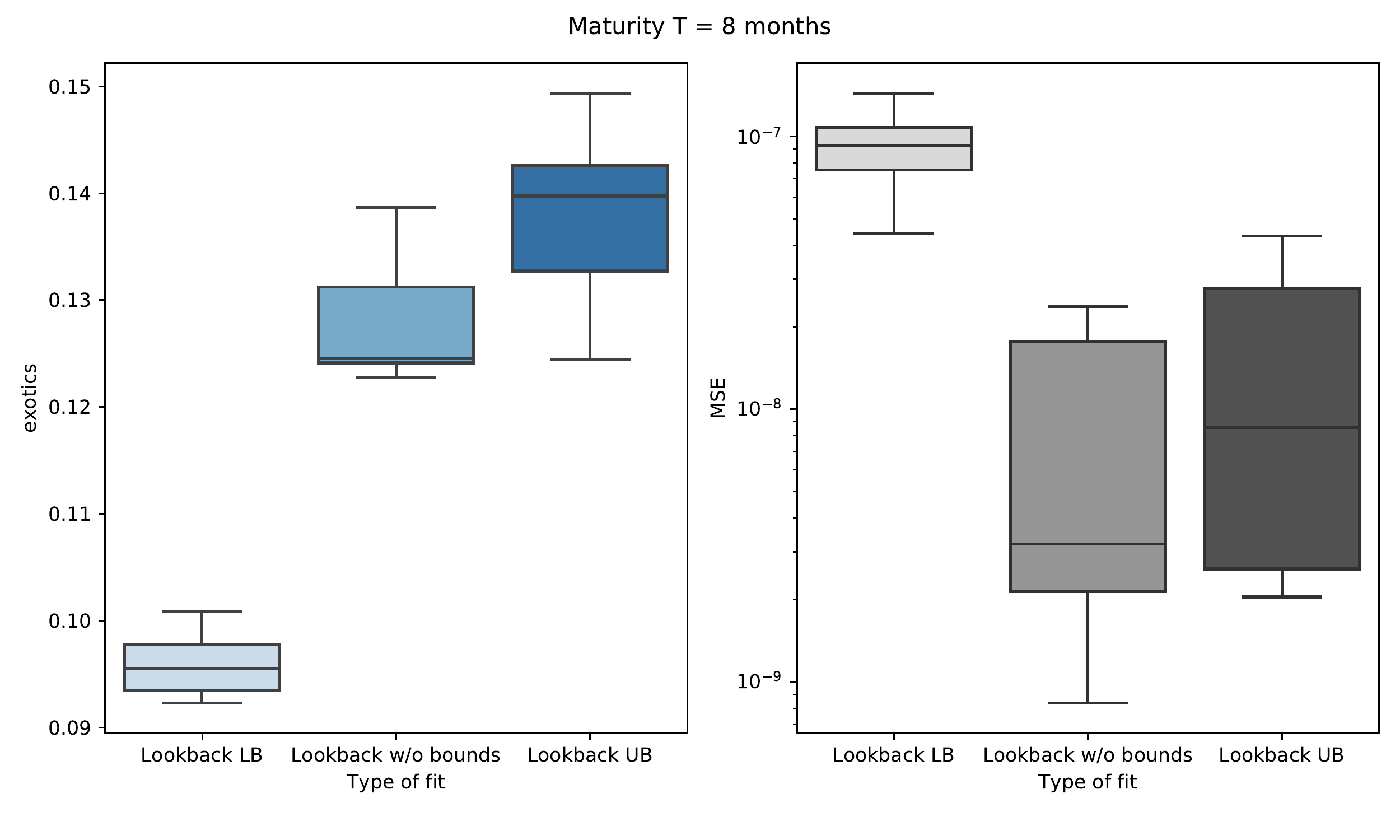}	
  \includegraphics[clip,width=0.45\textwidth]{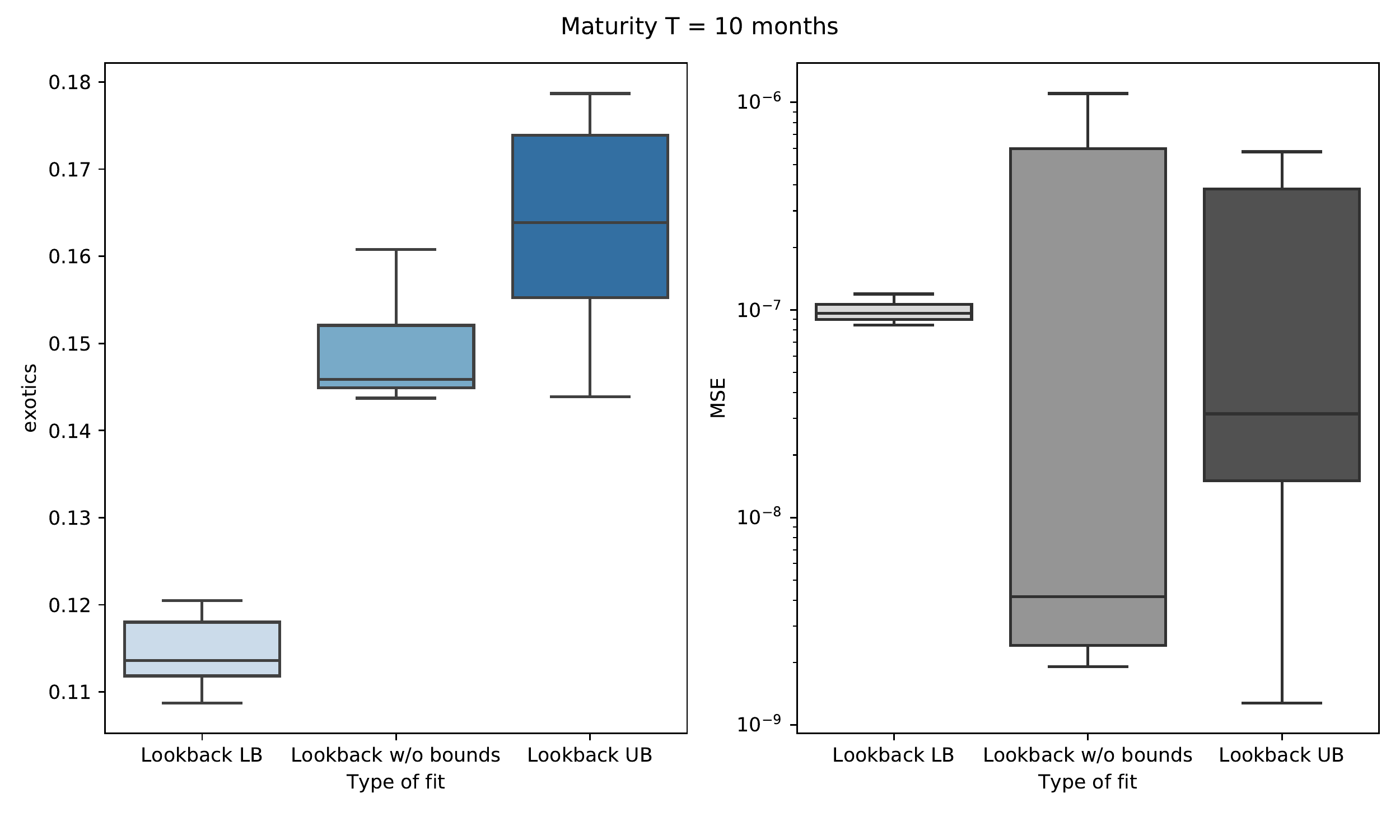}
  \includegraphics[clip,width=0.45\textwidth]{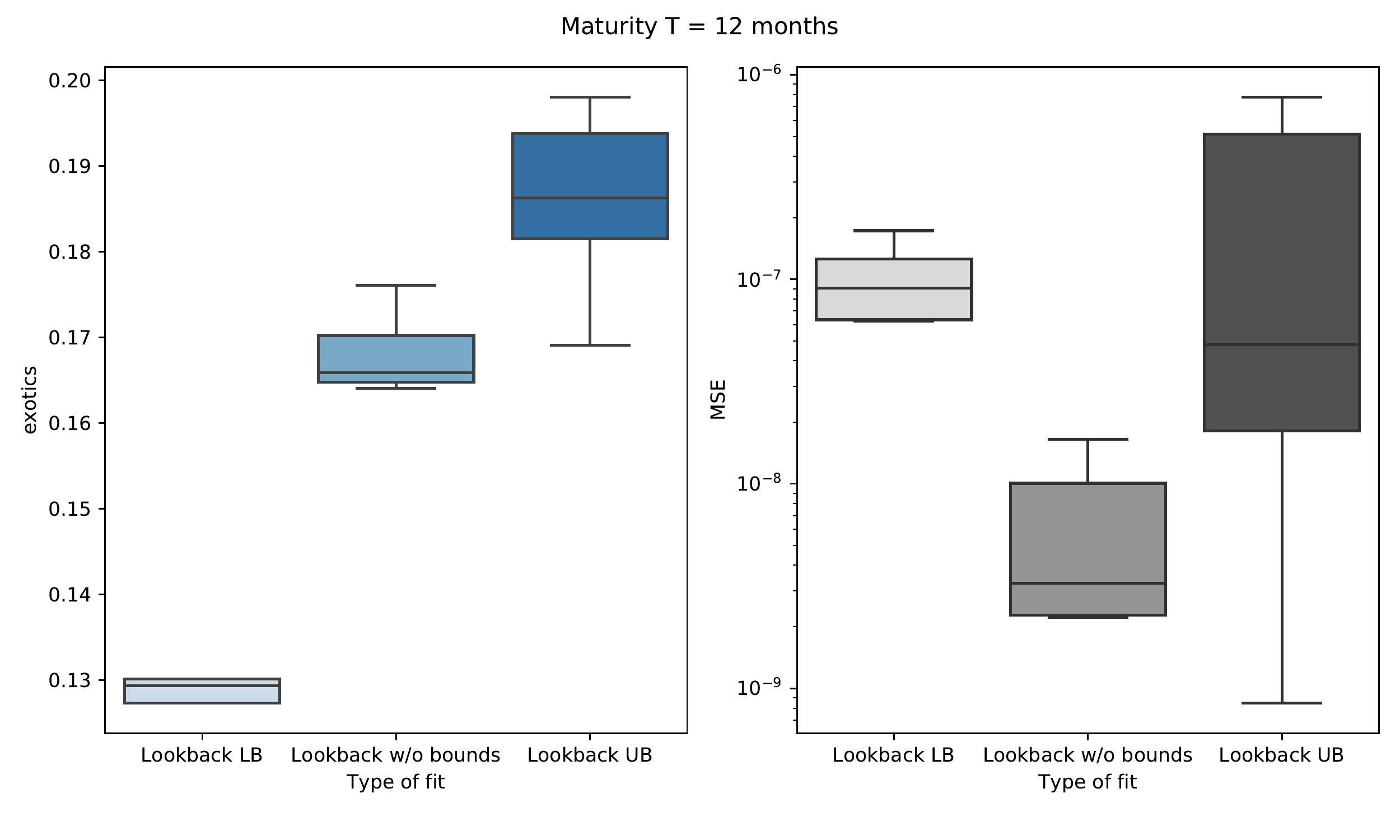} 
  \caption{Box plots for the Local Stochastic Volatility model~\eqref{eq LSV SDE}. Exotic option price quantiles are in blue in the left-hand box-plot groups.
The MSE quantiles of market data calibration is in grey, in the right-hand box-plot groups. 
Each box plot comes from 10 different runs of Neural SDE calibration. 
The three box-plots in each group arise respectively from aiming for a lower bound of the illiquid derivative(left), only calibrating to market and then pricing the illiquid derivative (middle) and aiming for a lower bound of the illiquid derivative (right).}
\label{fig LSV boxplots}
\end{figure}

\subsection{Conclusions from calibrating for LSV neural SDE}
Each calibration is run ten times with different initialisations of the network parameters, with the goal to check the robustness of the exotic option price 
 $\mathbb E^{\mathbb Q(\theta)}[\Psi]$ for each calibrated Neural SDE. 
The blue boxplots in Figure~\ref{fig LSV boxplots} provide different quantiles for the exotic option 
 price $\mathbb E^{\mathbb Q(\theta)}[\Psi]$ and the obtained bounds after running all the experiments $10$ times. 
We make the following observations from calibrating LSV Neural SDE:
\begin{enumerate}[i)]
\item We note that our methods achieve high calibration accuracy to the market data (measured by MSE) with consistent bounds on the exotic option prices. 
See Figure~\ref{fig LSV boxplots}.
\item The calibration is accurate not only in MSE on prices but also on individual implied volatility curves, see Figure~\ref{fig LSV calibration}
 and others in Appendix~\ref{sec LSVfitquality}.
\item The LSV Neural SDE produces noticeable ranges for prices of illiquid derivatives, see again Figure~\ref{fig LSV boxplots}.  
\end{enumerate}

\begin{figure}[h]
  \centering 
  \includegraphics[clip,width=0.3\textwidth]{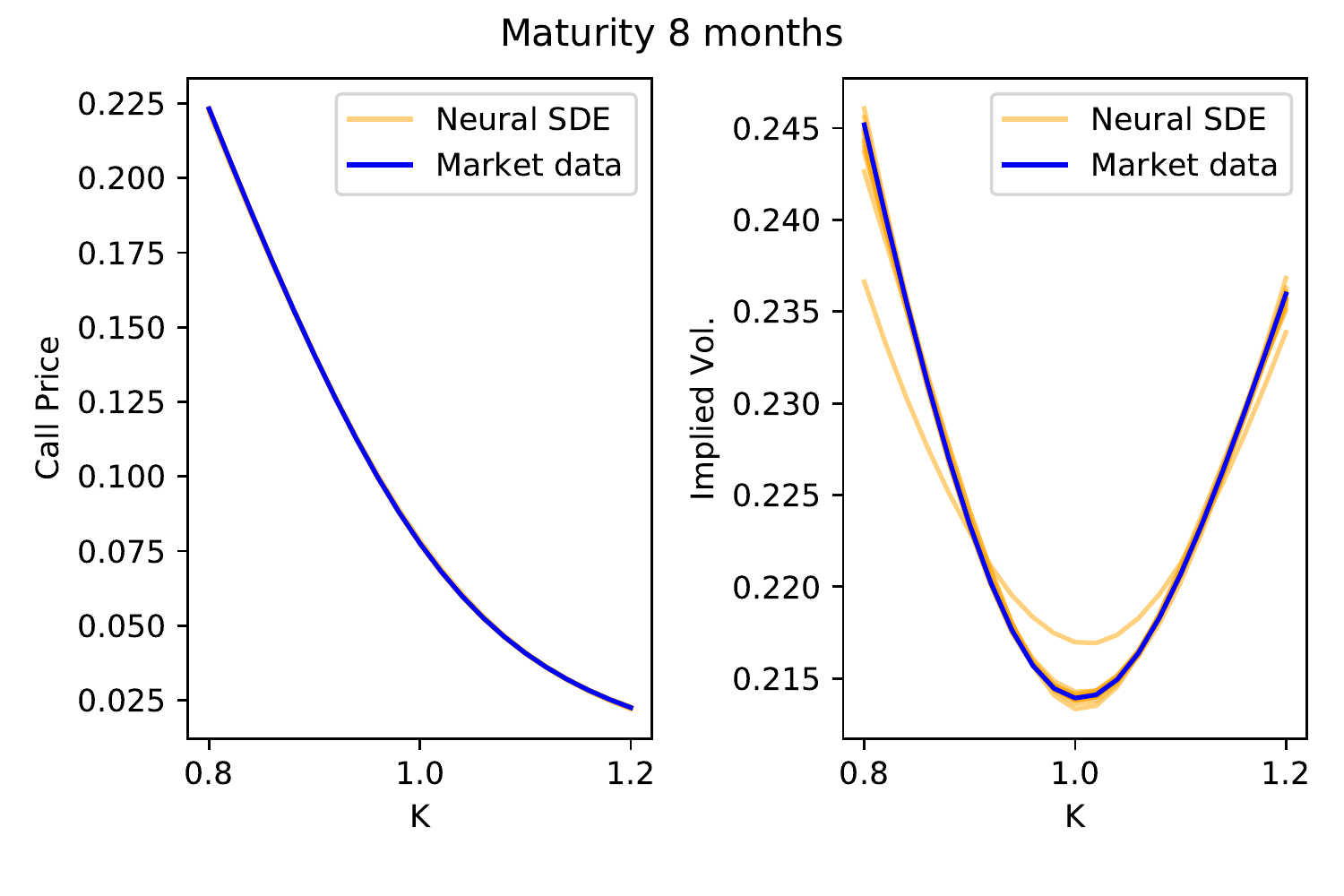}
  \includegraphics[clip,width=0.3\textwidth]{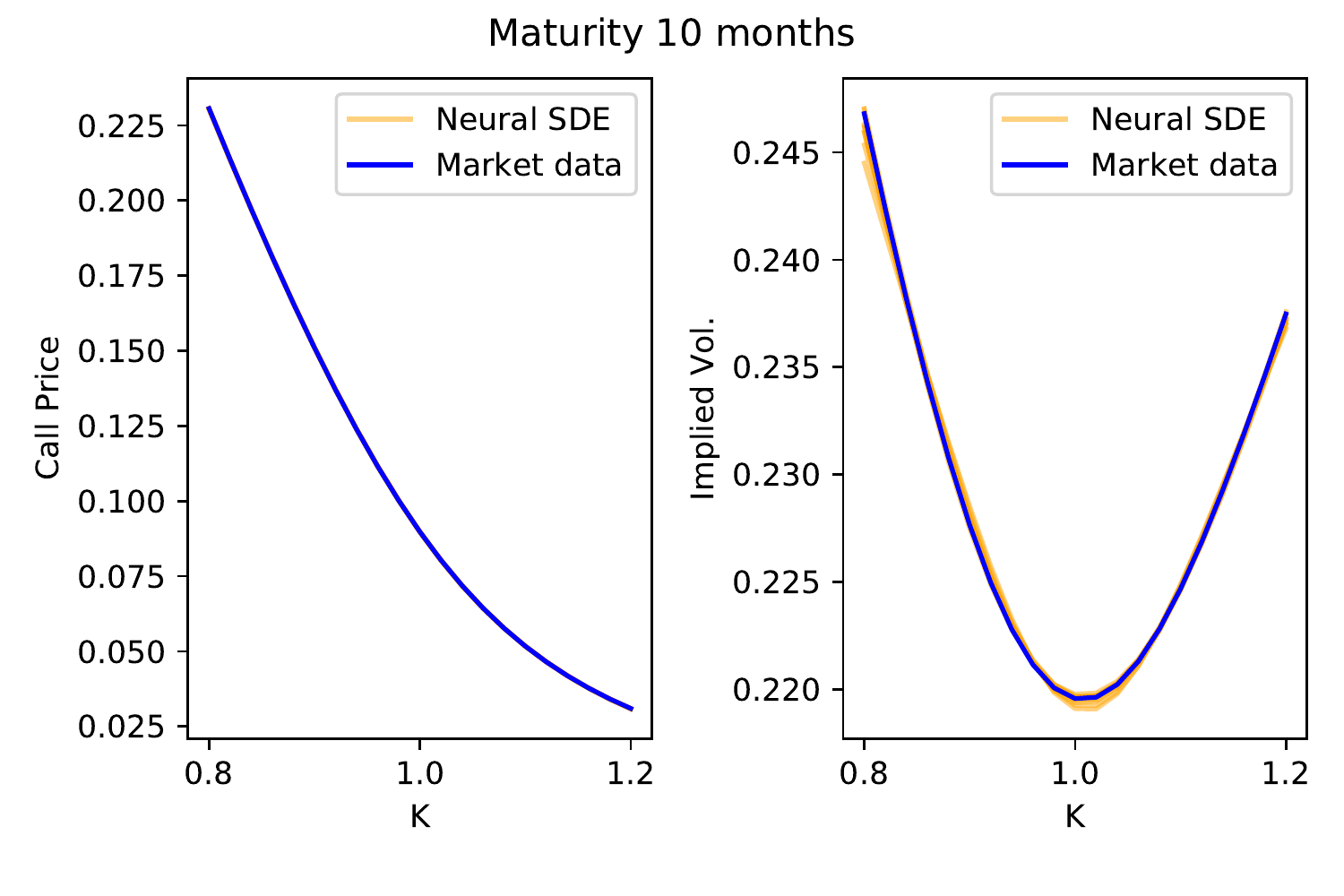}
  \includegraphics[clip,width=0.3\textwidth]{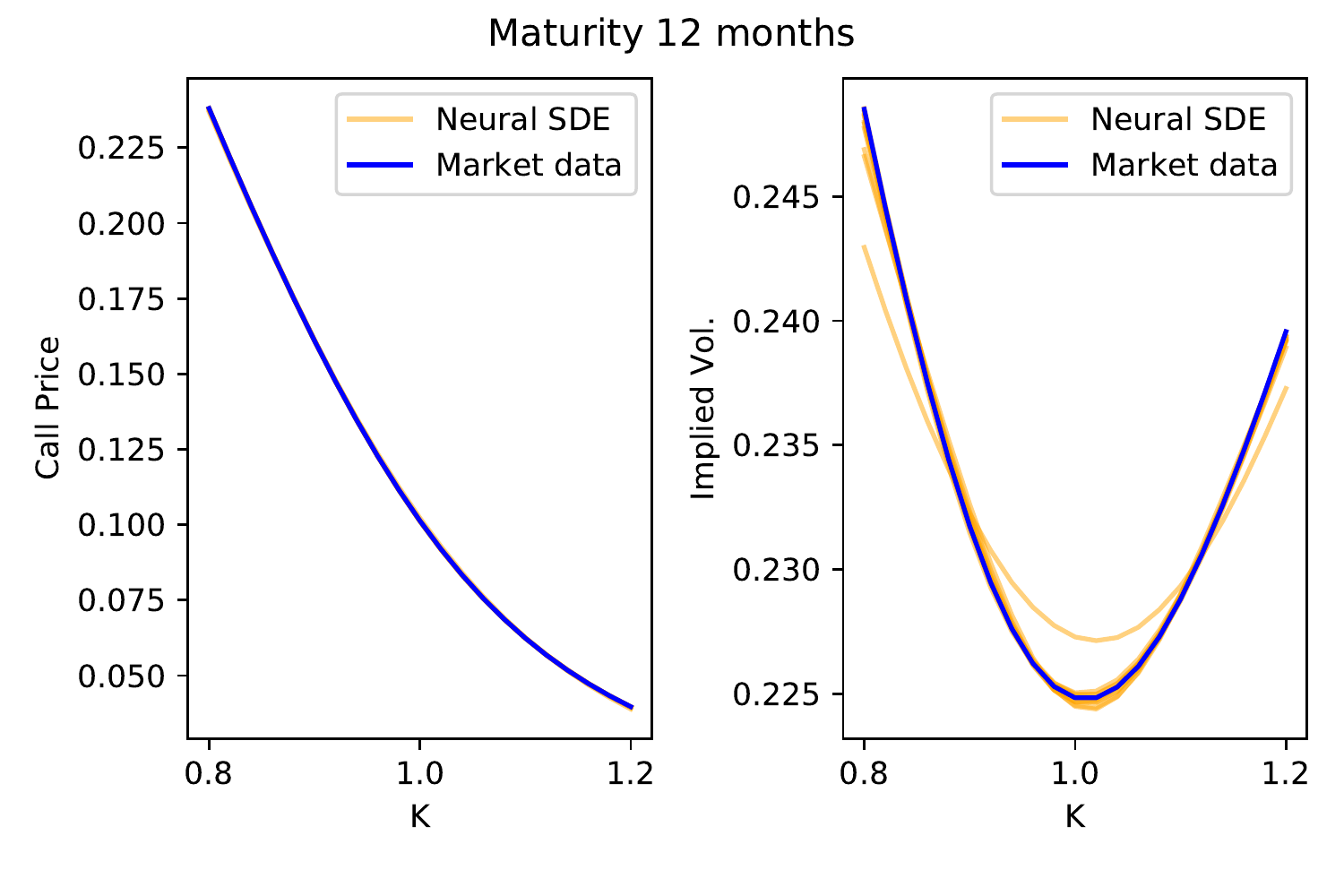} 
  \caption{Comparing market and model data fit for the Neural SDE LSV model~\eqref{eq LSV SDE} when targeting only the {\em market data}.
We see vanilla option prices and implied volatility curves of the 10 calibrated Neural SDEs vs. the market data for different maturities.
}
\label{fig LSV calibration}  
\end{figure}

\subsection{Hedging strategy evaluation} We calculate the error of the portfolio hedging strategy of the lookback option 
at maturity $T=6$ months, given by the empirical variance 
\[
\mathbb{V}ar^{N} \left[ \psi\left(X^{\pi,\theta}\right)-\sum_{k=0}^{N_{\text{steps}}-1} \bar{\mathfrak h}(t_k, (X_{t_k\wedge t_j}^{\pi,\theta})_{j=0}^{N_{\text{steps}}}, \xi_{\Psi})\Delta \tilde{\bar S}^{\pi,\theta}_{t_{k}} \right] \,.
\]
The histogram in Figure~\ref{fig hedge error} is calculated on $N=400\, 000$ different paths and provides the values of $s$,
\[
s:=\Psi\left(X^{\pi,\theta}\right)-\sum_{k=0}^{N_{\text{steps}}-1} \bar{\mathfrak h}(t_k, (X_{t_k\wedge t_j}^{\pi,\theta})_{j=0}^{N_{\text{steps}}}, \xi_{\Psi})\Delta \tilde{\bar S}^{\pi,\theta}_{t_{k}} - \mathbb E^N\left[\Psi\left(X^{\pi,\theta}\right)\right]
\]
i.e. such that 
\[
\mathbb E^N[s^2] = \mathbb{V}ar^{N} \left[ \Psi\left(X^{\pi,\theta}\right)-\sum_{k=0}^{N_{\text{steps}}-1} \bar{\mathfrak h}(t_k, (X_{t_k\wedge t_j}^{\pi,\theta})_{j=0}^{N_{\text{steps}}}, \xi_{\Psi})\Delta \tilde{\bar S}^{\pi,\theta}_{t_{k}} \right]. 
\]
We obtain $\mathbb E^N[s^2] = 1.6\times 10^{-3}$.
%\begin{equation}\label{eq loss cv real}
%\bar \xi^{\ast} \in  \argmin_{\bar \xi} \mathbb V\text{ar}\left[\phi((X_t)_{t\in[0,T]}) -  \int_0^T \bar {\mathfrak h}(r,(X_{r\wedge t})_{t\in[0,T]},\bar \xi) \,d \bar S^\theta_r \bigg | \mathcal F_0\right]
%\end{equation}
%
%
%\[
%\mathcal E = \Phi - \mathbb E[\Phi | \mathcal F_0] - \int_0^T Z_s \, dW_s\
%\]

\begin{figure}[h]\label{fig hedge error}
\centering
\includegraphics[clip,width=0.6\textwidth]{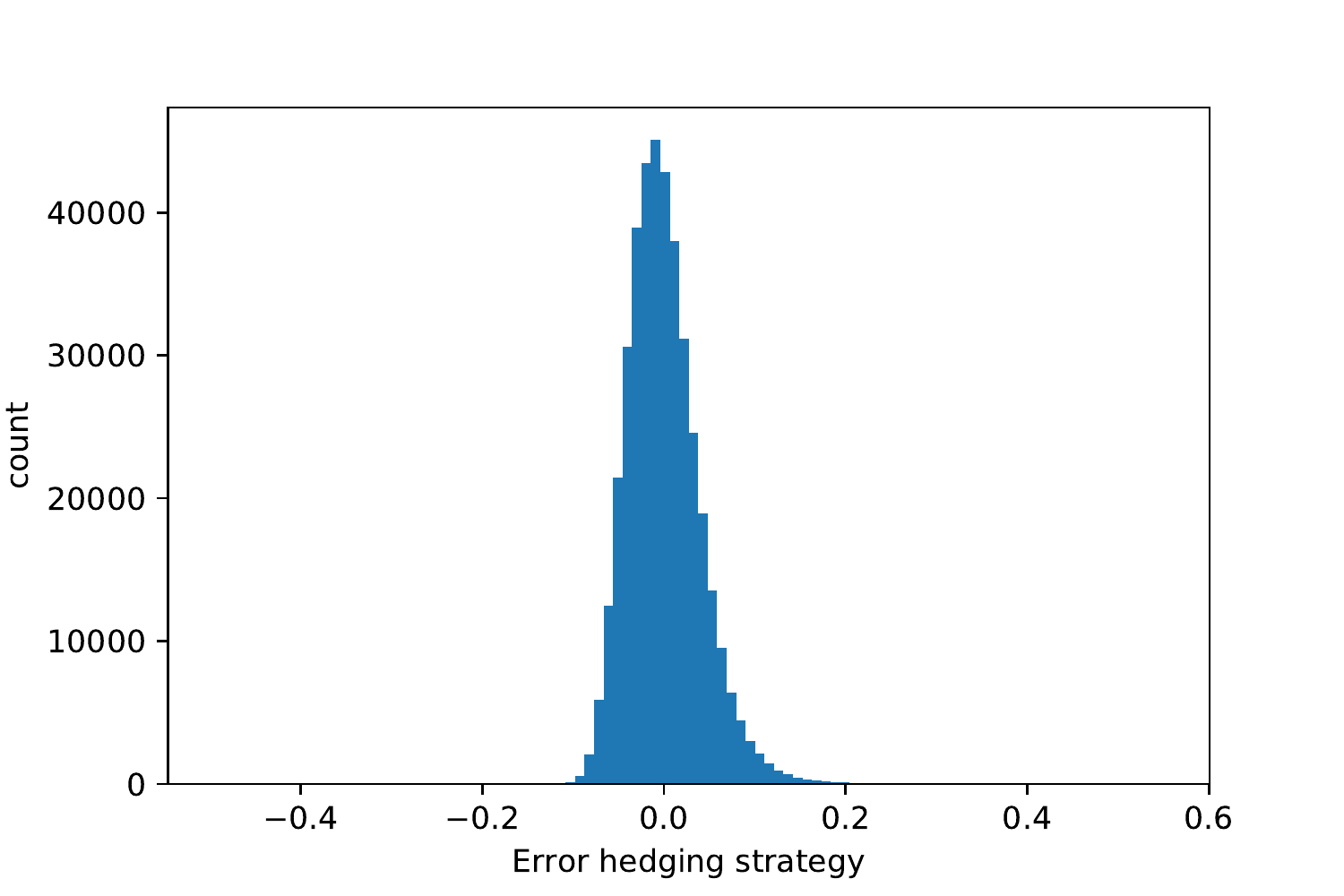}
\caption{Error of the portfolio hedging strategy for the lookback option}	
\end{figure}

Finally, we study the effect of the control variate parametrisation on the learning speed
in Algorithm~\ref{alg LSV calibration vanilla}.
Figure~\ref{fig test cv} displays the evolution of the Root Mean Squared Error of two runs of 
calibration to market vanilla option prices for two-months maturity: the blue line using Algorithm~\ref{alg LSV calibration vanilla} with simultaneous learning of the hedging strategy, and the orange line without the hedging 
strategy. 
We recall that from Section~\ref{sec many maturities}, the Monte Carlo estimator $\partial_{\theta}h^{N}(\theta)$ is a biased estimator of 
$\partial_{\theta}h(\theta)$
An upper bound of the bias is given by Corollary~\ref{eq cor sq loss bias}, that shows
that by reducing the variance of Monte Carlo estimator of the option price then the bias of $\partial_{\theta}h^{N}(\theta)$ is also reduced, yielding better convergence behaviour of the 
stochastic approximation algorithm. 
This can be observed in Figure~\ref{fig test cv}. 
\begin{figure}[h]\label{fig test cv}
  \centering 
  \includegraphics[clip,width=0.6\textwidth]{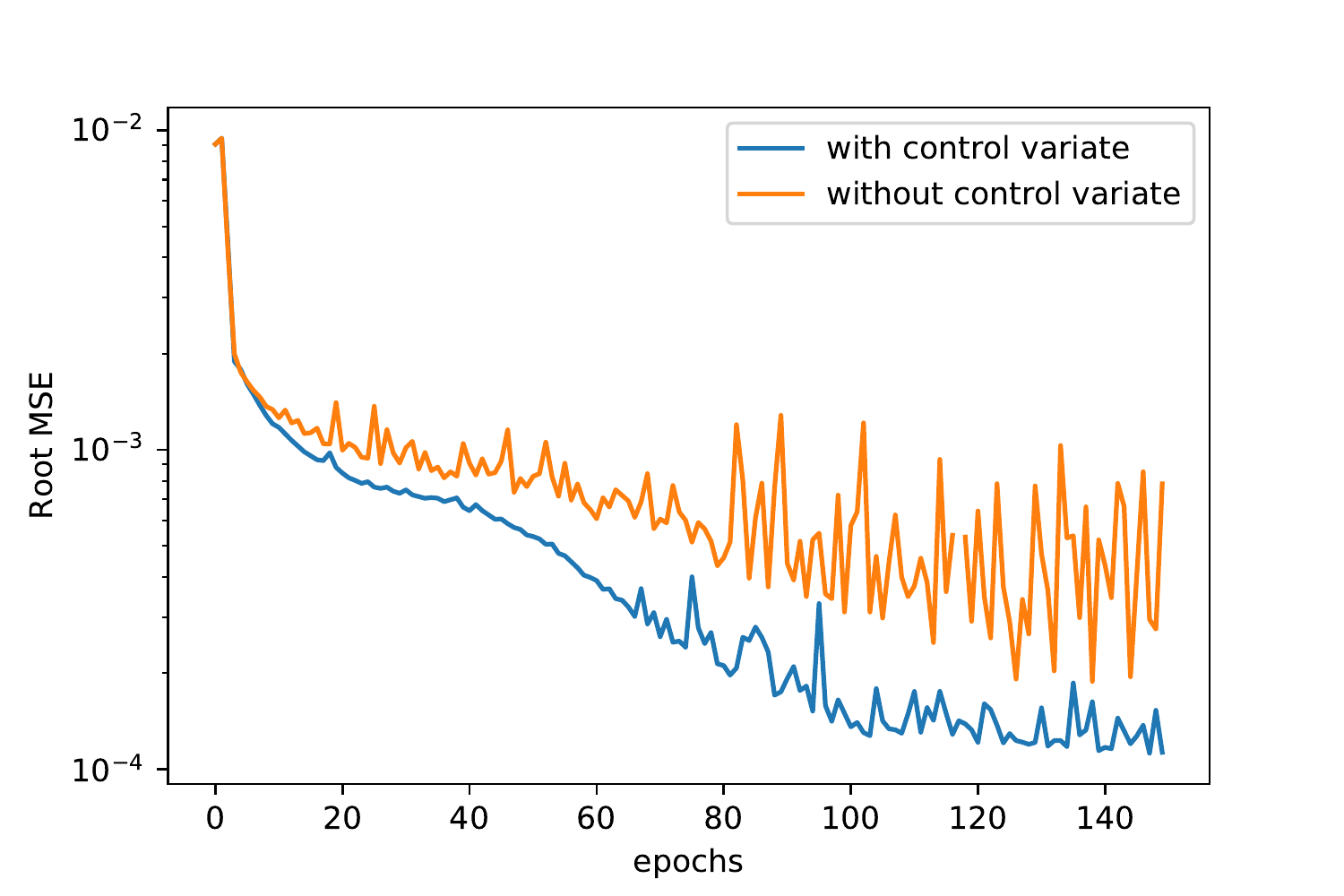}
  \caption{Root Mean Squared Error of calibration to Vanilla option prices with and without hedging strategy parametrisation}
\end{figure}

\section*{Acknowledgements}
This work was supported by the Alan Turing Institute under EPSRC grant no. EP/N510129/1.
We thank Antoine Jacquier (Imperial) for fruitful discussions on the topic of the paper.

\section*{Declarations of Interest}
The authors report no conflicts of interest. The authors alone are responsible for the content and writing of the paper.

%%%%%%%%%%%%%%%%%%%%%%%%%%%%%%%%%
%%%% \BEGIN BIBLIOGRAPHY
%%%%%%%%%%%%%%%%%%%%%%%%%%%%%%%%%%%%%%%%%%%%%%%%%%%%\p%%%%%%%%%%%%%%%%%%%%%

\raggedright
\bibliographystyle{apalike}
%wmaainf}%alpha}%plain}

\bibliography{Bibliography,Particles}  % ``name``.bib is the name of thedatabase
%\printbibliography

%%%% \END BIBLIOGRAPHY
%%%%%%%%%%%%%%%%%%%%%%%%%%%%%%%%%%%%%%%%%%%%%%%%%%%%%%%%%%%%%%%%%%%%%%%%%%

\appendix

\section{Bound on bias in gradient descent}
\label{sec grad des bias}

We complete the analysis from Section~\ref{sec stoch alg for calibration} for a general loss function here.

\begin{theorem}\label{th bias}
Let Assumption \ref{ass diff payoff} hold. Consider the family of neural SDEs \eqref{eq:nsde}. We have 
\begin{equation}\label{eq bias 1}
\begin{split}
	| \mathbb E[\partial_{\theta}h^{N}(\theta) ] - \partial_{\theta}h(\theta) |
 \leq & \left( \mathbb E\left[ (\partial_x \ell(\mathbb E^{\mathbb Q^N(\theta)}[\Phi^{cv}(X^{\theta})],\mathfrak  p(\Phi) )- \partial_x \ell(\mathbb E^{\mathbb Q}[\Phi^{cv}(X^{\theta})],\mathfrak  p(\Phi) ) )^2 \right] \right)^{1/2} \\
 & \times \left( \mathbb E \left[\left(\mathbb E^{\mathbb Q^N(\theta)}[\partial_\theta \Phi(X^{\theta})]\right)^2\right] \right)^{1/2}\,.
 \end{split}
\end{equation}
If in addition we assume that the loss function $\ell$ is three times differentiable in the  first variable with all its derivatives bounded, then 
\begin{equation}\label{eq bias 2}
\begin{split}
 	&\left| \mathbb E^{\mathbb Q}\left[\partial_{\theta} h^{N}(\theta)\right] -  \partial_{\theta}h(\theta) \right|\\ 
 	 & \leq \frac{1}{2}  \Big\{ \| \partial_x^3 \ell\|_{\infty}  | \mathbb E^{\mathbb Q}[\partial_\theta \Phi(X^{\theta})]| 
 \frac{1}{N}\mathbb Var^{\mathbb Q}[\Phi^{cv}(X^{\theta})] \\
 	 & + \| \partial_x^3 \ell\|_{\infty}
 	  \left( \frac{1}{N}\mathbb Var^{\mathbb Q}[\partial_\theta \Phi(X^{\theta})]\right)^{1/2}
 	   \left( \frac{1}{N^3} \mathbb E[(\Phi^{cv}(X^{\theta}) - \mathbb E^{\mathbb Q}[\Phi^{cv}(X^{\theta})])^4]  +  \frac{3}{N^2} (\mathbb Var^{\mathbb Q}[\Phi^{cv}(X^{\theta})])^2\right)^{1/2} \\
 & + 2 \| \partial_x^2 \ell\|_{\infty}  \left( \frac{1}{N} \mathbb Var^{\mathbb Q}[\partial_{\theta}\Phi(X^{\theta})]   \right)^{1/2}\left( \frac{1}{N} \mathbb Var^{\mathbb Q}[\Phi^{cv}(X^{\theta})] \right)^{1/2}\Big\}\,.
 \end{split}
\end{equation}

\end{theorem}
\begin{proof}
Observe that
\[
\mathbb E \left[\mathbb E^{\mathbb Q^N}[\Phi^{cv}(X^{\theta})]\right] = \mathbb E^{\mathbb Q}[\Phi^{cv}(X^{\theta})] \quad \text{and} \quad \mathbb E \left[ \mathbb E^{\mathbb Q^N}[\partial_\theta \Phi(X^{\theta})] \right] =\mathbb E^{\mathbb Q}[\partial_\theta \phi(X^{\theta})]\,.
\]  
Next, by adding and subtracting $\partial_x \ell(\mathbb E^{\mathbb Q}[\Phi^{cv}(X^{\theta})],\mathfrak  p(\Phi) )$ and using the Cauchy--Schwarz inequality we have
\[
\begin{split}
& | \mathbb E[\partial_{\theta}h^{N}(\theta) ] - \partial_{\theta}h(\theta) |\\
	& =  \left| \mathbb E\left[\left(\partial_x \ell\left(\mathbb E^{\mathbb Q^N}[\Phi^{cv}(X^{\theta})],\mathfrak  p(\Phi) \right) \pm \partial_x \ell\left(\mathbb E^{\mathbb Q}[\Phi^{cv}(X^{\theta})],\mathfrak  p(\Phi) \right) \right)\mathbb E^{\mathbb Q^N}\left[\partial_\theta \Phi(X^{\theta})\right] \right]
 - \partial_{\theta}h(\theta) \right|\,. \\
\end{split}
\]
Hence
\[
\begin{split}  
& | \mathbb E[\partial_{\theta}h^{N}(\theta) ] - \partial_{\theta}h(\theta) |\\
& = \left|\mathbb E\left[\left(\partial_x \ell\left(\mathbb E^{\mathbb Q^N}[\Phi^{cv}(X^{\theta})],\mathfrak  p(\Phi) \right)- \partial_x \ell\left(\mathbb E^{\mathbb Q}[\Phi^{cv}(X^{\theta})],\mathfrak  p(\Phi) \right) \right)\mathbb E^{\mathbb Q^N}\left[\partial_\theta \Phi(X^{\theta})\right] \right]  \right| \\
 & \leq \left( \mathbb E\left[ \left(\partial_x \ell\left(\mathbb E^{\mathbb Q^N}[\Phi^{cv}(X^{\theta})],\mathfrak  p(\Phi) \right)- \partial_x \ell\left(\mathbb E^{\mathbb Q}[\Phi^{cv}(X^{\theta})],\mathfrak  p(\Phi) \right) \right)^2 \right] \right)^{1/2} \left( \mathbb E \left[\mathbb E^{\mathbb Q^N}[\partial_\theta \Phi(X^{\theta})]\right]^2 \right)^{1/2}\,.
 \end{split}
\]
This concludes the proof of \eqref{eq bias 1}. 
To prove \eqref{eq bias 2}, we view $\partial_{\theta}h^{N}(\theta)$ as function of $(\mathbb E^{\mathbb Q^N}[\Phi^{cv}(X^{\theta})], \allowbreak \mathbb E^{\mathbb Q^N}[\partial_\theta \Phi(X^{\theta})])$ and expand into its Taylor series around $(\mathbb E^{\mathbb Q}[\Phi^{cv}(X^{\theta})], \mathbb E^{\mathbb Q}[\partial_\theta \Phi(X^{\theta})])$, i.e
 \[
 \begin{split}
 	&\partial_{\theta}h^{N}(\theta) =  \partial_{\theta}h(\theta) \\ 
 	& + \partial_x^2 \ell\left(\mathbb E^{\mathbb Q}[\Phi^{cv}(X^{\theta})],\mathfrak  p(\Phi) \right) \mathbb E^{\mathbb Q}[\partial_\theta \Phi(X^{\theta})]\left(\mathbb E^{\mathbb Q^N}[\Phi^{cv}(X^{\theta})]  - \mathbb E^{\mathbb Q}[\Phi^{cv}(X^{\theta})]  \right)  \\
 & 	+ \partial_x \ell\left(\mathbb E^{\mathbb Q}[\Phi^{cv}(X^{\theta})],\mathfrak  p(\Phi) \right) \left( \mathbb E^{\mathbb Q^N}[\partial_\theta \Phi(X^{\theta})]  - \mathbb E^{\mathbb Q}[\partial_\theta \Phi(X^{\theta})]  \right) \\
 & + \frac{1}{2} \int_{0}^1 \Big\{\partial_x^3 \ell\left(\xi_1^{\alpha},\mathfrak  p(\Phi) \right) \xi_2^{\alpha}\left(\mathbb E^{\mathbb Q^N}[\Phi^{cv}(X^{\theta})]  - \mathbb E^{\mathbb Q}[\Phi^{cv}(X^{\theta})]  \right)^2 \\
 & \qquad + 2 \partial_x^2 \ell\left(\xi_1^{\alpha},\mathfrak  p(\Phi) \right) \left(\mathbb E^{\mathbb Q^N}[\partial_{\theta}\Phi(X^{\theta})]  - \mathbb E^{\mathbb Q}[\partial_{\theta}\Phi(X^{\theta})]  \right) \left(\mathbb E^{\mathbb Q^N}[\Phi^{cv}(X^{\theta})]  - \mathbb E^{\mathbb Q}[\Phi^{cv}(X^{\theta})]  \right)\Big\} d\alpha\,,
 \end{split}
 \] 
where 
 \[
 \begin{split}
 \xi_1^{\alpha} = &  \mathbb E^{\mathbb Q}[\Phi^{cv}(X^{\theta})]  + \alpha \left(
 \mathbb E^{\mathbb Q^N}[\Phi^{cv}(X^{\theta})] - \mathbb E^{\mathbb Q}[\Phi^{cv}(X^{\theta})] \right)\,, \\ 
  \xi_2^{\alpha} = &  \mathbb E^{\mathbb Q}[\partial_\theta \Phi(X^{\theta})] + \alpha \left( \mathbb E^{\mathbb Q^N}[\partial_\theta \Phi(X^{\theta})] - \mathbb E^{\mathbb Q}[\partial_\theta \Phi(X^{\theta})] \right)\,.
 \end{split}
 \]
 Hence, using Cauchy-Schwarz inequality 
 \[
 \begin{split}
 	&\left| \mathbb E^{\mathbb Q}\left[\partial_{\theta} h^{N}(\theta)\right] -  \partial_{\theta}h(\theta) \right|\\ 
 	 & \leq \frac{1}{2} \int_{0}^1 \mathbb E \Big[ \Big\{ \| \partial_x^3 \ell\|_{\infty} | \mathbb E^{\mathbb Q}[\partial_\theta \Phi(X^{\theta})]| \left(\mathbb E^{\mathbb Q^N}[\Phi^{cv}(X^{\theta})]  - \mathbb E^{\mathbb Q}[\Phi^{cv}(X^{\theta})]  \right)^2 \\
 	 & + \alpha\| \partial_x^3 \ell\|_{\infty}
 	 \left| \mathbb E^{\mathbb Q^N}[\partial_\theta \Phi(X^{\theta})] - \mathbb E^{\mathbb Q}[\partial_\theta \Phi(X^{\theta})] \right|
 	   \left(\mathbb E^{\mathbb Q^N}[\Phi^{cv}(X^{\theta})]  - \mathbb E^{\mathbb Q}[\Phi^{cv}(X^{\theta})]  \right)^2 \\
 & + 2 \| \partial_x^2 \ell\|_{\infty}  \left| \mathbb E^{\mathbb Q^N}[\partial_{\theta}\Phi(X^{\theta})]  - \mathbb E^{\mathbb Q}[\partial_{\theta}\Phi(X^{\theta})]  \right|\left|\mathbb E^{\mathbb Q^N}[\Phi^{cv}(X^{\theta})]  - \mathbb E^{\mathbb Q}[\Phi^{cv}(X^{\theta})] \right|\Big\} \Big]d\alpha\, \\
& \leq \frac{1}{2}  \Big\{ \| \partial_x^3 \ell\|_{\infty}  | \mathbb E^{\mathbb Q}[\partial_\theta \Phi(X^{\theta})]|
 \frac{1}{N}\mathbb Var^{\mathbb Q}[\Phi^{cv}(X^{\theta})] \\
 	 & + \| \partial_x^3 \ell\|_{\infty}
 	  \left( \frac{1}{N}\mathbb Var^{\mathbb Q}[\partial_\theta \Phi(X^{\theta})]\right)^{1/2}
 	   \left( \mathbb E \left[ \left(\mathbb E^{\mathbb Q^N}[\Phi^{cv}(X^{\theta})]  - \mathbb E^{\mathbb Q}[\Phi^{cv}(X^{\theta})]  \right)^4 \right] \right)^{1/2} \\
 & + 2 \| \partial_x^2 \ell\|_{\infty}  \left( \frac{1}{N} \mathbb Var^{\mathbb Q}[\partial_{\theta}\Phi(X^{\theta})]   \right)^{1/2}\left( \frac{1}{N} \mathbb Var^{\mathbb Q}[\Phi^{cv}(X^{\theta})] \right)^{1/2}\Big\}\,.
 \end{split}
 \] 
 Now let  $\lambda^i:= \Phi^{cv}(X^{\theta,i}) - \mathbb E^{\mathbb Q}[\Phi^{cv}(X^{\theta})] $, and note that
 \[
 \begin{split}
\left( \sum_{i=1}^N \lambda^i \right)^4 
 =& \sum_{i=1}^N (\lambda^i )^4 + 3 \sum_{i_1\neq i_2}^N (\lambda^{i_1} )^2 (\lambda^{i_2} )^2
+ 4 \sum_{i_1\neq i_2 }^N (\lambda^{i_1} )^1 (\lambda^{i_2} )^3 \\
&+ 6 \sum_{i_1, i_2, i_3\,\,\,  \text{distinct} }^N \lambda^{i_1} \lambda^{i_2} (\lambda^{i_3} )^2
+ \sum_{i_1, i_2, i_3,i_4 \,\,\,\text{distinct} }^N \lambda^{i_1} \lambda^{i_2}\lambda^{i_3}\lambda^{i_4}\,. 
\end{split}
 \]
 Hence 
 \[
 \mathbb E \left[ \left(\mathbb E^{\mathbb Q^N}[\Phi^{cv}(X^{\theta})]  - \mathbb E^{\mathbb Q}[\Phi^{cv}(X^{\theta})]  \right)^4 \right] 
 = \frac{1}{N^3} \mathbb E[(\Phi^{cv}(X^{\theta}) - \mathbb E^{\mathbb Q}[\Phi^{cv}(X^{\theta})])^4]  +  \frac{3}{N^2} (\mathbb Var^{\mathbb Q}[\Phi^{cv}(X^{\theta})])^2\,. 
  \]
 The proof is complete.  
\end{proof}

\section{Data used in calibration}\label{CalData}

We used Heston model to generate prices of calls and puts. The model is
\begin{eqnarray}\label{CalDataEq}
dX_t &=& r X_t dt + X_t \sqrt{V_t}dW_t, \quad X_0 = x_0 \\
dV_t &=& \kappa (\mu - V_t)dt + \eta \sqrt{V_t}dB_t, \quad V_0 = v_0 \\
d\langle B, W \rangle_t &=& \rho dt\,.
\end{eqnarray}
It is well know that for this model a semi-analytic formula can be used to calculate option prices, see~\cite{heston1997closed} but also~\cite{LHT}. 
The choice of parameters below was used to generate target model calibration prices.
\begin{equation}
\label{eq heston params}
x_0 = 1,~ r = 0.025\,, \,\,\, \kappa= 0.78\,, \,\,\, \mu = 0.11\,, \,\,\, \eta =0.68\,, \,\,\, V_0 =0.04\,, \,\,\, \rho = 0.044,
\end{equation}

\begin{figure}[h!tbp]
  \centering 
  \includegraphics[width=0.45\textwidth]{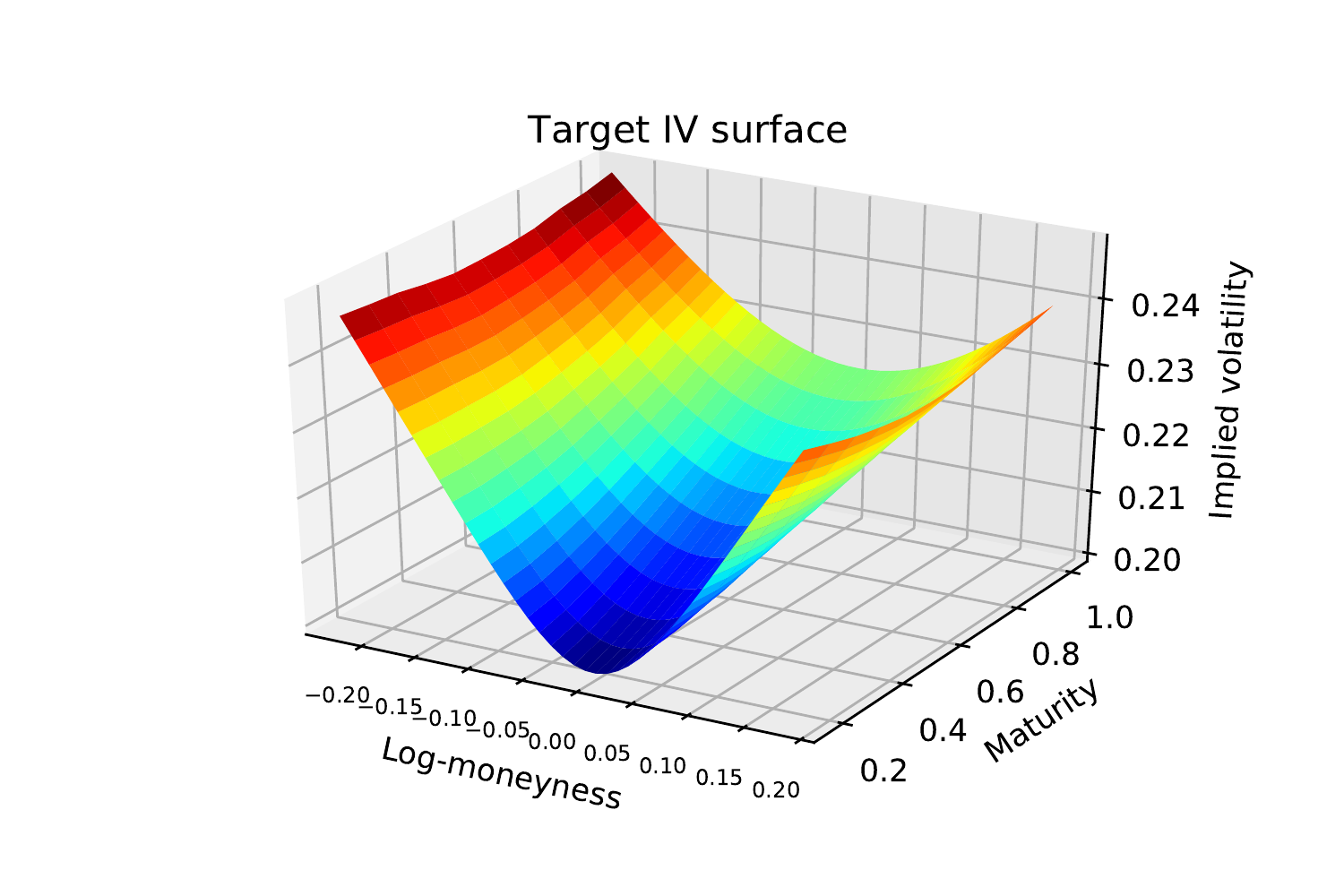}
\caption{The ``market'' data used in calibration of the Neural SDE models. 
In fact the implied volatility surface comes from~\eqref{CalDataEq} and~\eqref{eq heston params}.}
\label{TargetIVs}
\end{figure}

%\blue{Would be interesting to see the IV surface of an Neural SDE model.}
Options with bi-monthly maturities up to one year with varying range of strikes were used as market data for Neural SDE calibration.
The call / put option prices were obtained from the Heston model using Monte Carlo simulation with $10^7$ Brownian trajectories.
We use bimonthly maturities up to one year for considered calibrations. Varying range of strikes is used among different calibrations. 
See Figure~\ref{TargetIVs} for the resulting ``market'' data.

\section{Feed-forward neural networks}
\label{FFNNs}

Feed-forward neural networks are functions constructed by composition of affine map and non-linear activation function.
We fix a locally Lipschitz activation function $\mathbf a:\mathbb R \to \mathbb R$ as ReLU  function $a(z)=(0,z)_+$ and
for $d\in \mathbb N$ define $\mathbf A_d : \mathbb R^d \to \mathbb R^d$ as the function given, for $x=(x_1,\ldots,x_d)$ by 
$\mathbf A_d(x) = (\mathbf a(x_1),\ldots, \mathbf a(x_d))$.
We fix $L\in \mathbb N$ (the number of layers), $l_k \in \mathbb N$, $k=0,1,\ldots L-1$ (the size of input to layer $k$) and $l_L \in \mathbb N$ (the size of the network output). 
A fully connected artificial neural network is then given by $\theta = ((W_1,B_1), \ldots, (W_L, B_L))$,
where, for $k=1,\ldots,L$, we have real $l_{k-1}\times l_k$ matrices $W_k$ and real $l_k$ dimensional
vectors $B_k$. 	
%Neural network parameters $W_{k}$ and bias terms $B_k$ are sampled from uniform distribution $U[-\lambda,\lambda]$ with $\lambda=\frac{1}{\sqrt{l_{k-1}}}$. 
%We use ADAM optimiser with multistep learning rate scheduler and RMSE loss function. 

The artificial neural network defines a function $\mathcal R(\cdot, \theta) : \mathbb R^{l_0} \to \mathbb R^{l_L}$ given recursively, for $x_0 \in \mathbb R^{l_0}$, by 
\[
\mathcal R(x_0, \theta) = W_L x_{L-1} + B_L\,, \,\,\,\, 
x_k = \mathbf A_{l_k}(W_k x_{k-1} + B_k)\,,k=1,\ldots, L-1\,.
\]
%We can also define the function $\mathcal P$ which counts the number of parameters as 
%\[
%\mathcal P(\Phi) = \sum_{i=1}^L (l_{k-1}l_k + l_k )\,.
%\]
We will call such class of fully connected artificial neural networks $\mathcal{DN}$.
Note that since the activation functions and architecture are fixed the learning task entails
finding the optimal $\theta \in \mathbb R^{\mathcal P}$ where $p$ is the number of parameters in
$\theta$ given by
\[
\mathcal P(\theta) = \sum_{i=1}^L (l_{k-1}l_k + l_k )\,.
\]

\section{LV neural SDEs calibration accuracy}
\label{LVfitquality}

Figures~\ref{FigUnc}, \ref{FigLB}, \ref{FigUB} present implied volatility fit of local volatility neural SDE model~\eqref{NeuralLV} calibrated to: market vanilla data only; market vanilla data with lower bound constraint on lookback option payoff; market vanilla data with upper bound constraint on lookback option payoff respectively. 
Figures~\ref{FigUncprice}, \ref{FigLBprice}, \ref{FigUBprice} present target option price fit of local volatility neural SDE model~\eqref{NeuralLV} calibrated to: market vanilla data only; market vanilla data with lower bound constraint on lookback option payoff; market vanilla data with upper bound constraint on lookback option payoff respectively. 
High level of accuracy in all calibrations is achieved due to the hedging neural network incorporated into model training.

\begin{figure}[h]
\centering 
\includegraphics[clip,width=0.3\textwidth]{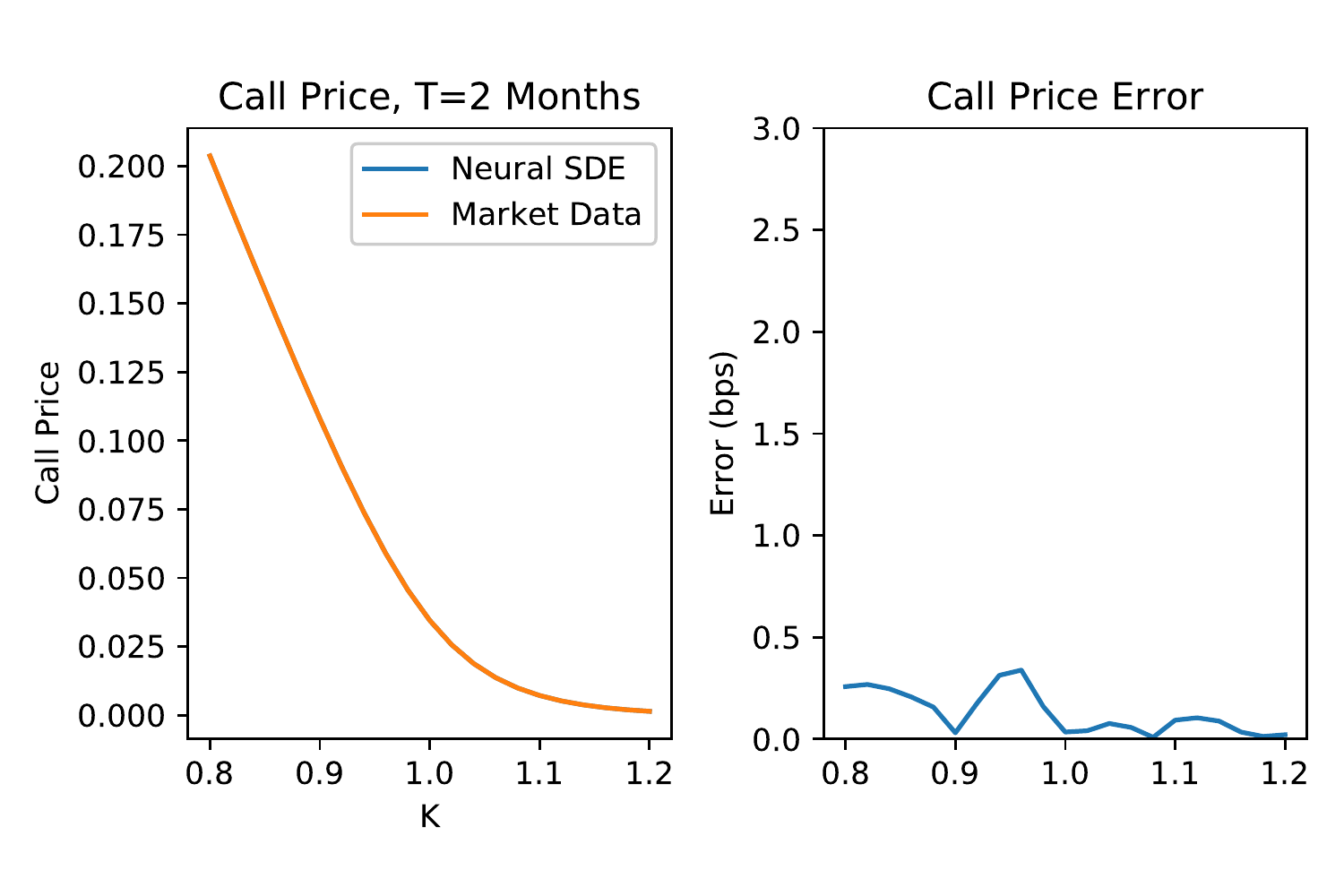}
\includegraphics[clip,width=0.3\textwidth]{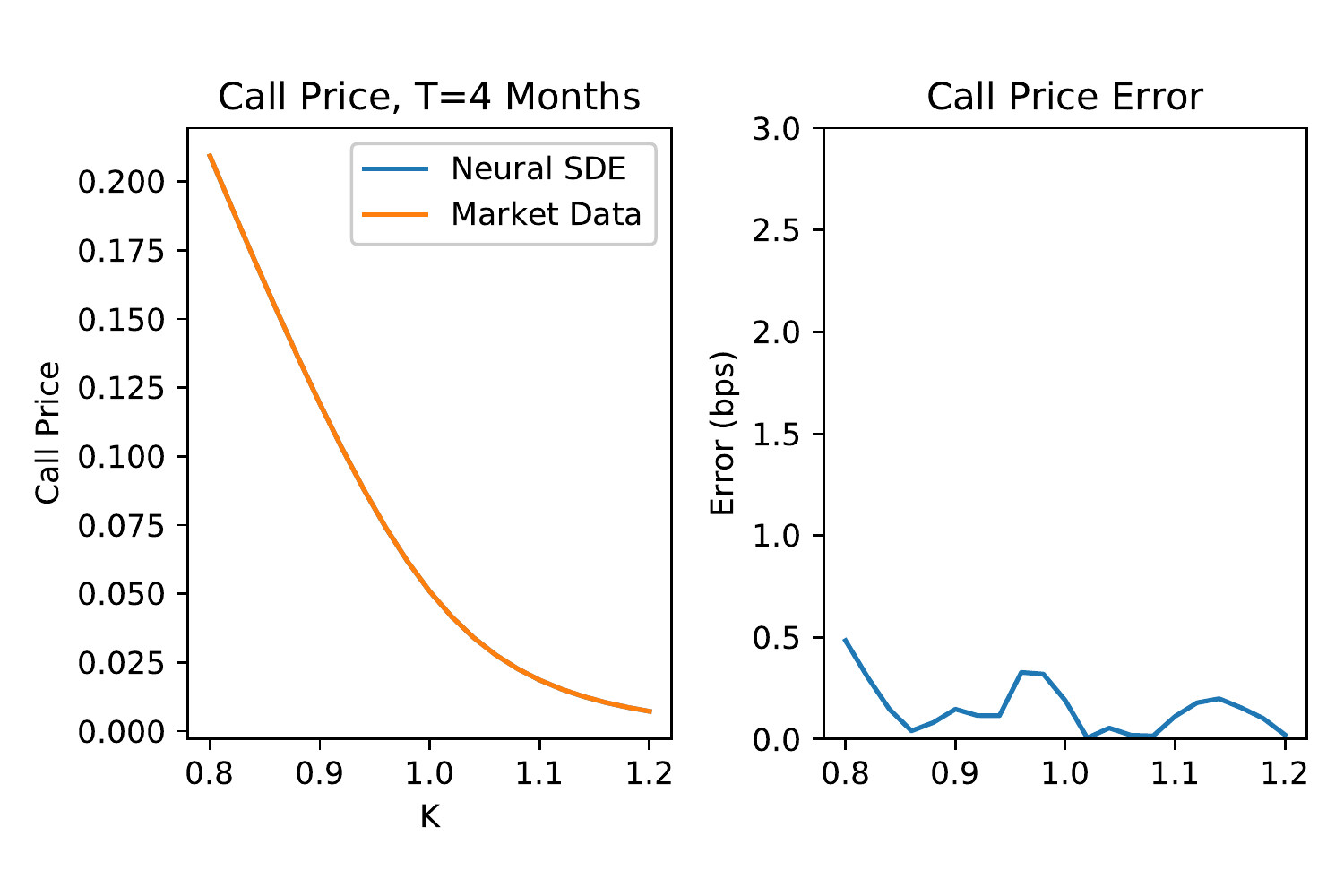}
\includegraphics[clip,width=0.3\textwidth]{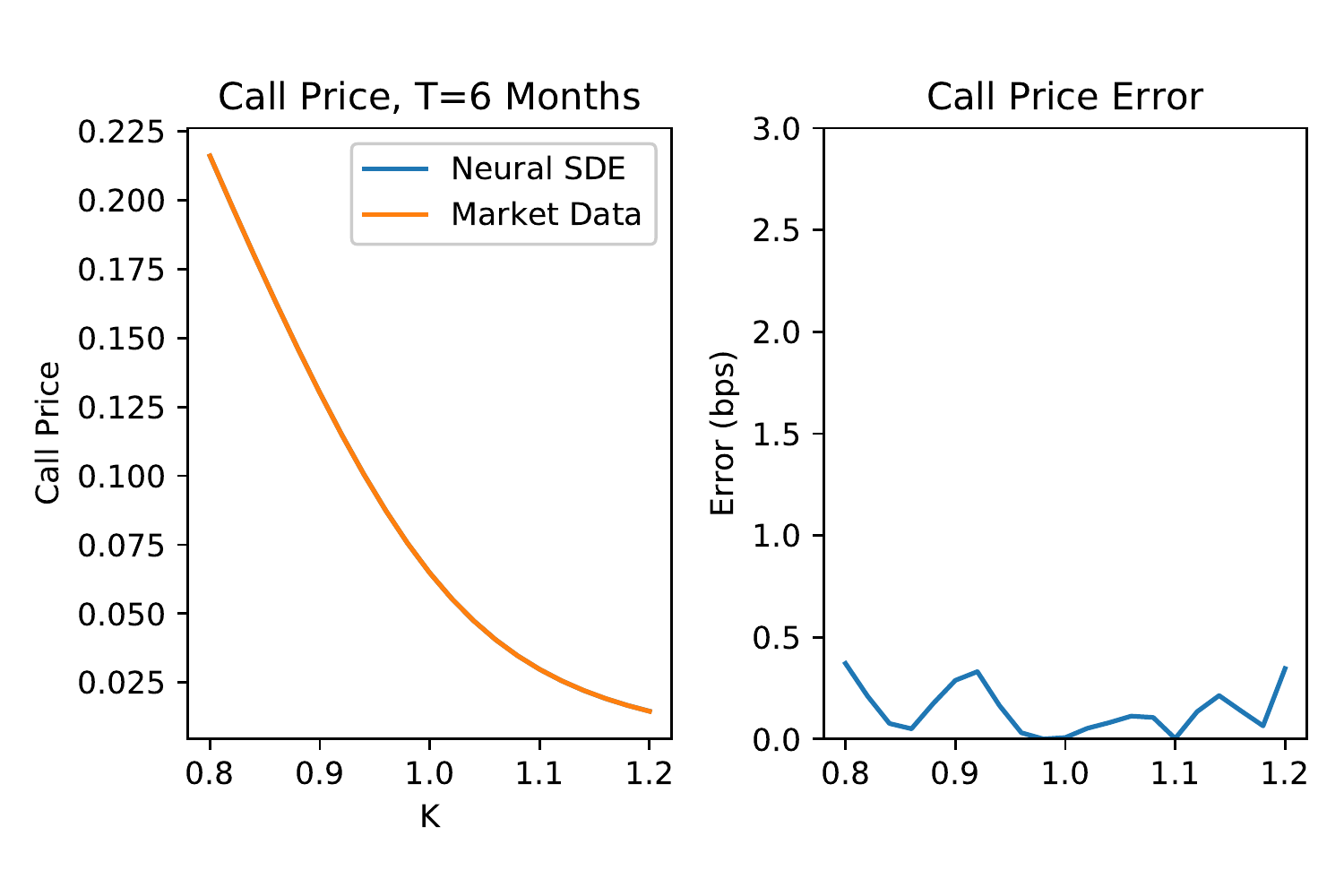} 
\includegraphics[clip,width=0.3\textwidth]{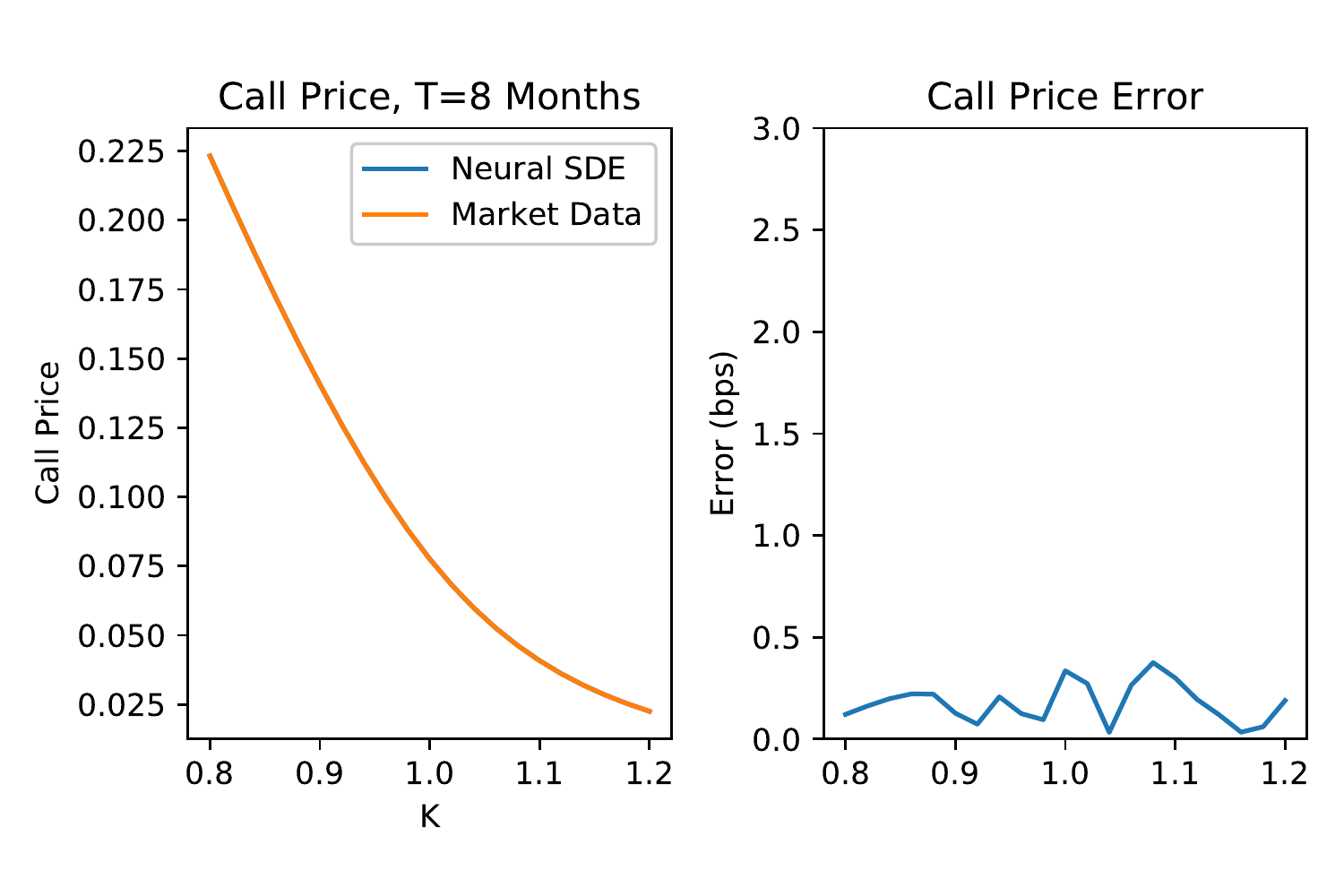}
\includegraphics[clip,width=0.3\textwidth]{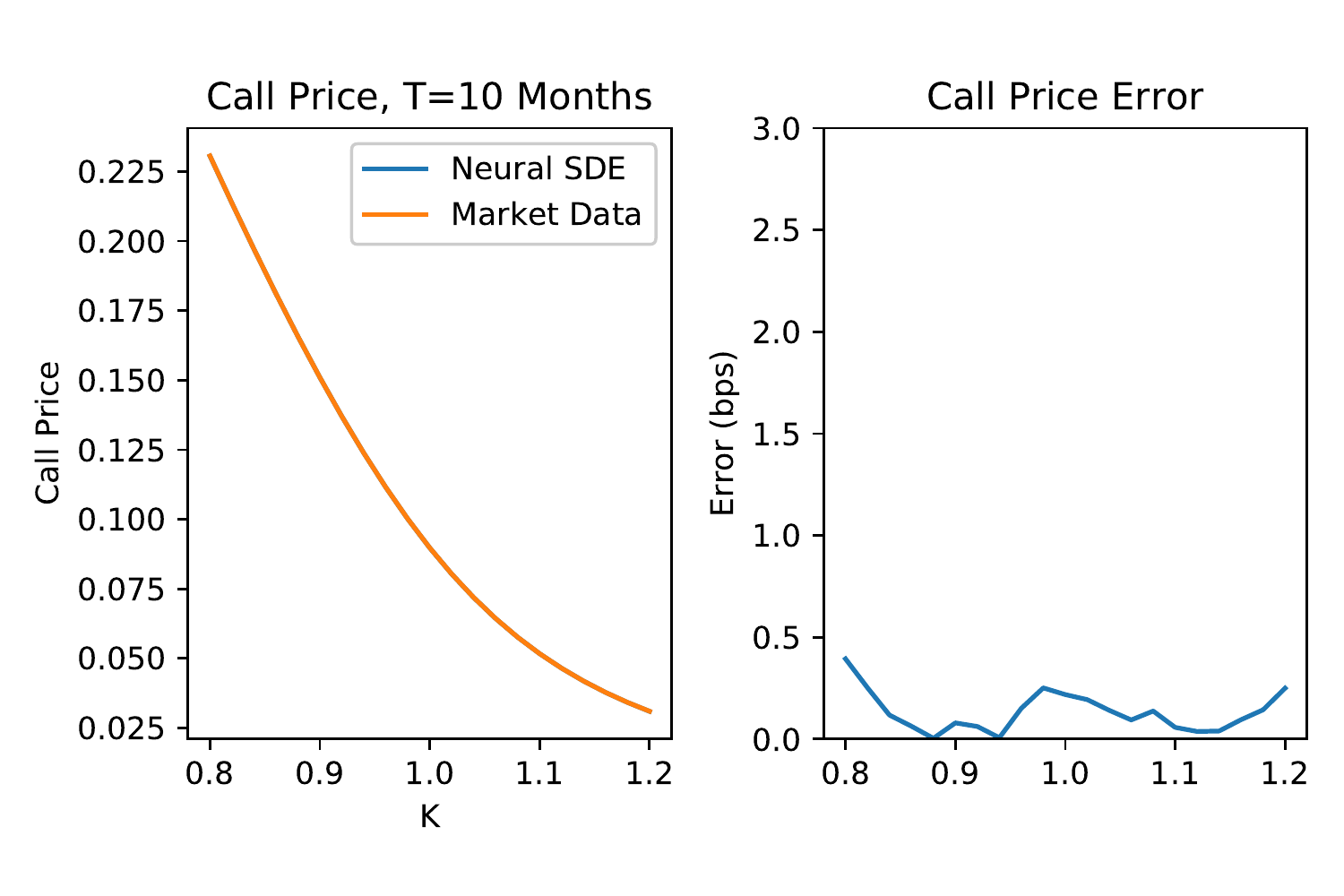} 
\includegraphics[clip,width=0.3\textwidth]{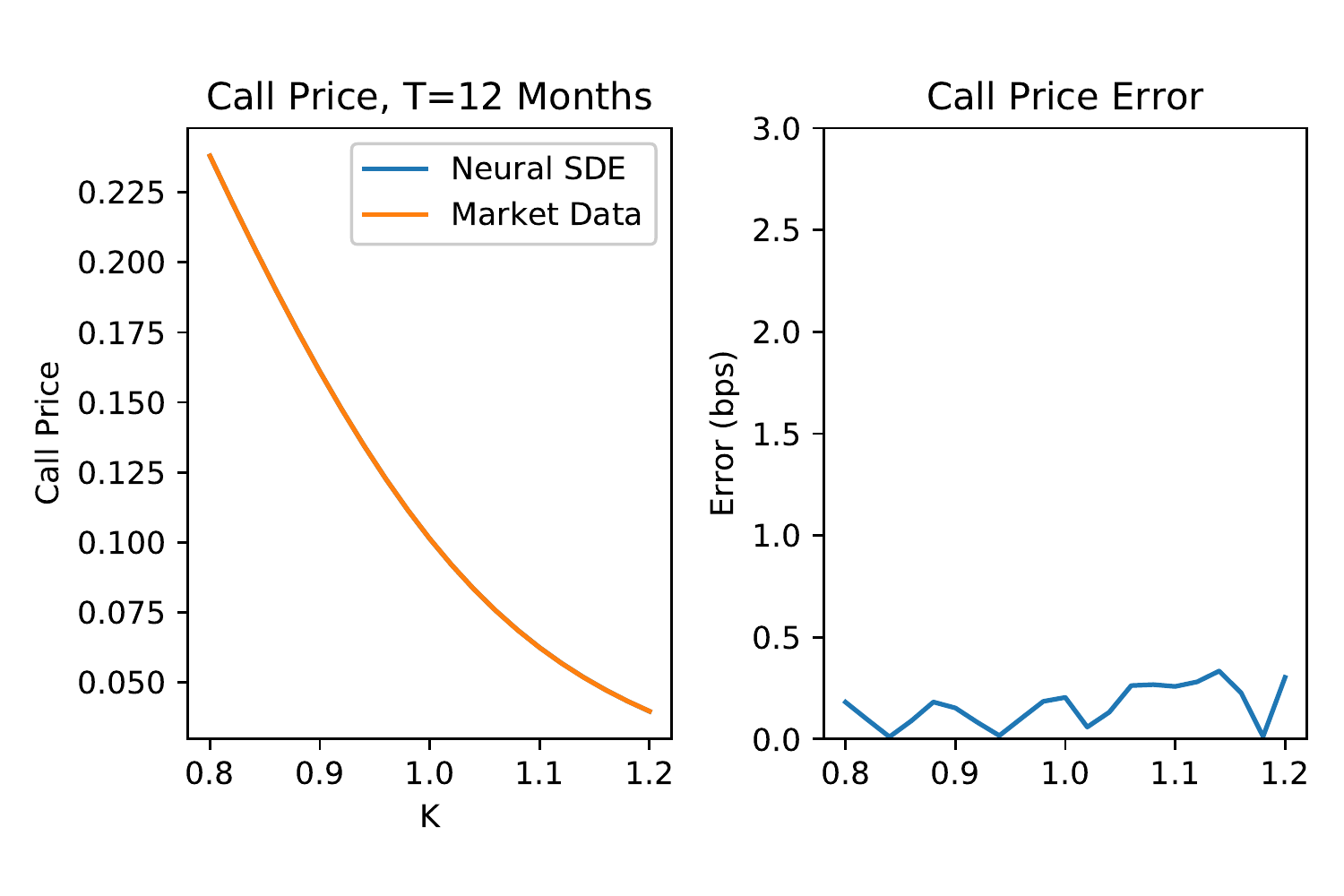}
\caption{Calibrated neural SDE LV model and market target prices comparison.}
\label{FigUncprice}
\end{figure}

\begin{figure}[h]
\centering 
\includegraphics[clip,width=0.3\textwidth]{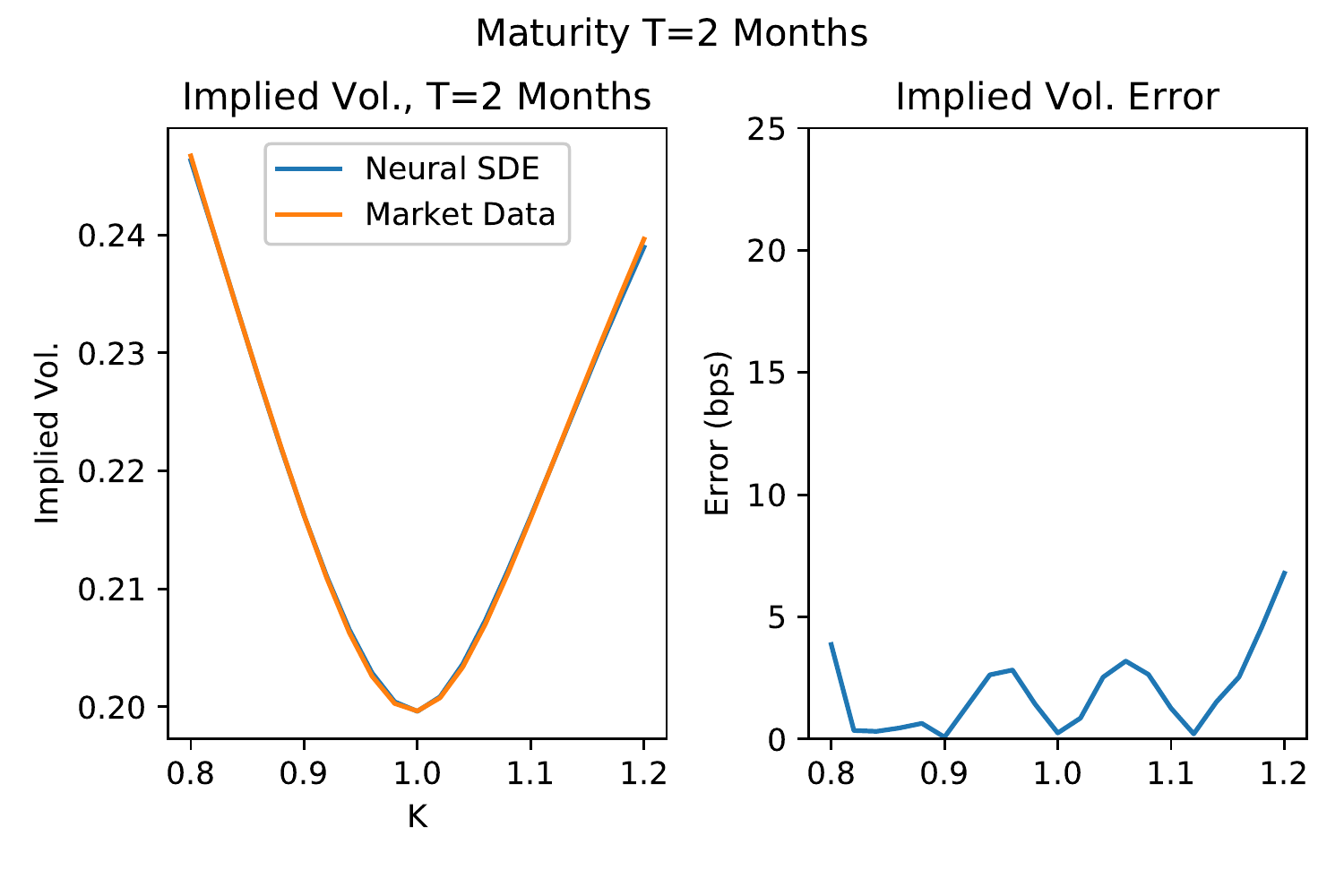}
\includegraphics[clip,width=0.3\textwidth]{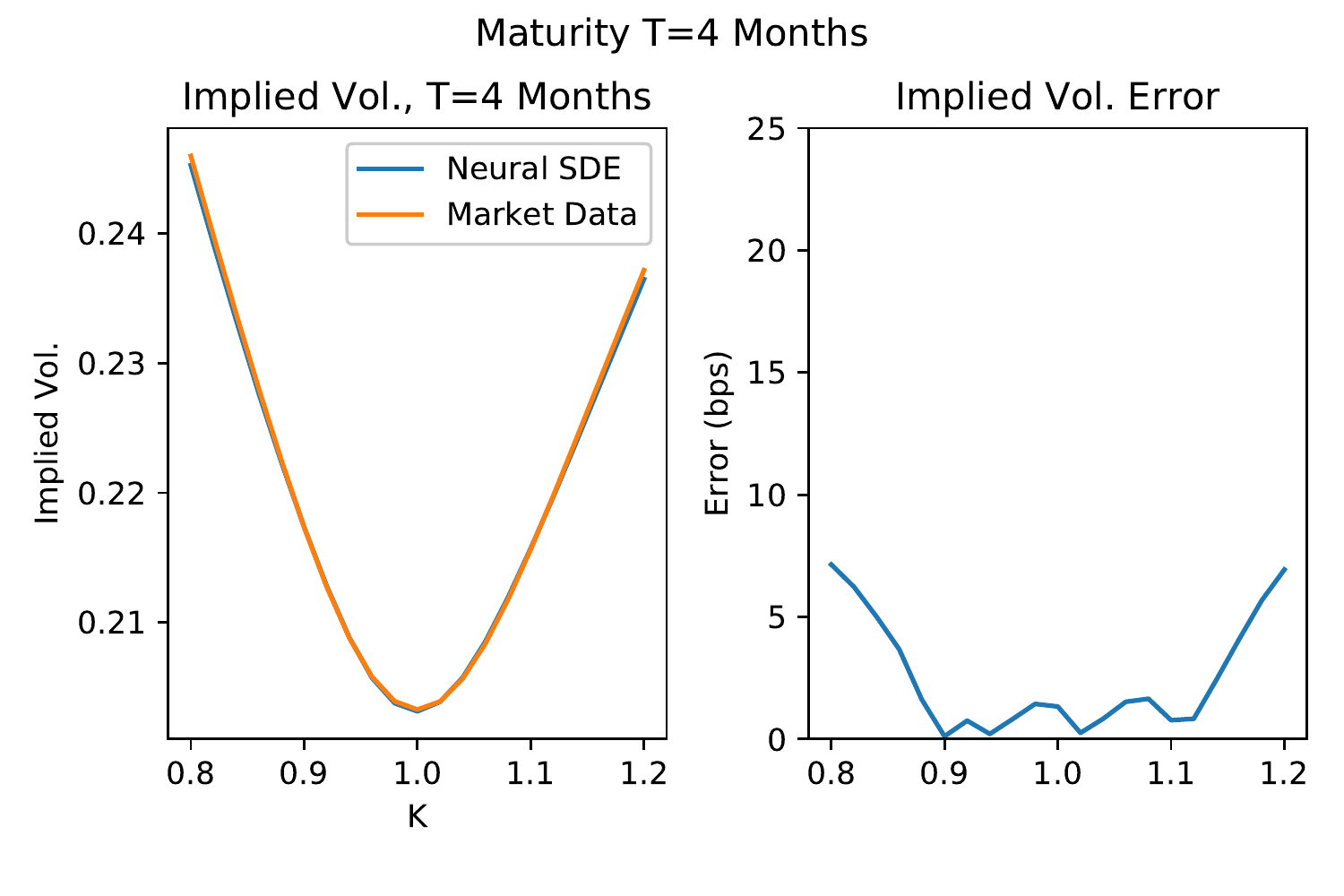}
\includegraphics[clip,width=0.3\textwidth]{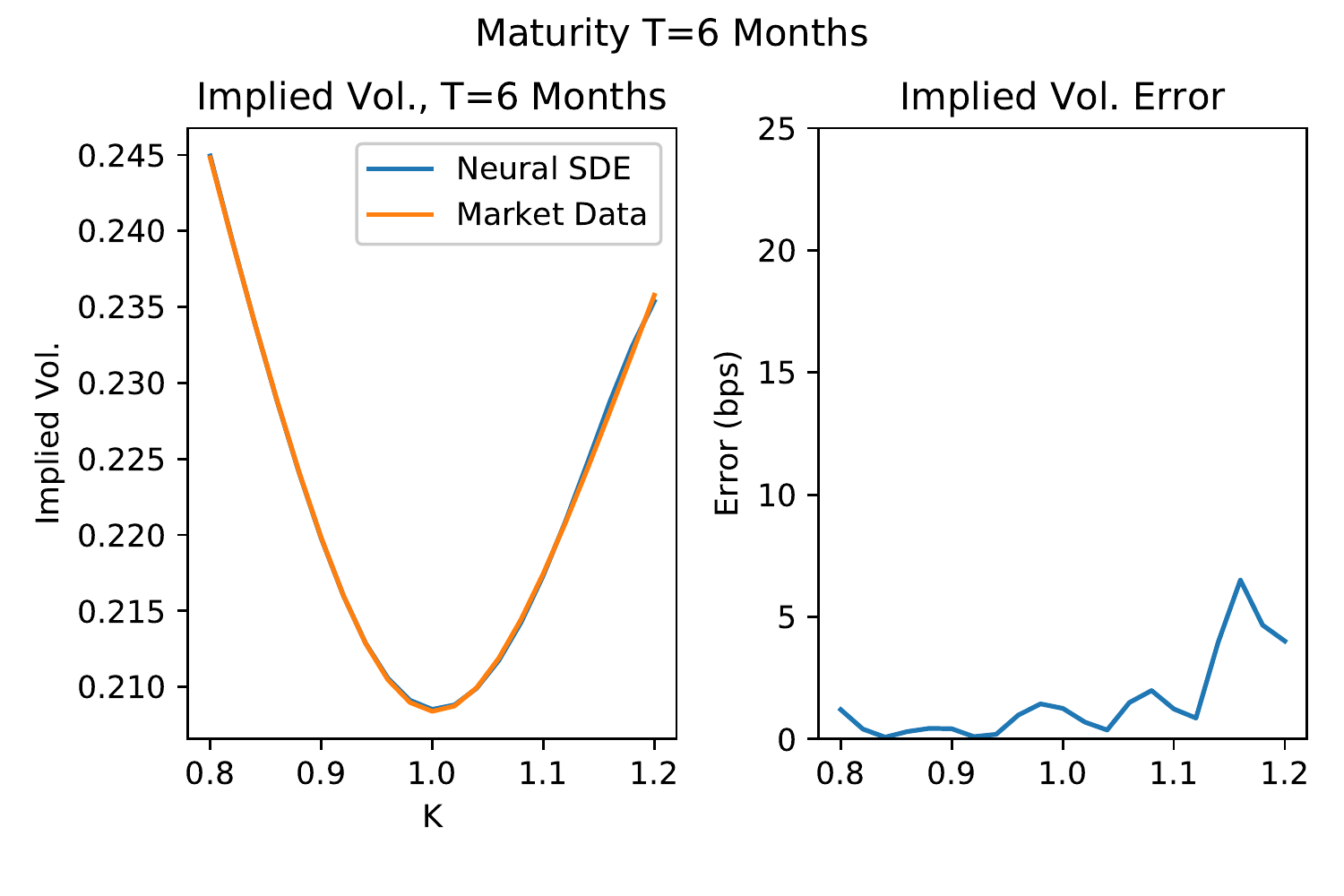} 
\includegraphics[clip,width=0.3\textwidth]{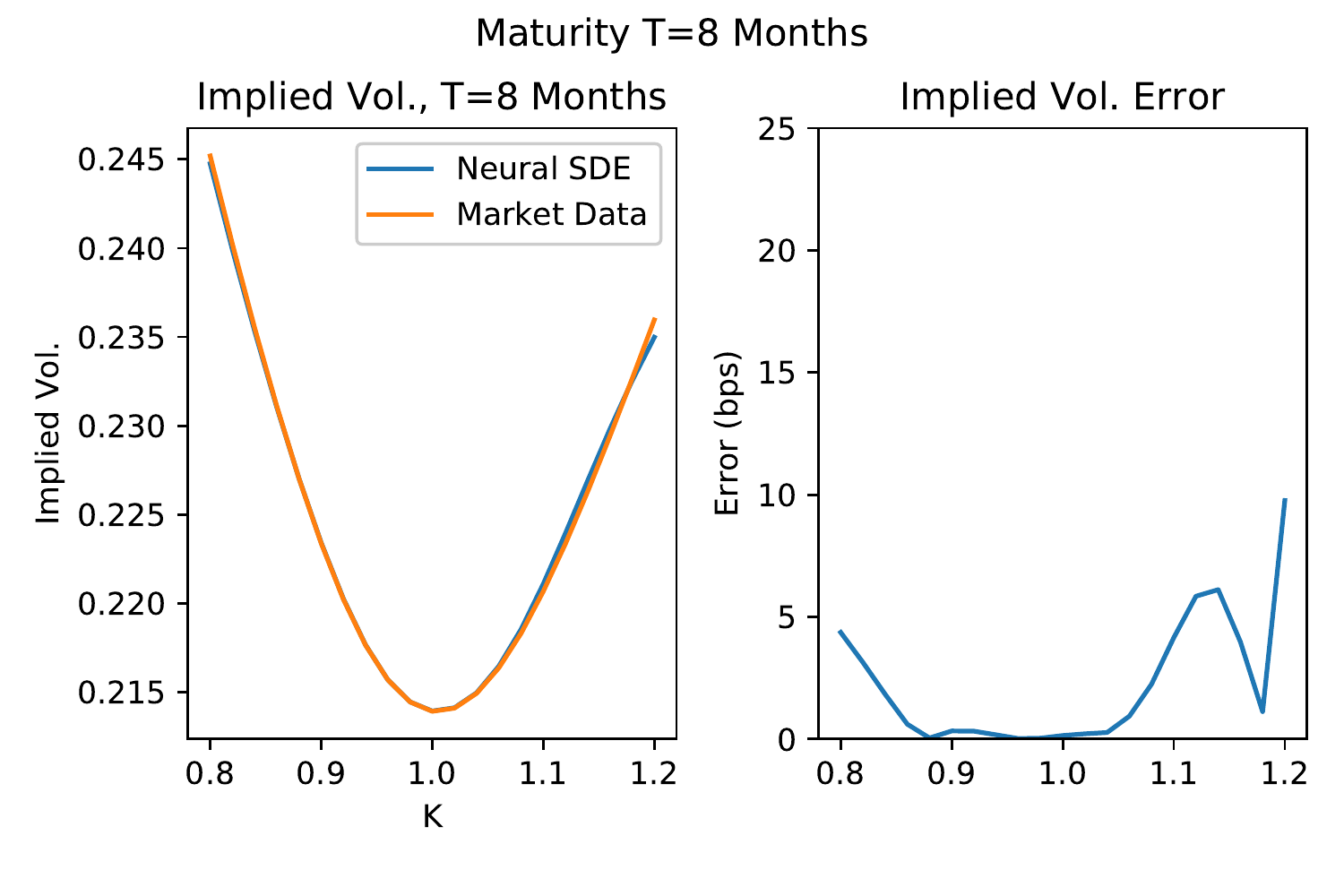}
\includegraphics[clip,width=0.3\textwidth]{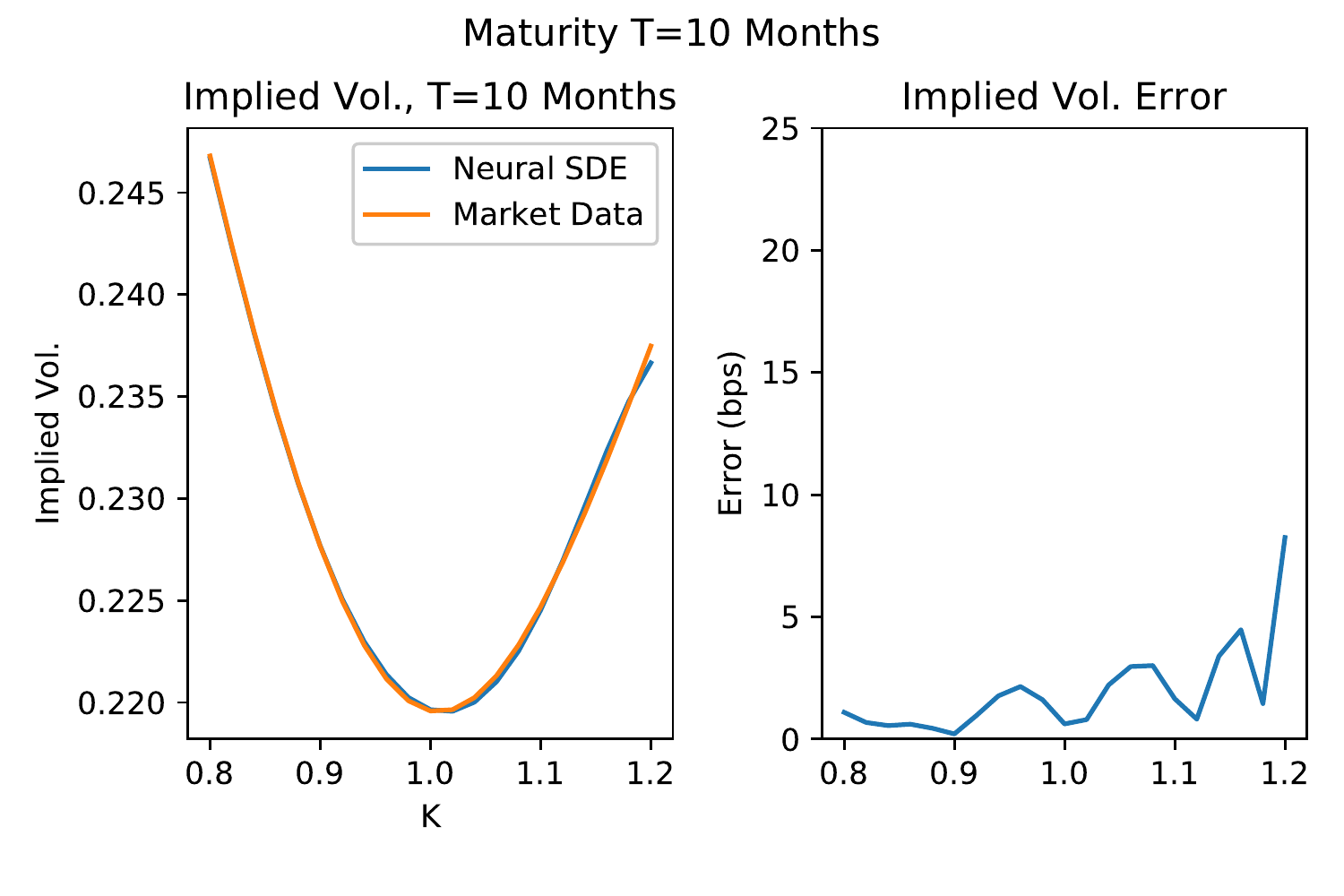} 
\includegraphics[clip,width=0.3\textwidth]{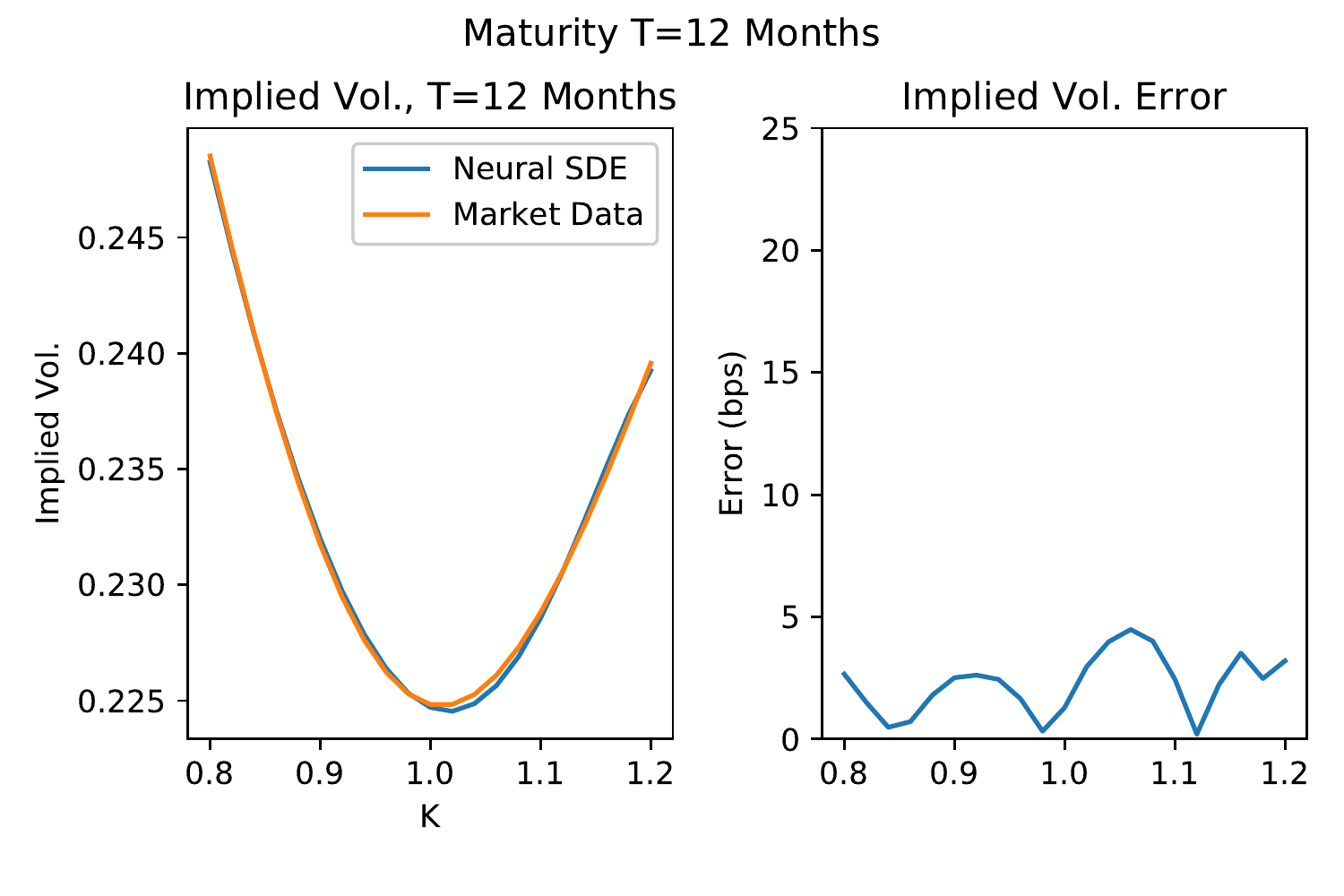}
\caption{Calibrated neural SDE LV model (with lower bound minimization on exotic payoff) and target market data implied volatility comparison.}
\label{FigLB}  
\end{figure}

\begin{figure}[h]\label{FigLBprice}
  \centering 
  \includegraphics[clip,width=0.3\textwidth]{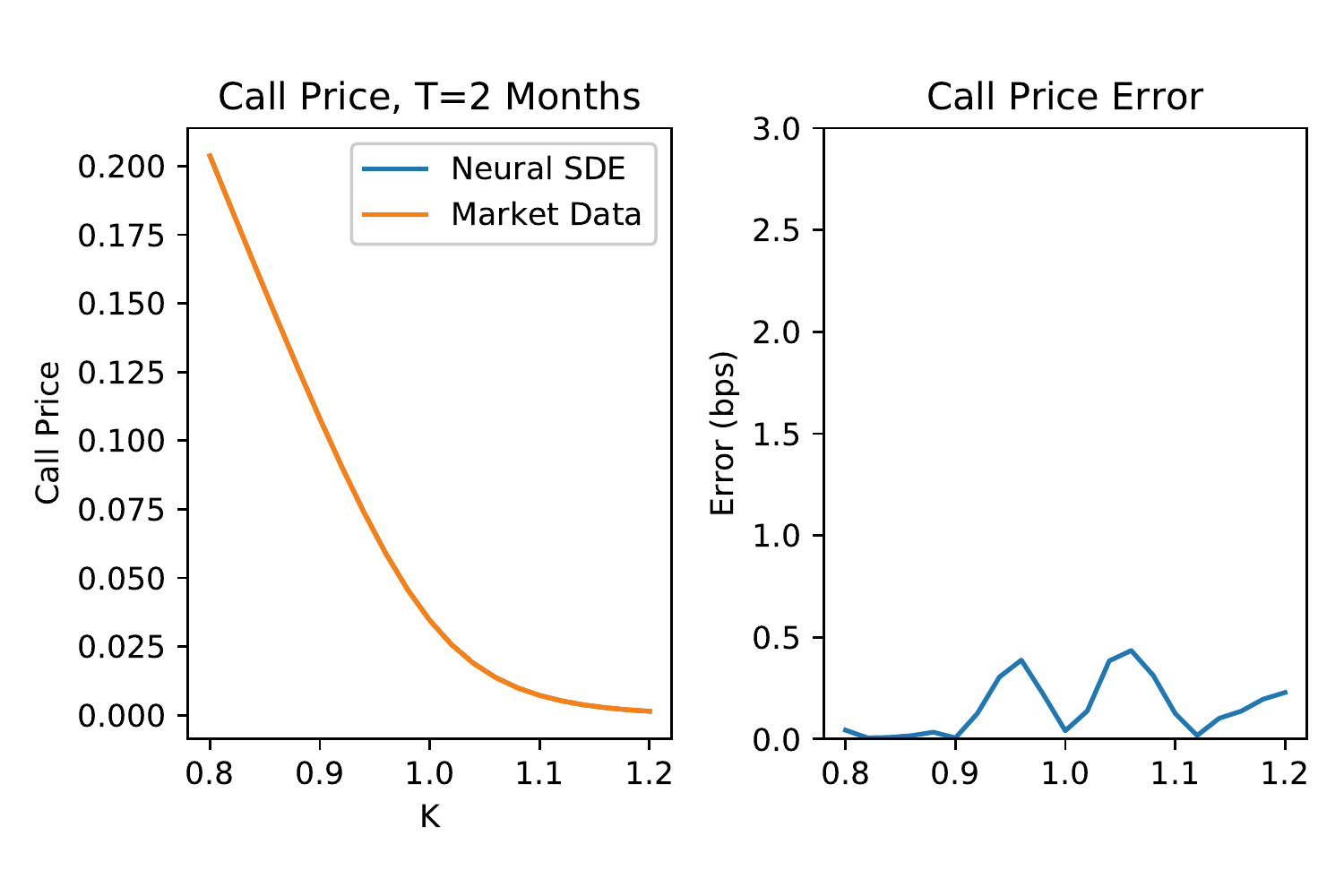}
  \includegraphics[clip,width=0.3\textwidth]{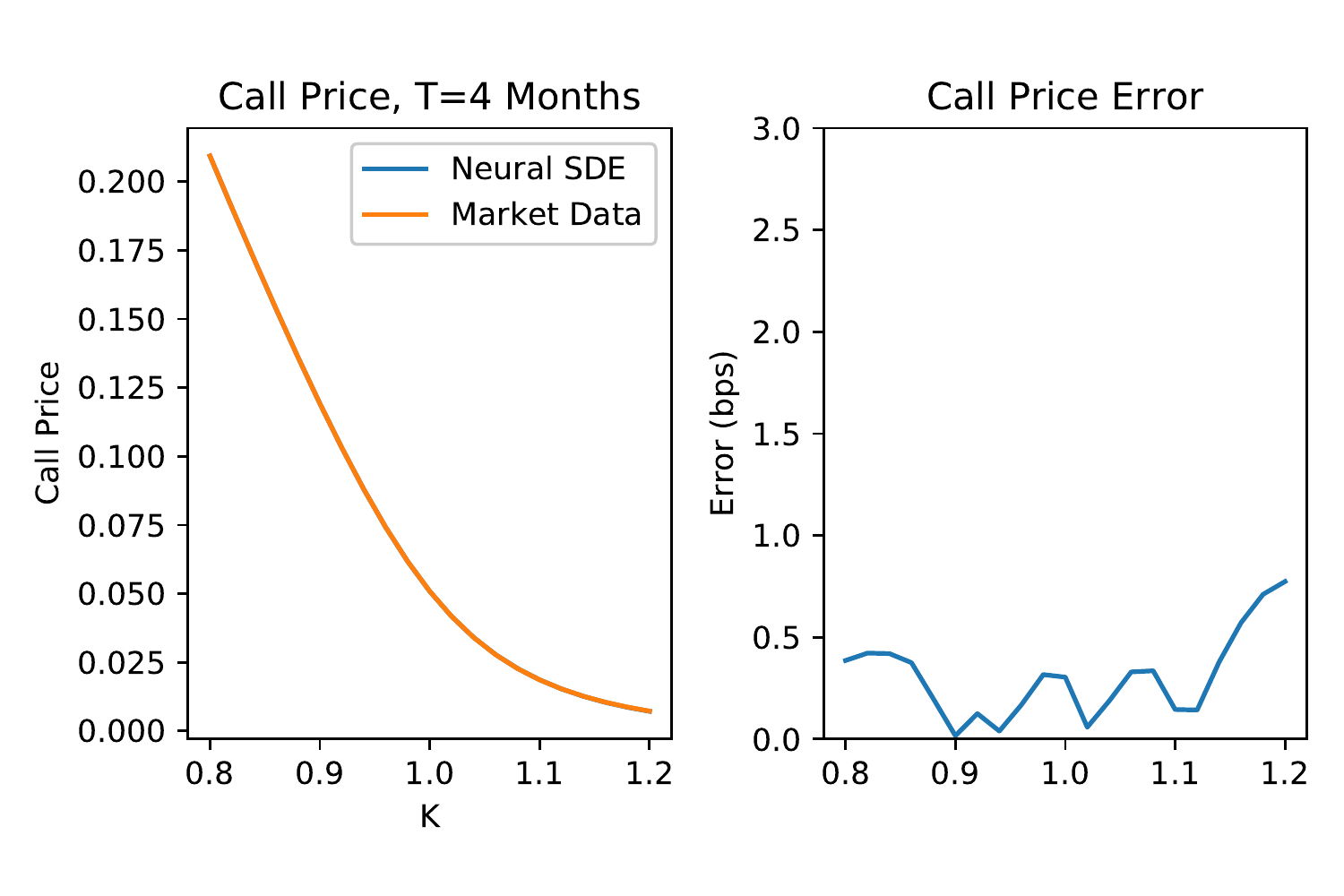}
  \includegraphics[clip,width=0.3\textwidth]{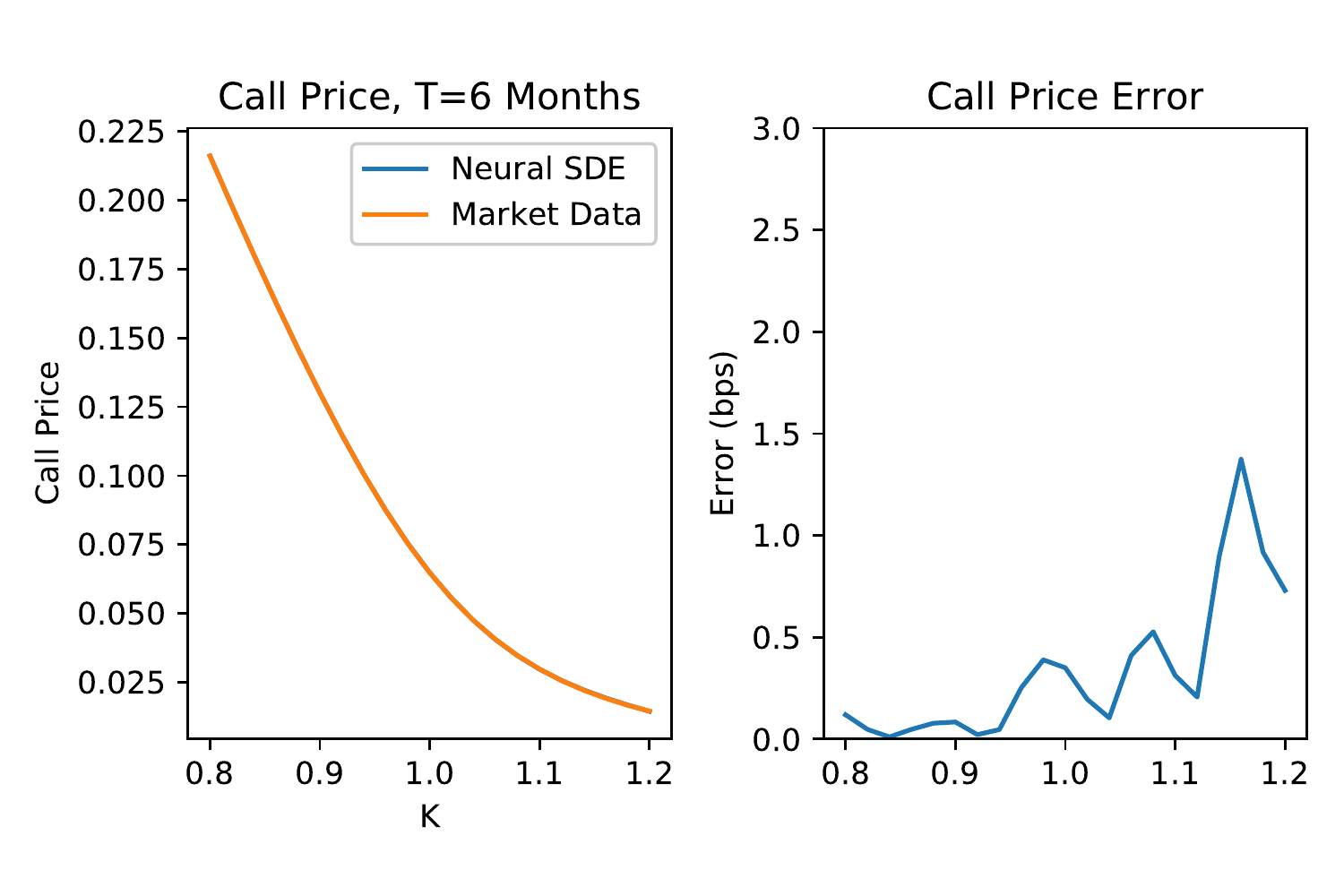} 
  \includegraphics[clip,width=0.3\textwidth]{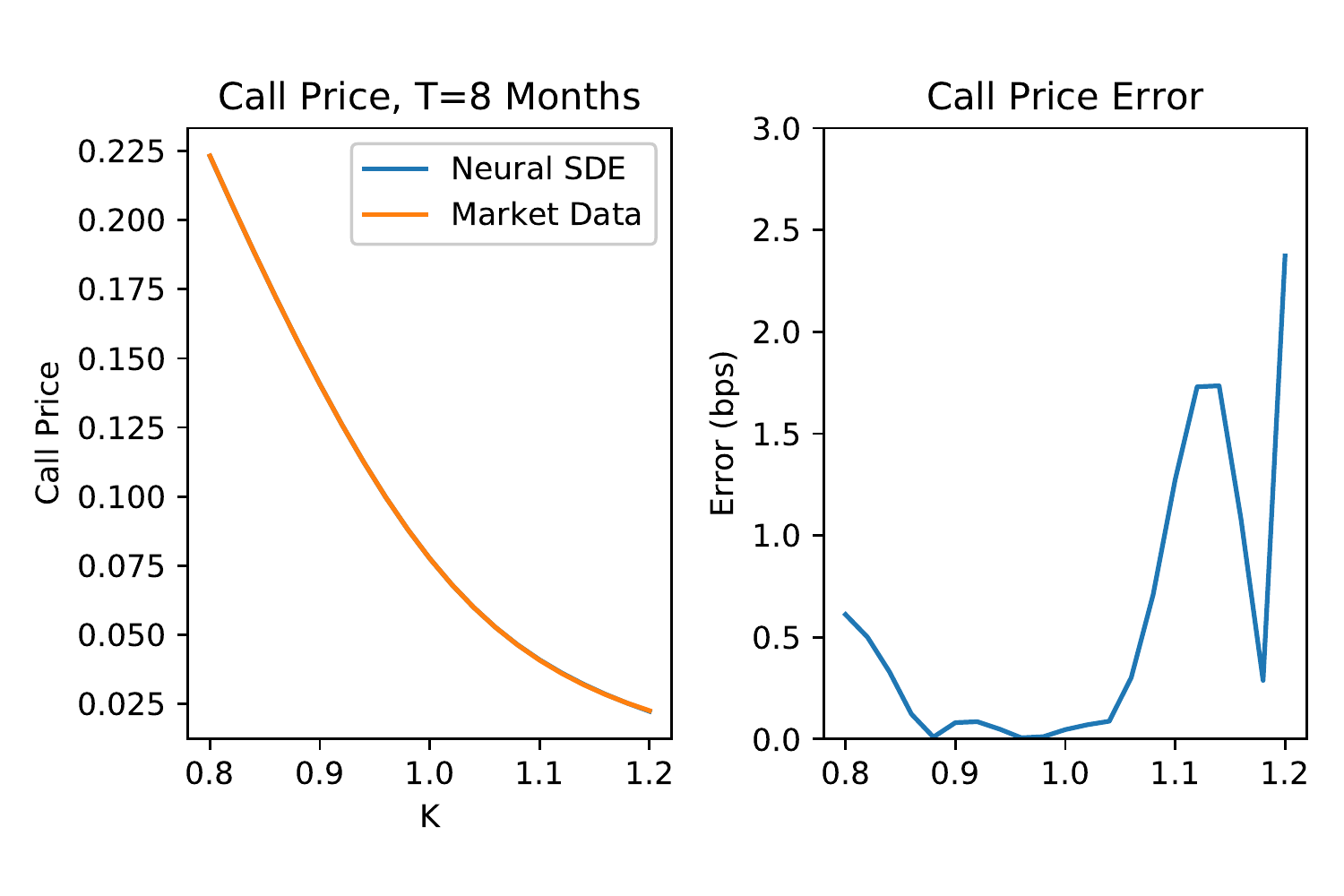}
  \includegraphics[clip,width=0.3\textwidth]{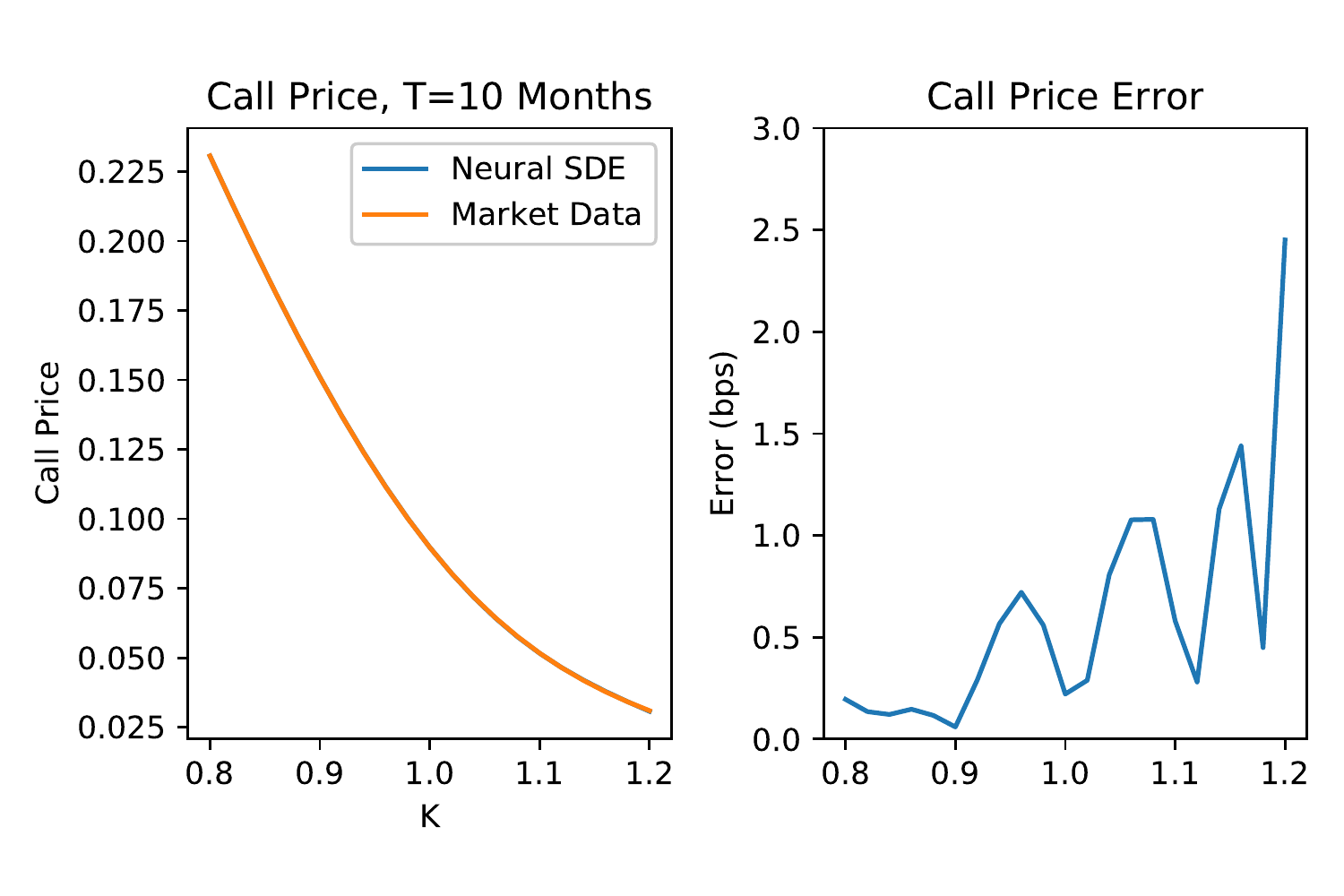} 
  \includegraphics[clip,width=0.3\textwidth]{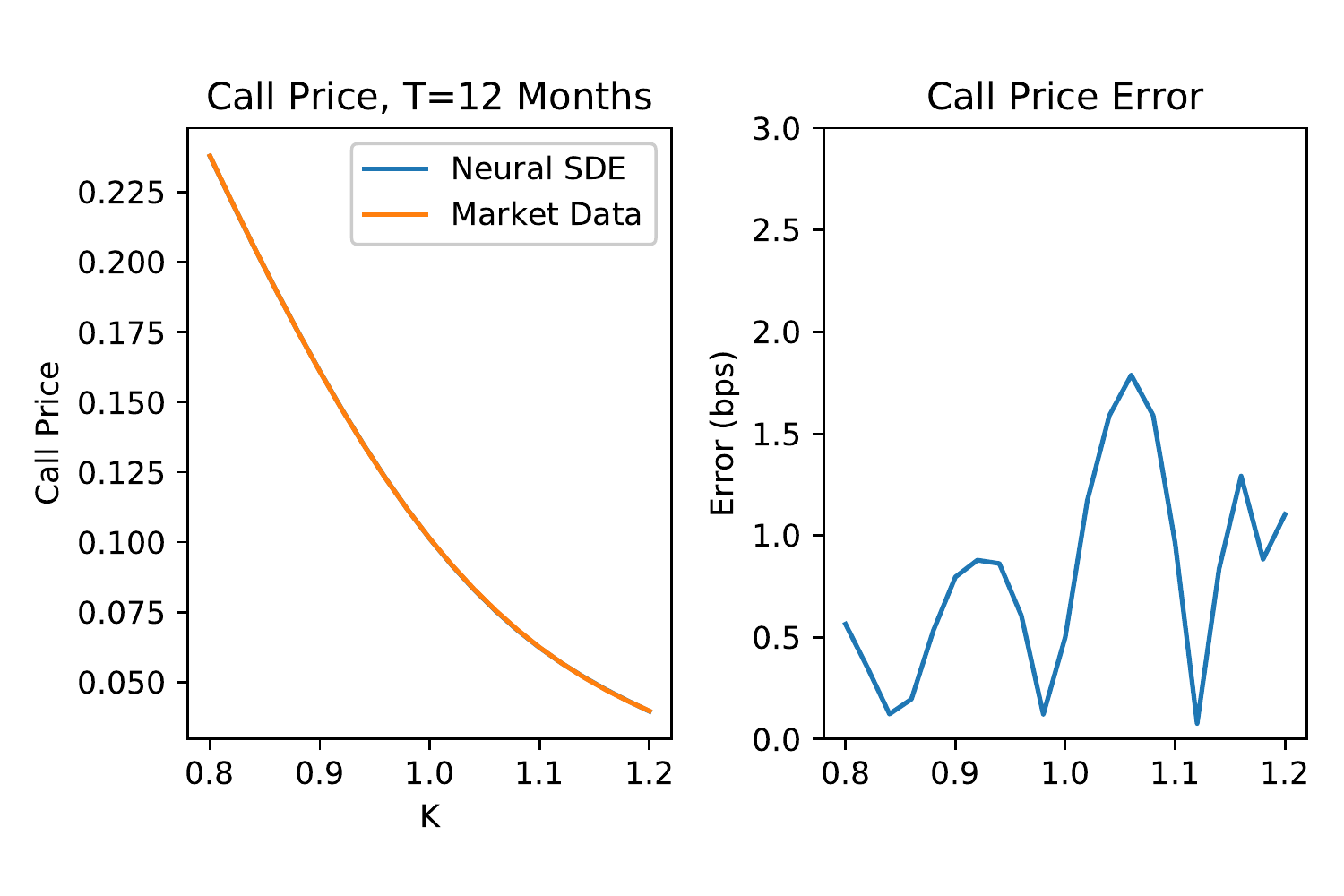}
  \caption{Calibrated neural SDE LV model (with lower bound minimization on exotic payoff) and target market prices comparison.}
  \label{FigLBprice}
\end{figure}

\begin{figure}[h]\label{FigUB}
\centering 
\includegraphics[clip,width=0.3\textwidth]{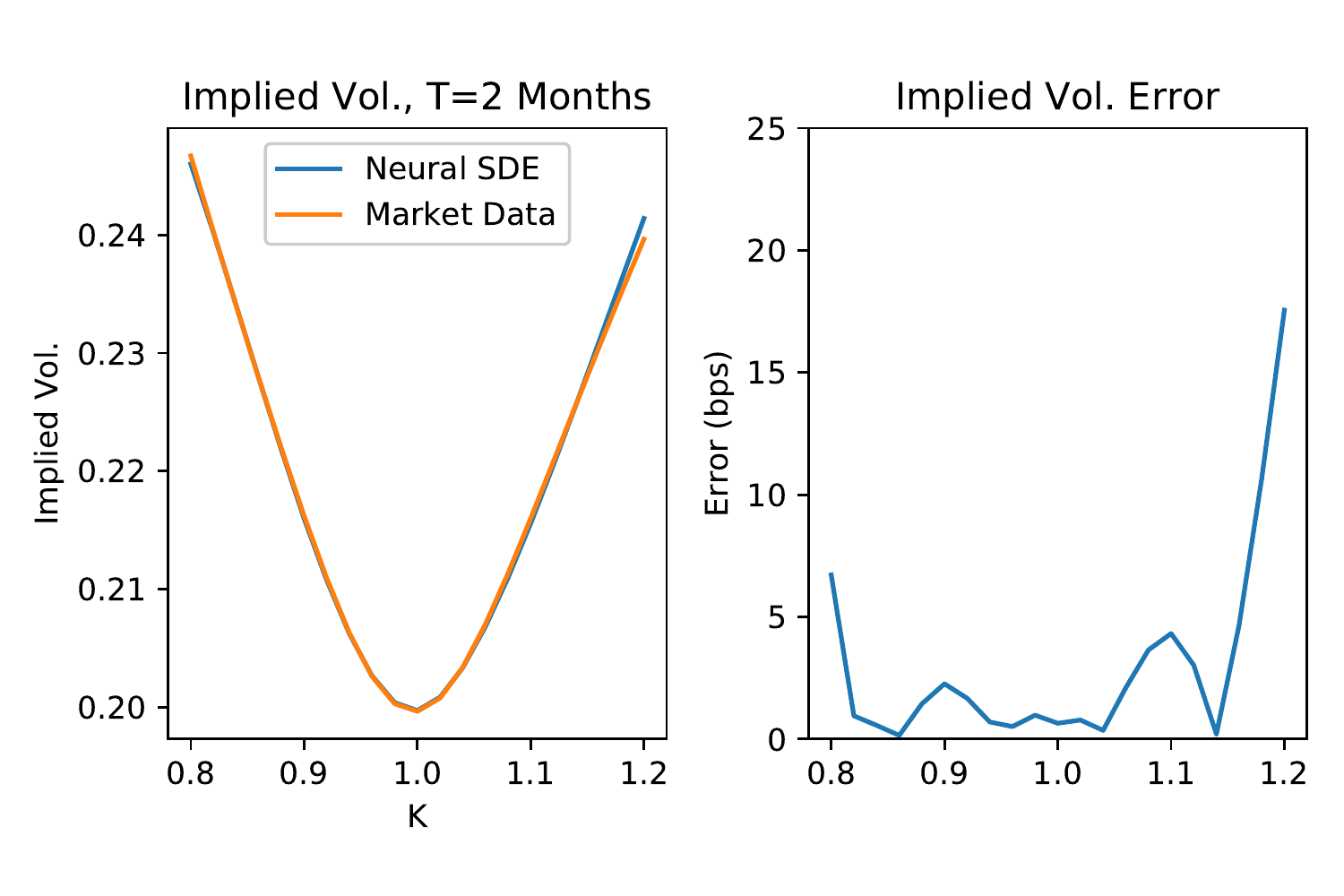}
  \includegraphics[clip,width=0.3\textwidth]{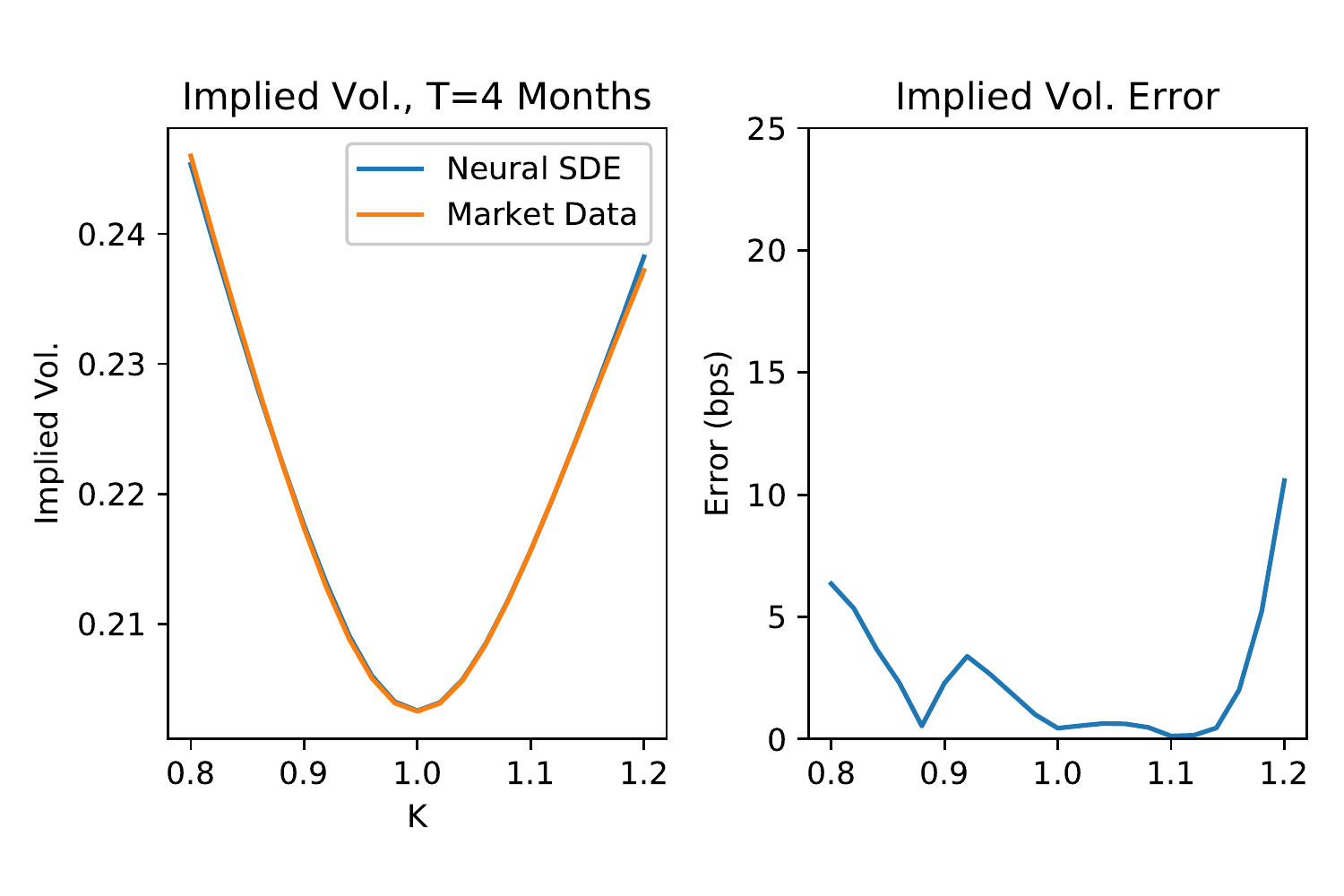}
  \includegraphics[clip,width=0.3\textwidth]{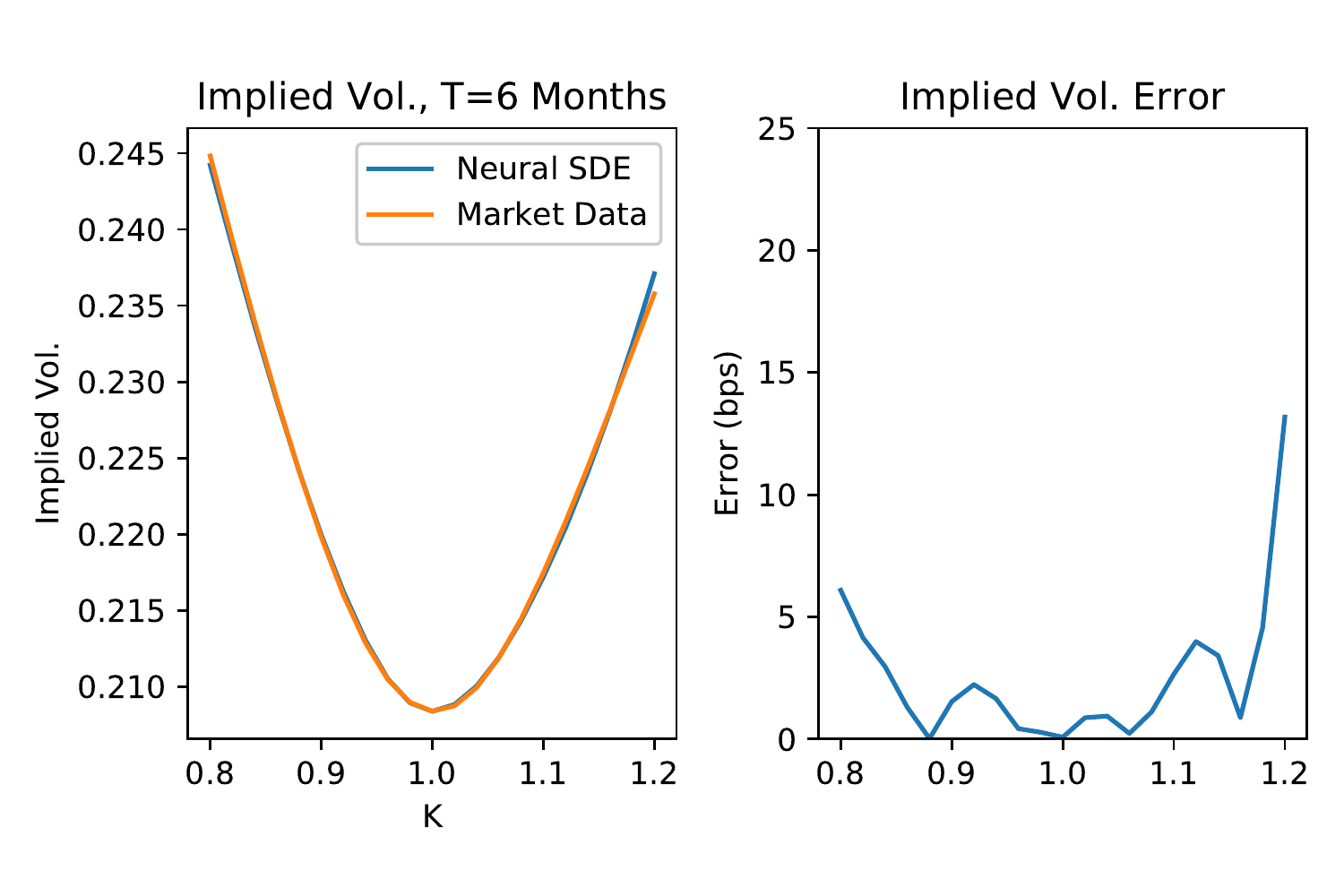} 
  \includegraphics[clip,width=0.3\textwidth]{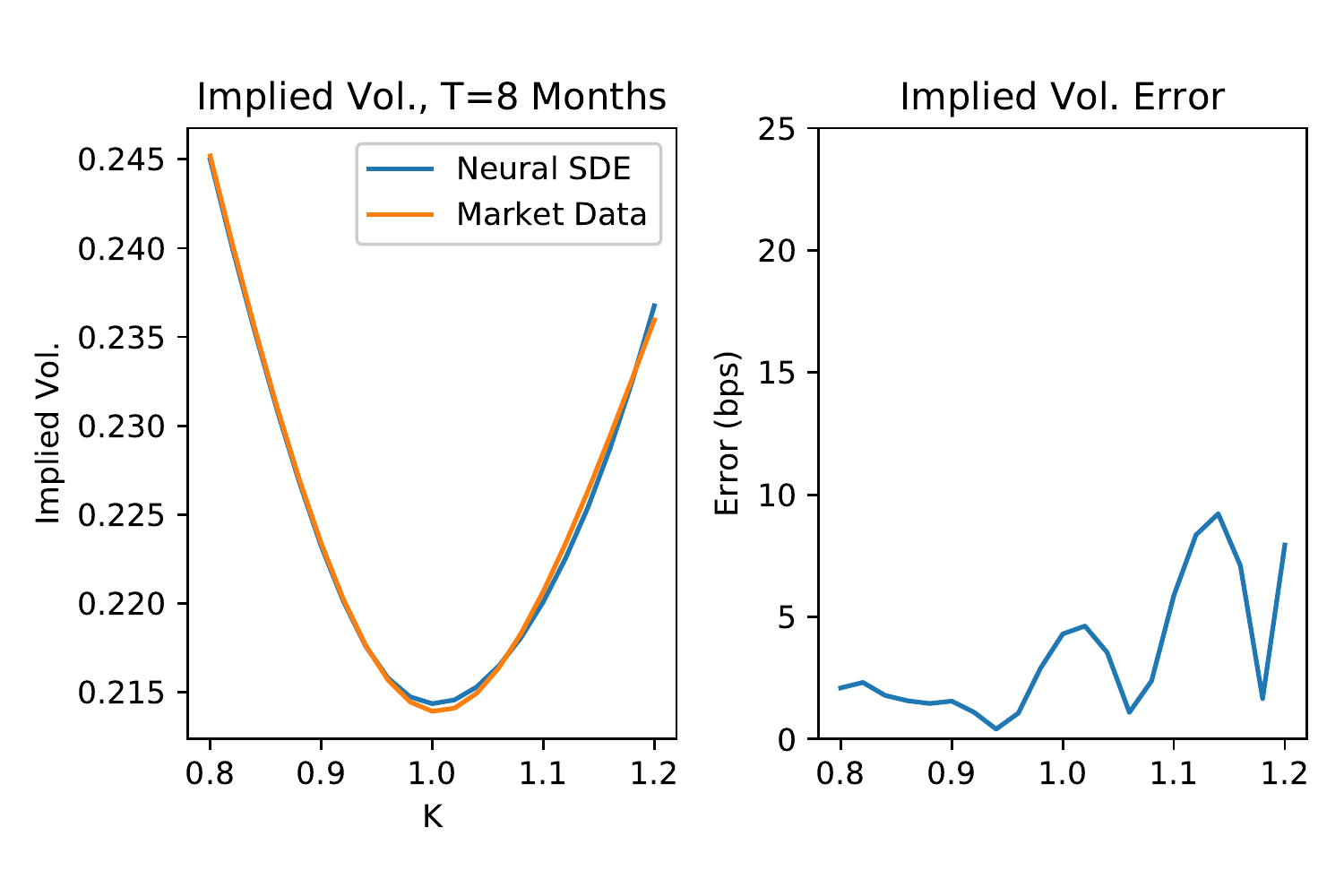}
  \includegraphics[clip,width=0.3\textwidth]{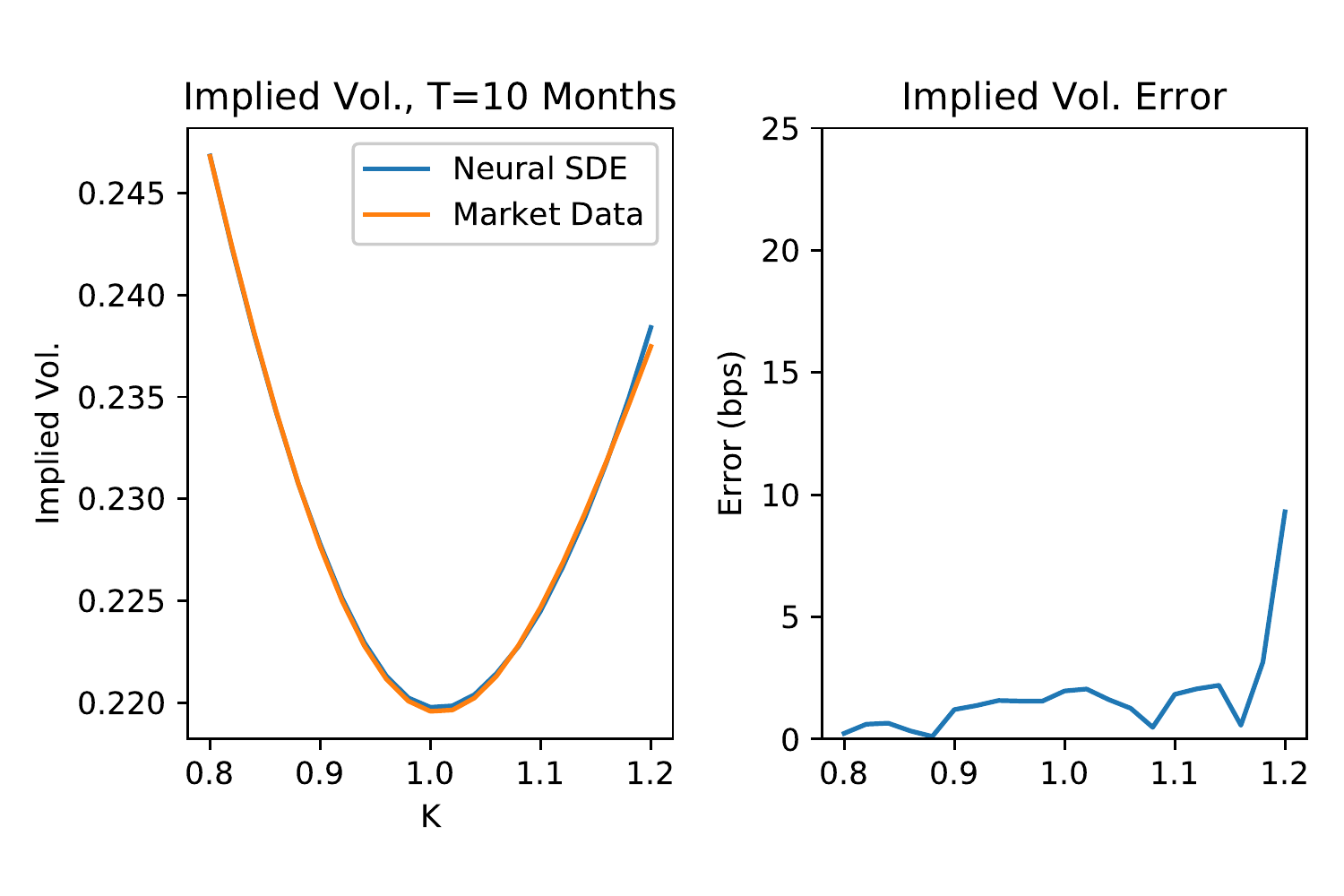} 
  \includegraphics[clip,width=0.3\textwidth]{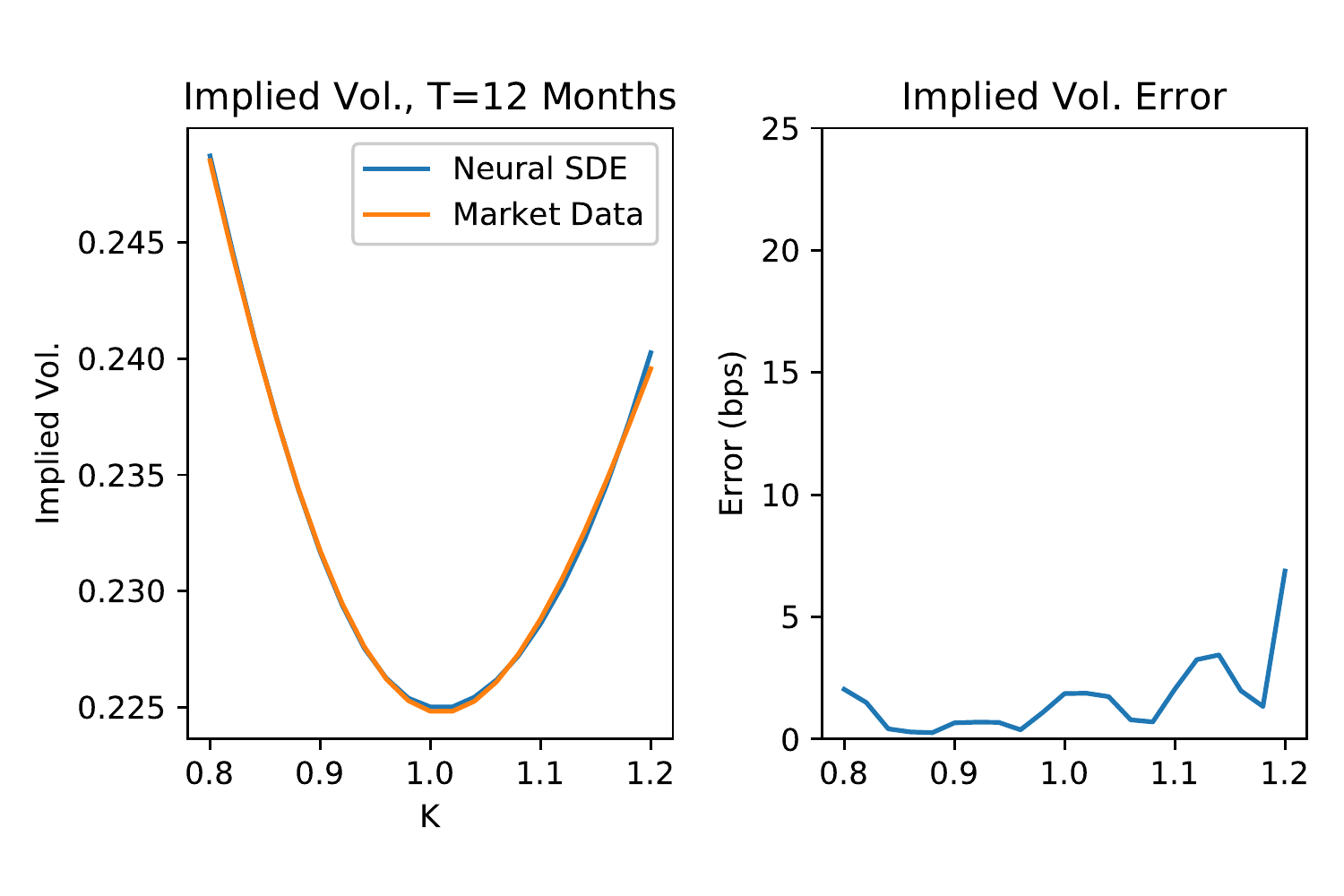}
  \caption{Calibrated neural LV model (with upper bound maximization on exotic payoff) and target market data implied volatility comparison.}
 \label{FigUB}
\end{figure}

\begin{figure}[h]
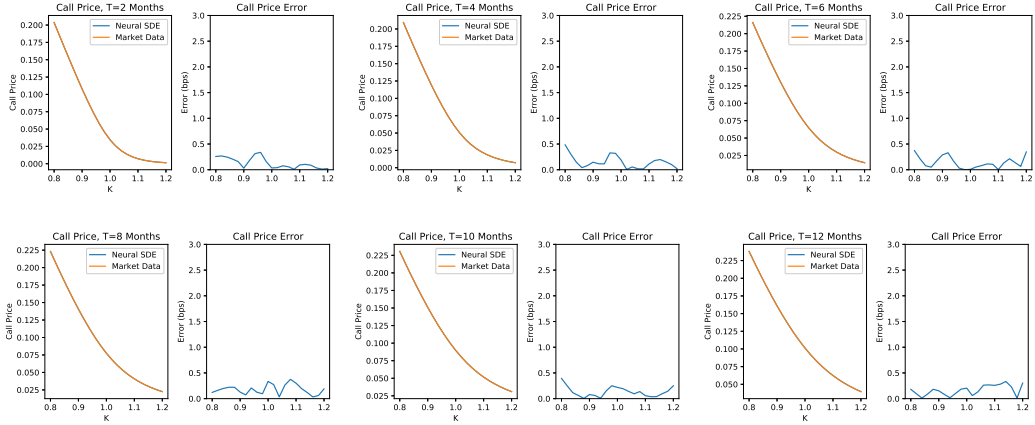

\centering 
\includegraphics[clip,width=0.3\textwidth]{figures_LV/Neural_SDE_maturity16_vanilla}
  \includegraphics[clip,width=0.3\textwidth]{figures_LV/Neural_SDE_maturity32_vanilla}
  \includegraphics[clip,width=0.3\textwidth]{figures_LV/Neural_SDE_maturity48_vanilla} 
  \includegraphics[clip,width=0.3\textwidth]{figures_LV/Neural_SDE_maturity64_vanilla}
  \includegraphics[clip,width=0.3\textwidth]{figures_LV/Neural_SDE_maturity80_vanilla} 
  \includegraphics[clip,width=0.3\textwidth]{figures_LV/Neural_SDE_maturity96_vanilla}
\caption{Calibrated LV neural model (with upper bound maximization on exotic payoff) and target market prices comparison.}
\label{FigUBprice}
\end{figure}

\section{LSV neural SDEs calibration accuracy}
\label{sec LSVfitquality}

Figures~\ref{fig LSV calibration lb}, \ref{fig LSV calibration} and~\ref{fig LSV calibration ub} provide the Vanilla call option price and the implied volatility curve
for the calibrated models. 
In each plot, the blue line corresponds to the target data (generated using Heston model), and each orange line corresponds to one run of the 
Neural SDE calibration. We note again in this plots how the absolute error of the calibration to the vanilla prices is consistently of $\mathcal O(10^{-4})$.

\begin{figure}[h]
  \centering 
  \includegraphics[clip,width=0.3\textwidth]{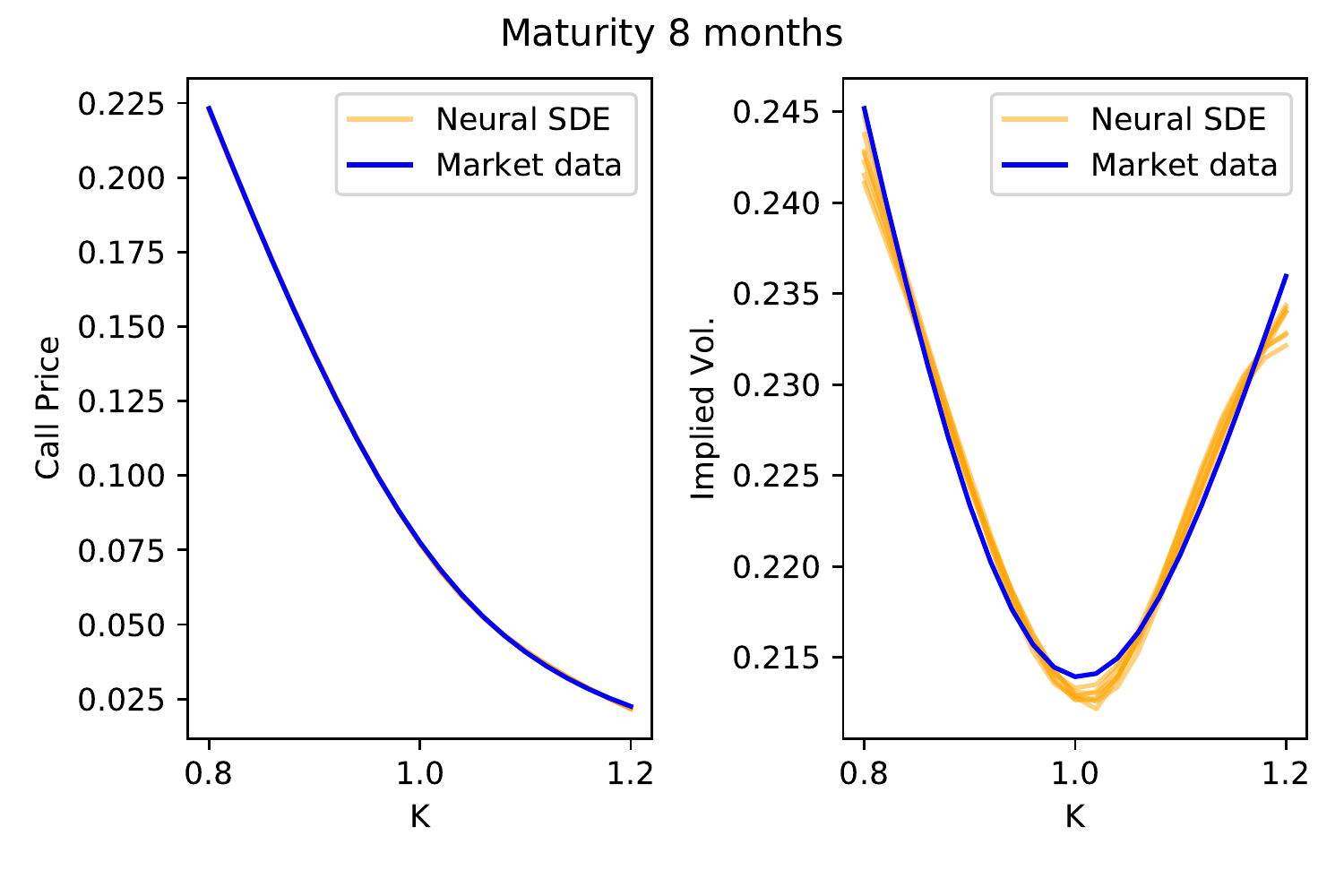}
  \includegraphics[clip,width=0.3\textwidth]{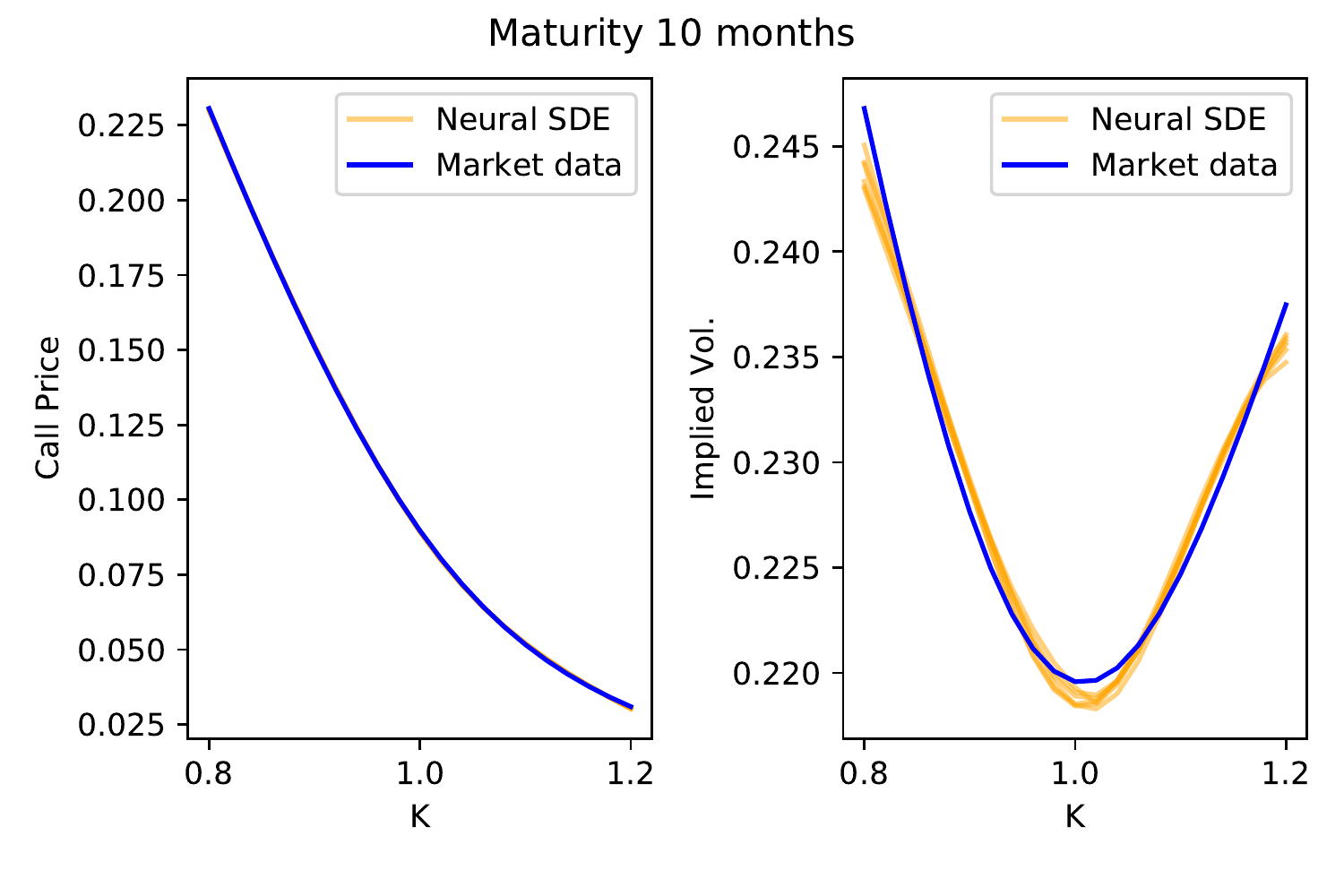}
  \includegraphics[clip,width=0.3\textwidth]{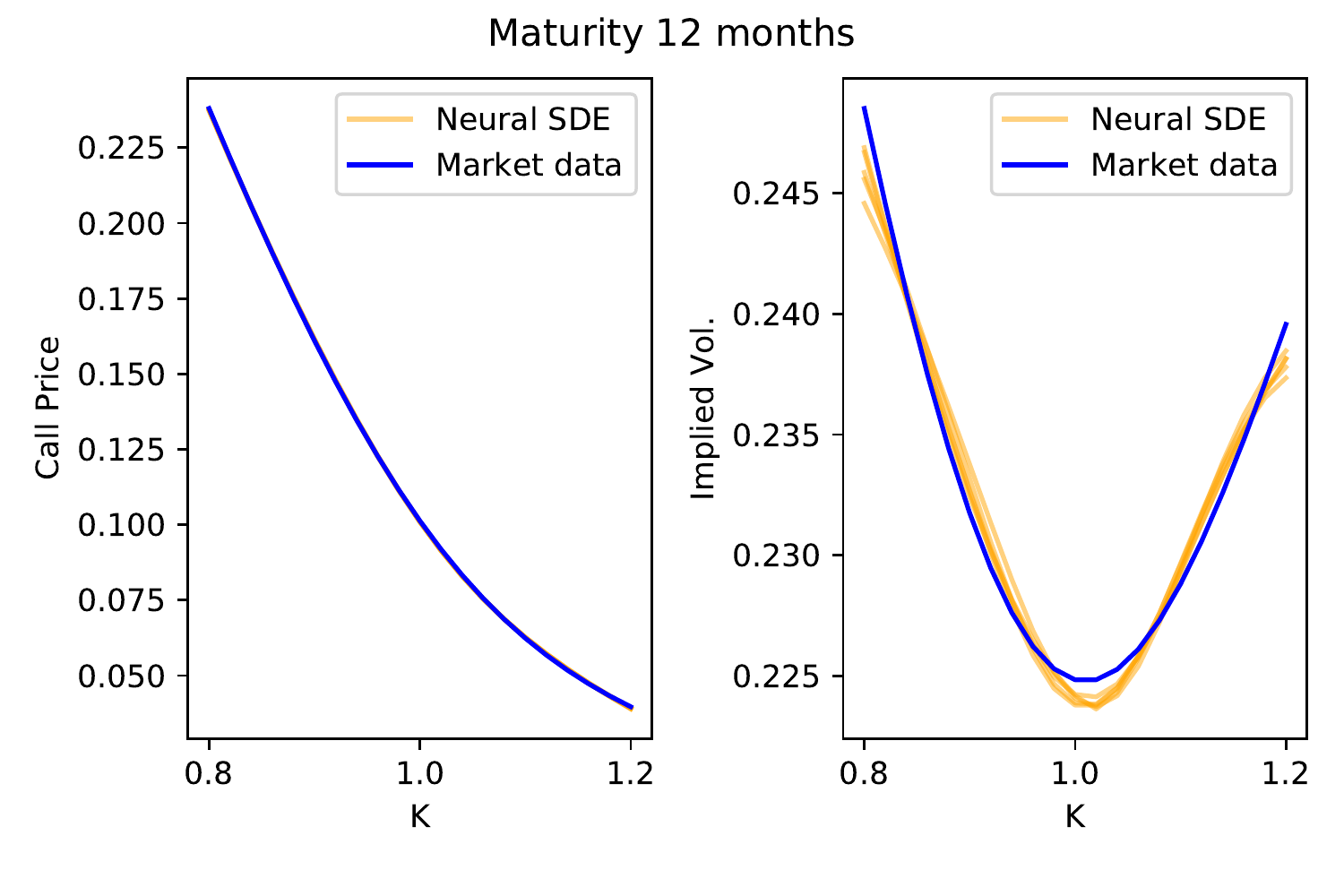} 
  \caption{Comparing market and model data fit for the Neural SDE LSV model~\eqref{eq LSV SDE} when targeting the {\em lower bound} on the illiquid derivative. 
We see vanilla option prices and implied volatility curves of the 10 calibrated Neural SDEs vs. the market data for different maturities.
}
\label{fig LSV calibration lb}  
\end{figure}

\begin{figure}[h]
  \centering 
  \includegraphics[clip,width=0.3\textwidth]{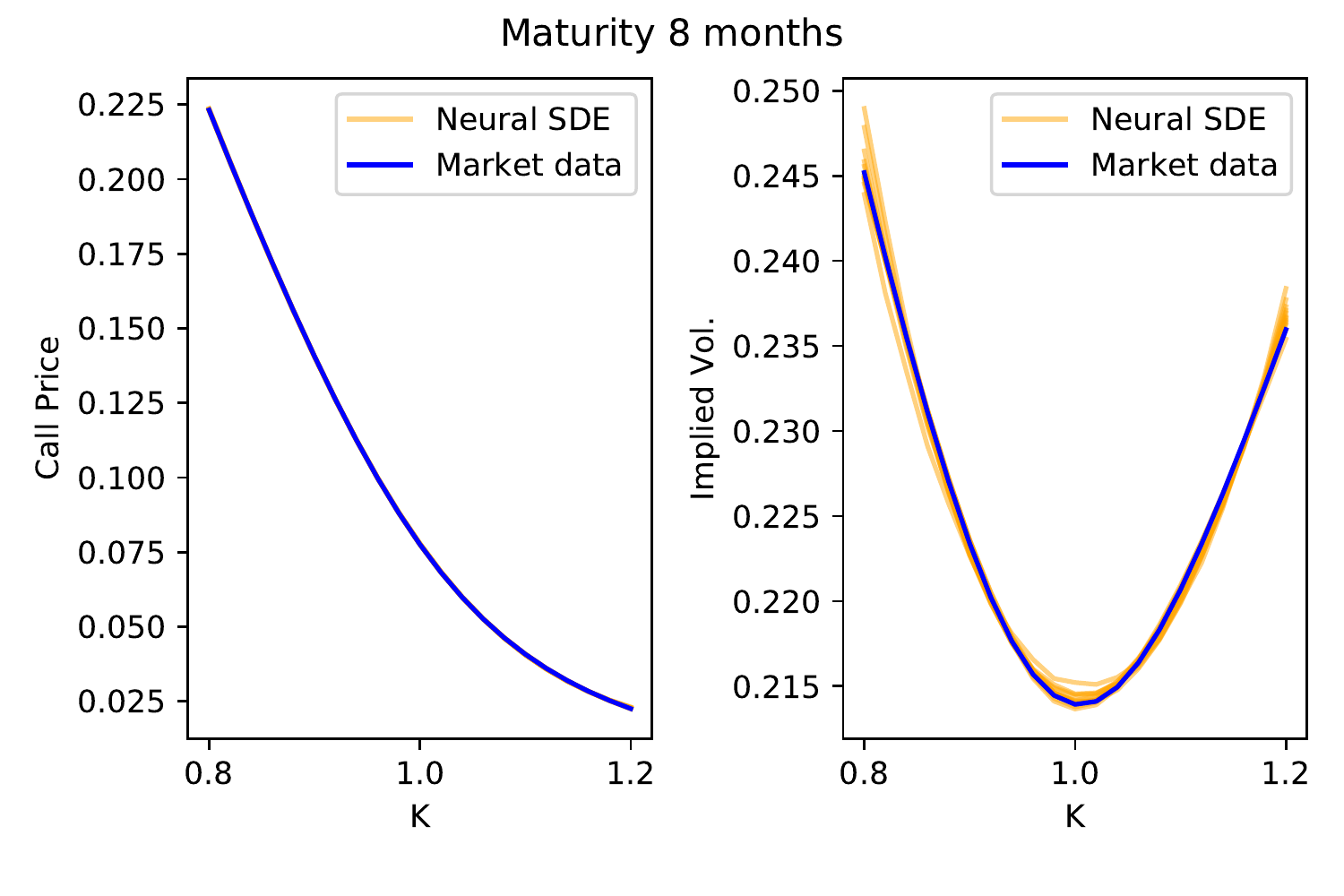}
  \includegraphics[clip,width=0.3\textwidth]{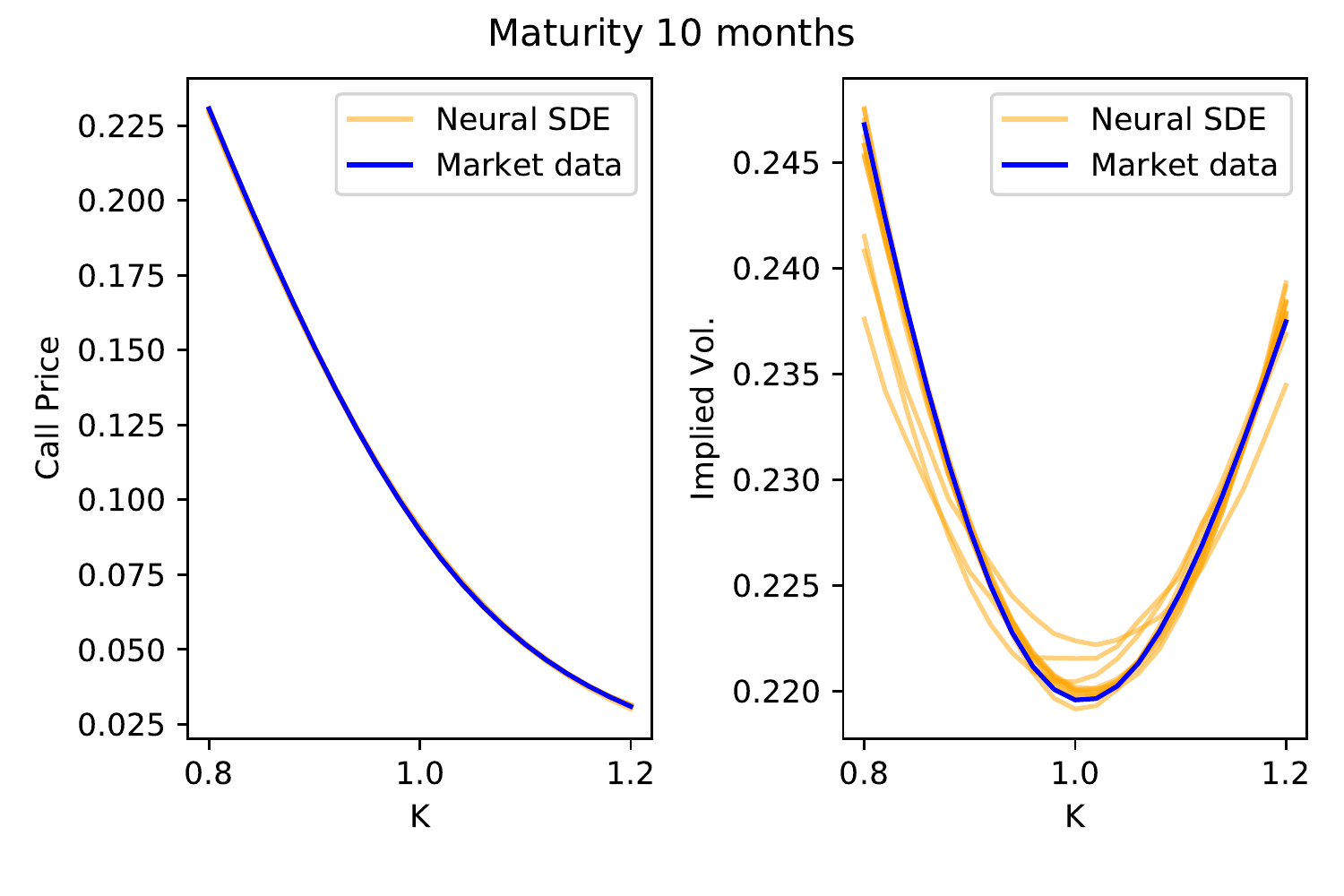}
  \includegraphics[clip,width=0.3\textwidth]{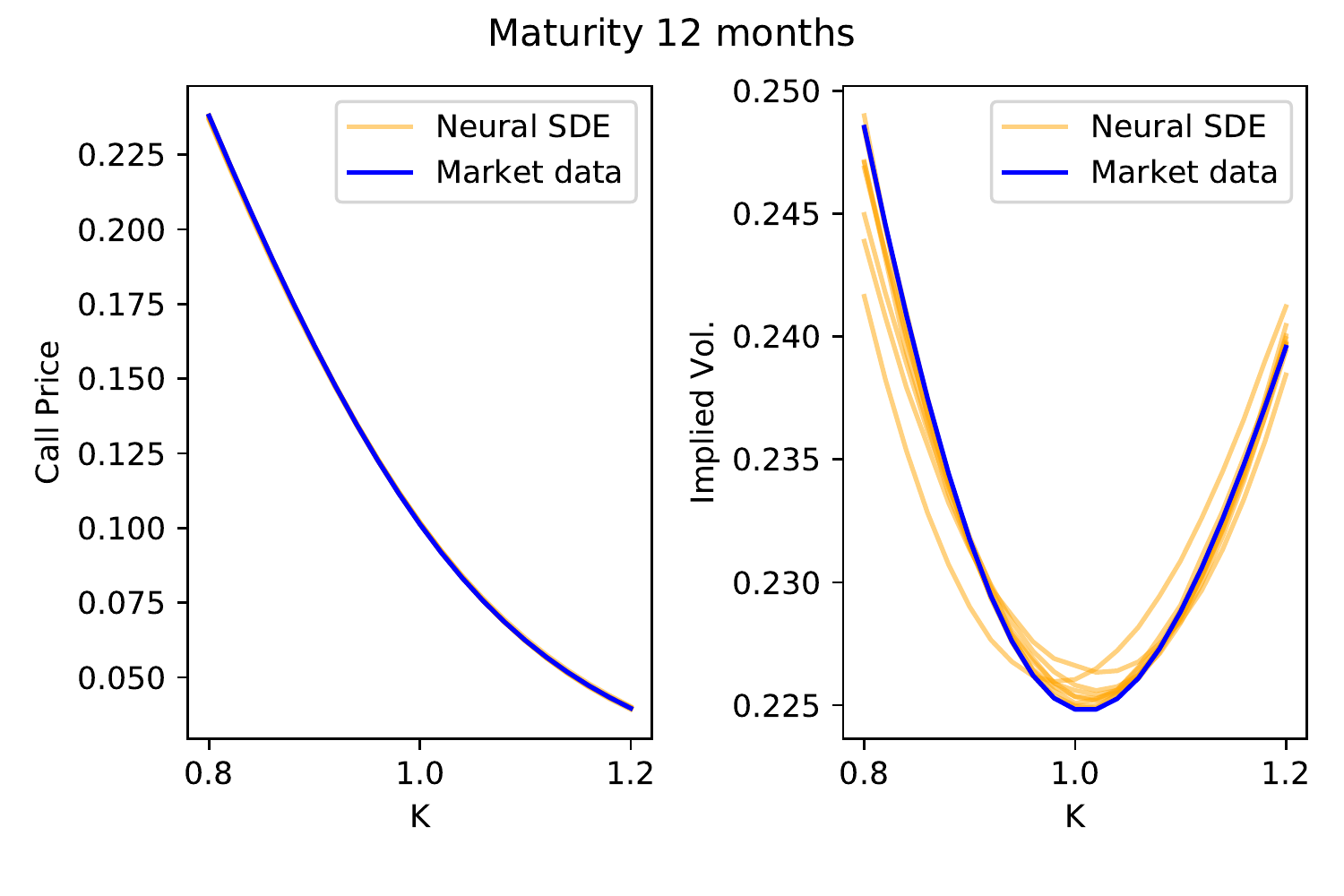} 
  \caption{Comparing market and model data fit for the Neural SDE LSV model~\eqref{eq LSV SDE} when targeting the {\em upper bound} on the illiquid derivative. 
We see vanilla option prices and implied volatility curves of the 10 calibrated Neural SDEs vs. the market data for different maturities.
}
\label{fig LSV calibration ub}
\end{figure}

\section{Exotic price in LV neural SDEs}\label{LVtables}

Below we see how different random seeds, constrained optimization algorithms and number of strikes used in the market data input affect the illiquid derivative price in the Local Volatility Neural SDE model.

\begin{figure}[h]
  \centering 
  \includegraphics[clip,width=0.45\textwidth]{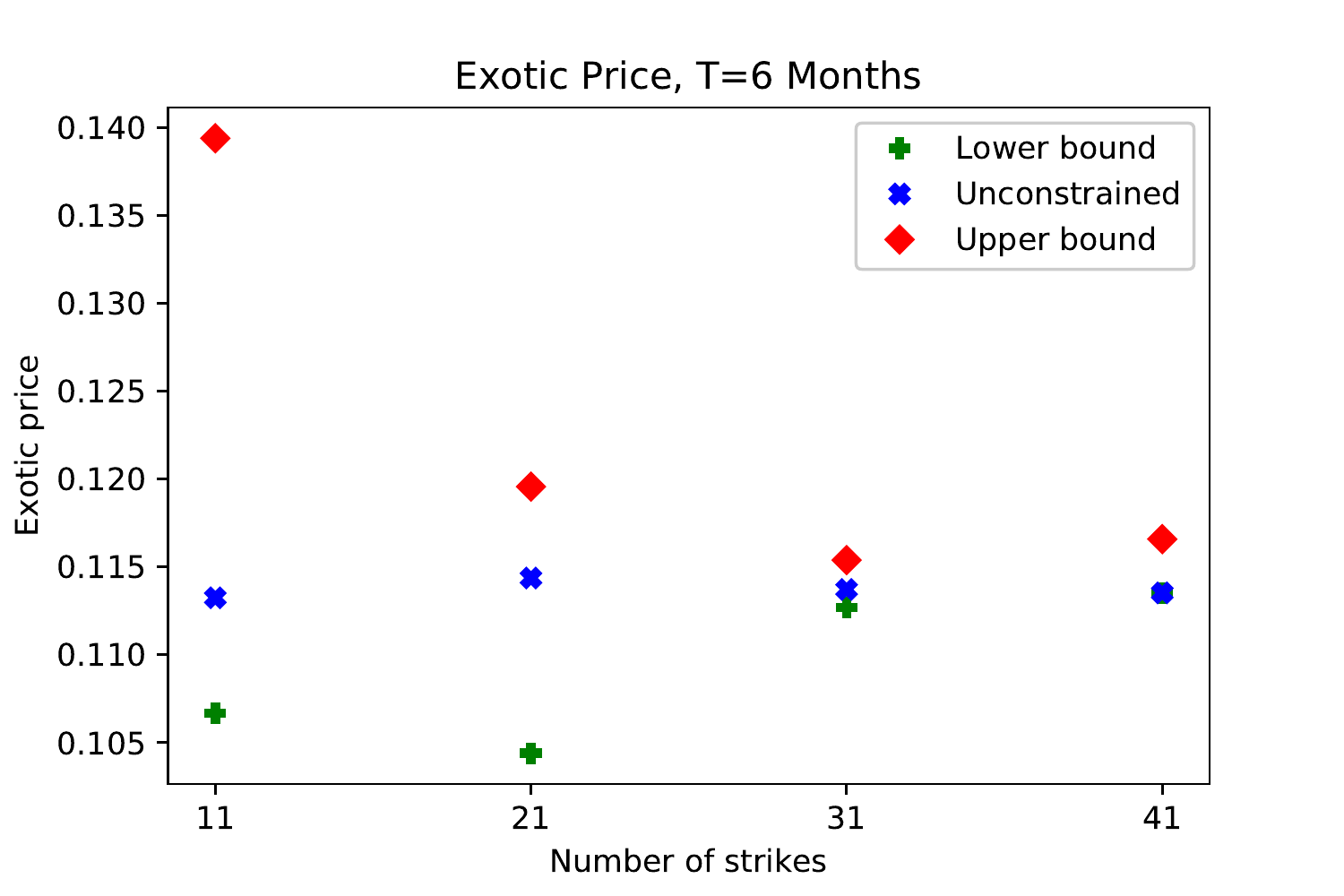}
  \includegraphics[clip,width=0.45\textwidth]{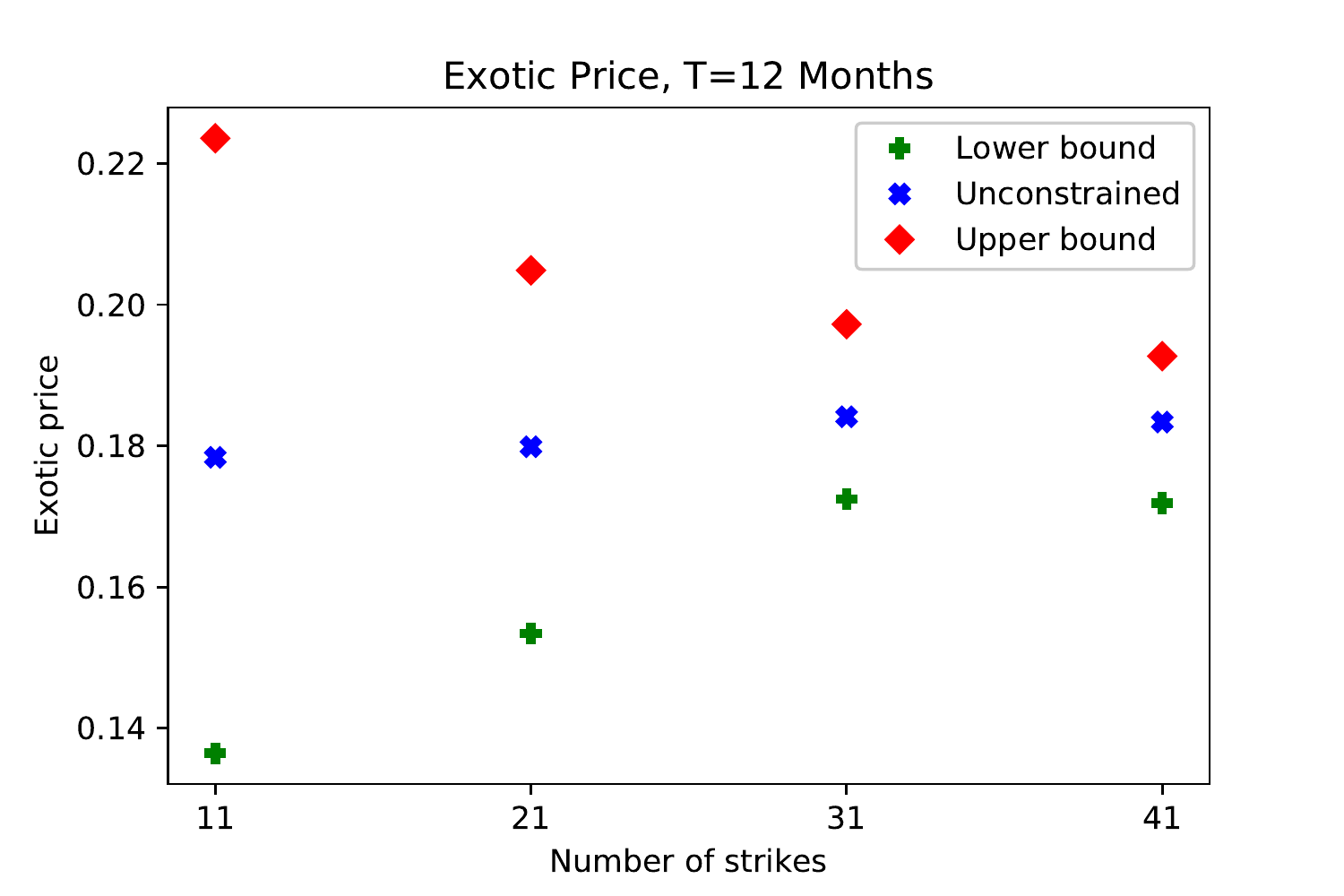}
  \caption{Lookback exotic option price in lower, upper and unconstrained implied by perfectly calibrated LV neural SDE calibrated to varying number of market option quotes.}
\label{Strikechart} 
\end{figure}

\begin{table}[h]
\begin{tabular}{|c|c||c|c|c|c|c|c|c|} 
\hline 
Initialisation & Calibration type & t=2/12 & t=4/12 & t=6/12 & t=8/12 & t=10/12 & t=1 \\ 
\hline \hline
1 & Unconstrained & .055 & .088 & .113 & .134 &  .159 &  .178  \\
\hline 
2 & Unconstrained & .056 &  .086 & .113 & .132 & .158 & .175  \\ 
\hline 
1 & LB Lag. mult. & .055 & .086 & .107 & .127 & .143 & .154 \\ 
\hline
2 & LB Lag. mult.& .055 & .084 & .098 & .113 & .125 & .139 \\ 
\hline
1 & UB Lag. mult.& .056 & .099 & .119  & .142 & .163 & .214 \\ 
\hline
2 & UB Lag. mult. & .059 & .101 & .131  & .156 & .208 & .220\\
\hline
1 & LB Augmented & .055 & .077 & .107 & .113 & .127 & .136 \\ 
\hline
2 & LB Augmented & .056 & .085 & .109 & .123 & .139 & .158 \\ 
\hline
1 & UB Augmented & .058 & .102 & .139 & .156 & .188 & .224 \\ 
\hline
2 & UB Augmented & .057 & .088 & .128  & .151 & .167 & .184 \\ 
\hline
- & Heston 400k paths & .058 & .087 & .111 & .133 & .154 & .174\\ 
%\hline
%- & Heston 10 mil paths & .058 & .087 & .111 & .133 & .154 & .174 \\ 
%\hline
%- & Local vol 400k paths & . & . & . & . & . & . \\ 
%\hline
%- & Local vol 10mil paths  & . & . & . & . & . & . \\ 
\hline
\end{tabular}
\caption{Impact of initialisation and constrained optimization algorithms on prices of an illiquid derivative (lookback call) implied by LV neural SDE calibrated to vanilla prices with $K=11$ strikes: $k_1=0.9, k_2=0.92,...,k_{11}=1.1$ for each maturity.}
\label{InitialisationImpact11strikes} 
\end{table}  

\begin{table}[h!tbp]
\begin{tabular}{|c|c||c|c|c|c|c|c|c|} 
\hline 
Initialisation & Calibration type & t=2/12 & t=4/12 & t=6/12 & t=8/12 & t=10/12 & t=1 \\ 
\hline \hline
1 & Unconstrained & .056 & .087 & .114 & .140 & .161 & .182  \\
\hline 
2 & Unconstrained & .056 & .087 & .114 & .136 & .161 & .180  \\ 
\hline 
1 & LB Lag. mult. & .056 & .086 & .110 & .123 & .141 & .153 \\ 
\hline
2 & LB Lag. mult.& .056 & .087 & .108 & .125 &  .150 & .155 \\ 
\hline
1 & UB Lag. mult.& .056 & .088 & .120  & .156 & .179 & .205 \\ 
\hline
2 & UB Lag. mult. & .056 & .088 & .118  & .153 & .187 & .208\\
\hline
1 & LB Augmented & .056 & .087 & .108 & .128 & .143 & .164 \\ 
\hline
2 & LB Augmented & .056 & .087 & .108 & .125 & .150 & .155 \\ 
\hline
1 & UB Augmented & .056 & .091 & .124 & .155 & .173 & .194\\ 
\hline
2 & UB Augmented & .056 & .088 & .125 & .146 & .167 &  .189\\ 
\hline
- & Heston 400k paths & .058 & .087 & .111 & .133 & .154 & .174 \\ 
\hline
- & Heston 10mil paths & .058 & .087 & .111 & .133 & .154 & .174 \\ 
%\hline
%- & Local vol 400k paths & . & . & . & . & . & . \\ 
%\hline
%- & Local vol 10mil paths  & . & . & . & . & . & . \\ 
\hline
\end{tabular}
\caption{Impact of initialisation Prices of ATM lookback call implied by LV neural SDE calibrated to vanilla prices with $K=21$ strikes: $k_1=0.8, k_2=0.82,...,k_{21}=1.2$ for each maturity.}
\label{InitialisationImpact21strikes} 
\end{table}

\begin{table}[h!tbp]
\begin{tabular}{|c|c||c|c|c|c|c|c|c|} 
\hline 
Initialisation & Calibration type & t=2/12 & t=4/12 & t=6/12 & t=8/12 & t=10/12 & t=1 \\ 
\hline \hline
1 & Unconstrained & .056 & .087 & .114 & .138 &  .162 &  .184  \\
\hline 
2 & Unconstrained & .056 &  .087 & .114 & .138 & .160 & .183  \\ 
\hline 
1 & LB Lag. mult. & .056 & .087 & .113 & .137 & .149 & .172 \\ 
\hline
2 & LB Lag. mult.& .056 & .087 & .113 & .136 & .155 & .165 \\ 
\hline
1 & UB Lag. mult.& .056 & .088 & .115  & .148 & .170 & .197 \\ 
\hline
2 & UB Lag. mult. & .056 & .087 & .114  & .144 & .170 & .198\\
\hline
1 & LB Augmented & .056 & .087 & .114 & .138 & .161 & .183 \\ 
\hline
2 & LB Augmented & .056 & .087 & .112 & .130 & .154 & .166 \\ 
\hline
1 & UB Augmented & .056 & .087 & .114 & .138 & .162 & .183 \\ 
\hline
2 & UB Augmented & .056 & .087 & .114 & .141 & .164 & .190 \\ 
\hline
- & Heston 400k paths& .058 & .087 & .111 & .133 & .154 & .174 \\ 
\hline
- & Heston 10mil paths & .058 & .087 & .111 & .133 & .154 & .174 \\ 
%\hline
%- & Local vol 400k paths & . & . & . & . & . & . \\ 
%\hline
%- & Local vol 10mil paths  & . & . & . & . & . & . \\ 
\hline
\end{tabular}
\caption{Impact of initialisation Prices of ATM lookback call implied by LV neural SDE calibrated to vanilla prices with $K=31$ strikes: $k_1=0.7, k_2=0.72,...,k_{31}=1.3$ for each maturity.}
\label{InitialisationImpact31strikes} 
\end{table}  

\begin{table}[h!tbp]
\begin{tabular}{|c|c||c|c|c|c|c|c|c|} 
\hline 
Initialisation & Calibration type & t=2/12 & t=4/12 & t=6/12 & t=8/12 & t=10/12 & t=1 \\ 
\hline \hline
1 & Unconstrained & .056 & .087 & .114 & .138 &  .160 &  .183  \\
\hline 
2 & Unconstrained & .056 &  .087 & .113 & .138 & .162 & .184  \\ 
\hline 
1 & LB Lag. mult. & .056 & .087 & .113 & .137 & .158 & .172 \\ 
\hline
2 & LB Lag. mult.& .056 & .087 & .113 & .137 & .153 & .171 \\ 
\hline
1 & UB Lag. mult.& .056 & .088 & .117  & .141 & .166 & .193 \\ 
\hline
2 & UB Lag. mult. & .056 & .087 & .116  & .140 & .166 & .192 \\
\hline
1 & LB Augmented & .056 & .087 & .113 & .136 & .153 & .172 \\ 
\hline
2 & LB Augmented & .056 & .087 & .113 & .136 & .152 & .169 \\ 
\hline
1 & UB Augmented & .056 & .087 & .114 & .138 & .160 & .182 \\ 
\hline
2 & UB Augmented & .056 & .087 & .116 & .140 & .164 & .191 \\ 
\hline
- & Heston 400k paths & .058 & .087 & .111 & .133 & .154 & .174 \\ 
\hline
- & Heston 10mil paths & .058 & .087 & .111 & .133 & .154 & .174 \\ 
%\hline
%- & Local vol 400kpaths & . & . & . & . & . & . \\ 
%\hline
%- & Local vol 10mil paths  & . & . & . & . & . & . \\ 
\hline
\end{tabular}

\caption{Impact of initialisation Prices of ATM lookback call implied by LV neural SDE calibrated to vanilla prices with $K=41$ strikes: $k_1=0.6, k_2=0.62,...,k_{41}=1.4$ for each maturity.}
\label{InitialisationImpact41strikes} 
\end{table}

\end{document}